\documentclass[final,3p,times]{elsarticle}

\usepackage{amsmath,amssymb,amsfonts,amsthm}
\usepackage{algorithm2e}
\usepackage{graphicx}
\usepackage{textcomp}
\usepackage{xcolor}
\usepackage{hyperref}
\usepackage{subcaption}
\usepackage{scalerel}
\usepackage[inline]{enumitem}
\usepackage{booktabs}
\usepackage{diagbox}
\usepackage{tikz}

\newtheoremstyle{problemstyle}          
        {}                              
        {}                              
        {\normalfont}                   
        {}                              
        {\bfseries\itshape}             
        {\normalfont\bfseries.}         
        {.5em}                          
        {}                              
\theoremstyle{problemstyle}

\newtheorem*{problem*}{Problem Statement}
\newtheorem{theorem}{Theorem}[section]
\newtheorem{lemma}[theorem]{Lemma}

\begin{document}

\newcommand{\staterep}[2]{\scaleobj{0.7}{\begin{pmatrix} #1 \\  #2 \end{pmatrix}}}
\renewcommand\qedsymbol{$\blacksquare$}
\SetKwComment{Comment}{/* }{ */}
\RestyleAlgo{ruled}

\newcommand\gil[1]{\textcolor{red}{GZ: #1}} 

\newcommand*\circled[1]{\tikz[baseline=(char.base)]{%
            \node[shape=circle,draw,inner sep=2pt] (char) {#1};}}

\begin{frontmatter}




\title{Performance Evaluation of Emerging Networks of Quantum Repeaters: \\
Analysis and Simulation}





\author{Kobi Ravid}

\affiliation{
  addressline={kobi.ravid@columbia.edu},
  organization={Columbia University, Department of Electrical Engineering},
  city={New York},
  state={NY},
  country={USA}
}

\author{Javad Ghaderi}

\affiliation{
  addressline={jg3465@columbia.edu},
  organization={Columbia University, Department of Electrical Engineering},
  city={New York},
  state={NY},
  country={USA}
}

\author{Gil Zussman}

\affiliation{
  addressline={gil.zussman@columbia.edu},
  organization={Columbia University, Department of Electrical Engineering},
  city={New York},
  state={NY},
  country={USA}
}


\begin{abstract}
Emerging experimental quantum network deployments are based on collaborations between research institutions and rely on telecommunications fiber infrastructure.
Currently, such networks consist of mostly Dual Input Repeaters (DIRs) and can be viewed as composed of (mostly) line graphs. In order to evaluate the performance of such networks (e.g., in terms of Entangled Pair Rate, EPR, and average Fidelity) we developed a custom quantum network simulation platform that takes into account both the quantum and classical network components.
To support the verification of the simulation results, we develop a novel approach for performance analysis of DIRs. We compare the simulation results to the analytical results for DIRs (simulation results are within $0.3\%$ error) and to prior results for repeaters with more than two inputs (simulation results are within $3\%$ error). We then observe that although the expected memory utilization cannot be bounded, \emph{a small DIR memory size is sufficient and leads to minimal performance degradation} (e.g., memory size of $3$ leads to only $3\%$ degradation in performance, for input arrival probabilities of $0.2$ and $0.8$).

We introduce a simple heuristic algorithm to maximize the EPR of practical line networks. For line networks with a single repeater and $3$ repeaters, we explore the dependency of incoming and outgoing memory size requirements on the number of nodes (network size), their connectivity (fiber length) and various quantum error settings (amplitude damping and phase flip channels). 
We evaluate the \emph{EPR} and the \emph{average Fidelity} for a sub-network of a deployed network with $3$ repeaters, using two alternative paths. 
We show that as the number of nodes increases or when the fiber lengths between nodes differ significantly, additional node memory is needed to maintain the same EPR (qubits are stored for longer duration to allow completion of entanglement across the network). For example, an equal length fiber networks with a single node or $3$ nodes require memory size of $45$ or $80$ (respectively) to maintain the same EPR of $0.79$. We also show that the memory size should be above some minimum threshold to allow EP creation (EPR $>0$) as pair selection and swapping completion propagate across the network.
For a single repeater networks, the \emph{average Fidelity} is better than $0.9$ for most simulated quantum error cases and it is more sensitive to  memory errors than to fiber errors. 
For $3$ repeater networks, as fiber lengths increase, the average Fidelity degrades for higher quantum errors (both in node memory and/or fibers) due to overall latency and memory storage time. For instance, while the channel quantum error is fixed, the Fidelity drops from $0.6 - 0.7$ (low quantum error memory settings) to $0.3 - 0.35$ (high quantum error memory settings).

\end{abstract}




\begin{keyword}



quantum network simulation, swapping, quantum repeater, entanglement rate, Discrete Time Markov Chain (DTMC), amplitude damping channel, phase flip channel, dephasing.

\end{keyword}

\end{frontmatter}

\section{Introduction}\label{sec:intro}
As quantum computing emerges, advances are made towards building distributed quantum networks~\cite{Caleffi2024,Mao2023,Oliveira2023}. Such networks can connect remote nodes while providing means for exchanging Entanglement Pairs (EPs) over long distances. Generally, such a network consists of Bell Pair (BP) generators, quantum repeaters (with quantum internal memory), and fiber interconnects~\cite{Hanson2018} or free space connectivity~\cite{Laurent2023}. 
EP swapping, purification (distillation)~\cite{Nielsen_Chuang_2010}, and routing decisions~\cite{Hussein2022} are used in order for an EP to be shared between two (remote) selected nodes~\cite{Laszlo2022}. As a result, higher layer algorithms such as teleportation and Quantum Key Distribution (QKD) can be enabled~\cite{Devulapalli2024,Cacciapuoti2020b}. 

Lattice topology, as illustrated in Fig.~\ref{sfig:arbitrary_network}, is a common network structure often used in simulations~\cite{Pant2019,Mao2023}.
On the other hand, emerging quantum network deployments are based on collaborations between research institutions and rely on deployed telecommunication fiber networks ~\cite{FigueroaRG, DCQNet,IEQNET}. For example, Fig.~\ref{fig:SCY_google_earth} shows the SCY-QNet network~\cite{FigueroaRG, Katramatos2024} connecting Stony Brook University (SBU), Brookhaven National Laboratory (BNL), and Columbia University (CBU). This long distance network mainly uses quantum Dual Input Repeaters (DIRs) and quantum memories as key components that extend the distances between potential EPs. 
Similarly, future quantum networks will take advantage of existing telecommunications (classic) fiber infrastructure. For example, the Crown Castle fiber infrastructure in Seattle, WA (Fig.~\ref{sfig:Seattle_geo}) and Orlando, FL (Fig.~\ref{sfig:Orlando_geo}) (data obtained from~\cite{Sathiya2020}), has graph representations that appear in Fig.~\ref{sfig:Seattle_graph} and Fig.~\ref{sfig:Orlando_graph}, respectively. By exploring the structure of the current and potential future practical quantum networks that will rely on telecommunications fiber, \emph{we observe that most repeaters will be DIRs and those networks can (mostly) be viewed as line graphs}. Once these networks mature, more complex network topologies such as the one illustrated in Fig.~\ref{sfig:arbitrary_network} may emerge. 
%

\begin{figure}
    \centering
    \begin{subfigure}[b]{0.3\linewidth}        
        \centering
        \includegraphics[trim = 0mm 120mm 80mm 0mm, clip,scale=0.28]{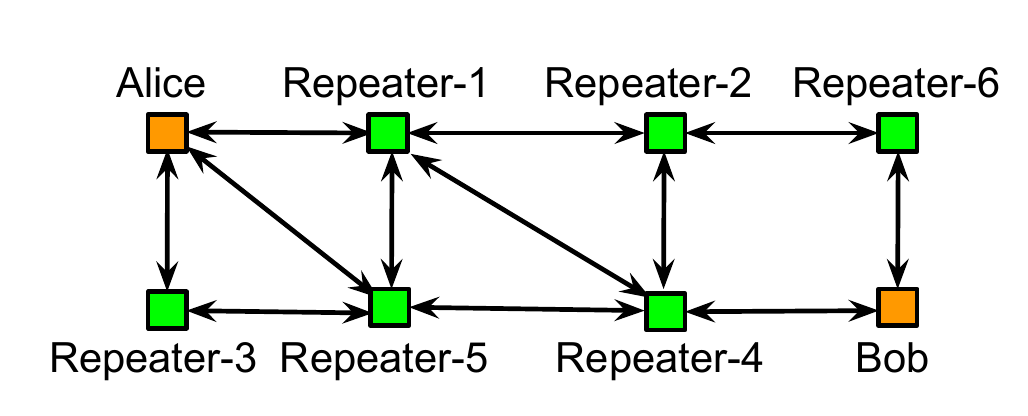}        
        \caption{Two end nodes network}
        \label{sfig:arbitrary_network}
    \end{subfigure}
    \begin{subfigure}[b]{0.3\linewidth}        
        \centering
        \includegraphics[trim = 0mm 115mm 70mm 0mm, clip,scale=0.26]{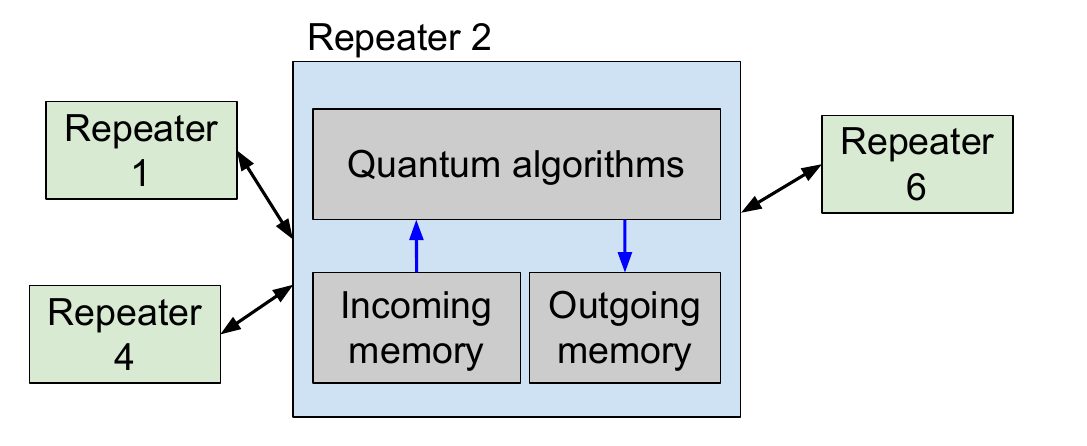}        
        \caption{Repeater 2 network connectivity}
        \label{sfig:HL_repeater_structure}
    \end{subfigure}
    \begin{subfigure}[b]{0.3\linewidth}        
        \centering
        \includegraphics[trim = 0mm 120mm 140mm 0mm, clip,scale=0.31]{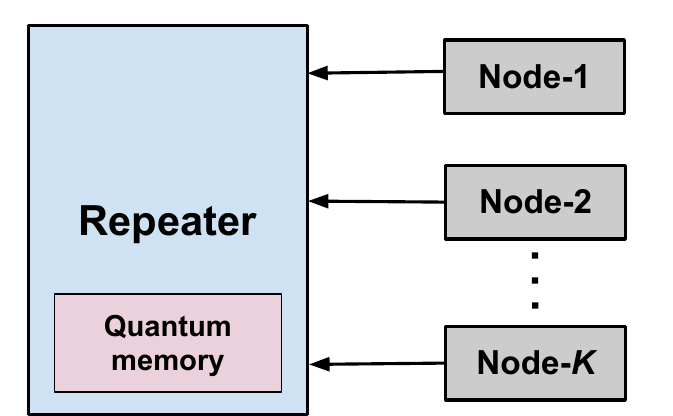}
        \caption{$K$ inputs single repeater}
        \label{sfig:single_repeater}
    \end{subfigure}
    \caption{(\subref{sfig:arbitrary_network}) Arbitrary network with two end nodes where each repeater sends and receives qubits, (\subref{sfig:HL_repeater_structure}) high level view of repeater connectivity (Repeater 2 connectivity is shown): incoming qubits (from other repeaters) are stored in Incoming memory, and outgoing qubits (to other repeaters) are stored in Outgoing memory. (\subref{sfig:single_repeater}) A simplified view of a repeater with $K$ input nodes. With probability $p_k, 0 < p_k < 1$, the repeater receives a qubit from the $k^\textrm{th}$ node, $1 < k < K$.}
    \label{fig:topology_examples_w_SR}
\end{figure}

\begin{figure}
    \centering 
    \includegraphics[trim = 0mm 110mm 0mm 10mm, clip,scale=0.52]{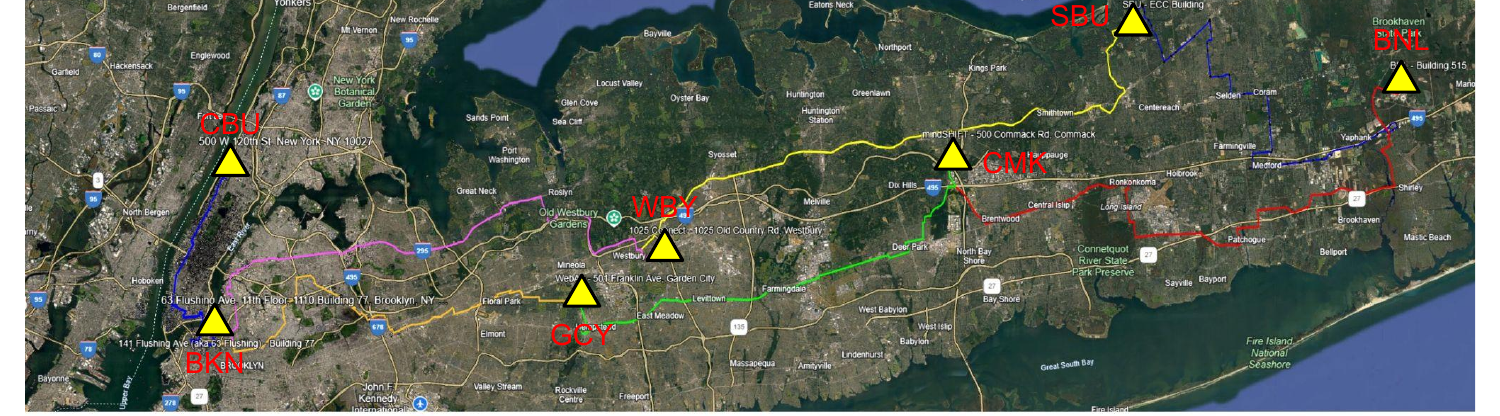}   
    \caption{The SCY-QNet quantum network testbed using fiber infrastructure connecting Stony Brook University (SBU), Brookhaven National Laboratory (BNL), and Columbia University (CBU)~\cite{FigueroaRG, Katramatos2024}. Map obtained from Google Earth ~\cite{GoogleEarth}.}
    \label{fig:SCY_google_earth}
\end{figure}

\begin{figure}
    \centering
    \begin{subfigure}[b]{0.47\linewidth}        
        \centering
        \includegraphics[trim = 20mm 80mm 20mm 90mm, clip,scale=0.29]{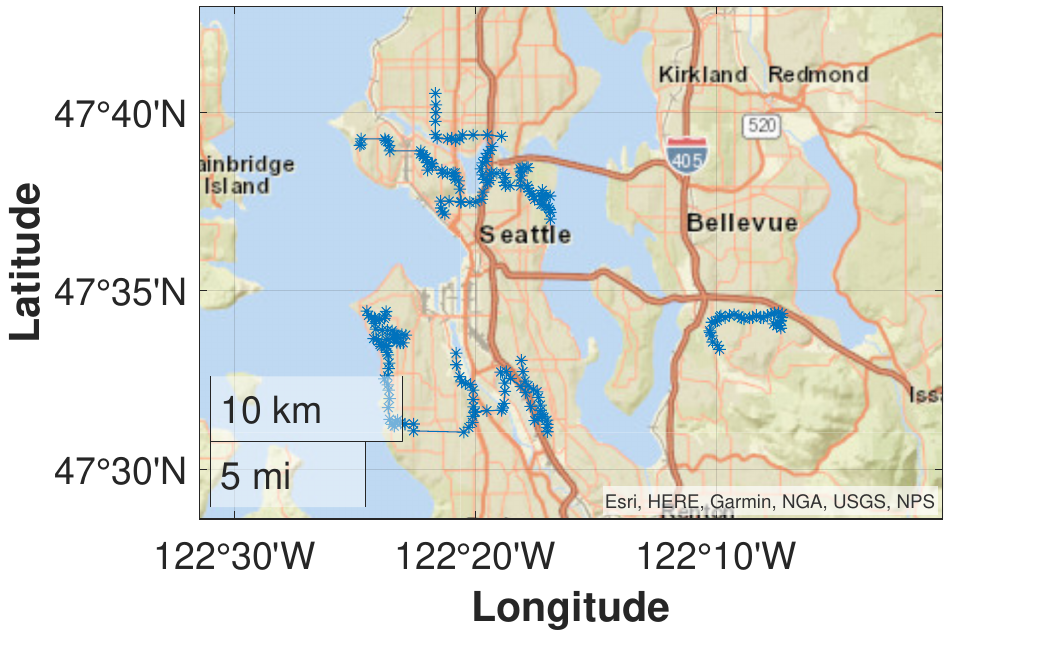}   
        \caption{}
        \label{sfig:Seattle_geo}
    \end{subfigure}
    \begin{subfigure}[b]{0.47\linewidth}        
        \centering
        \includegraphics[trim = 20mm 80mm 20mm 90mm, clip,scale=0.29]{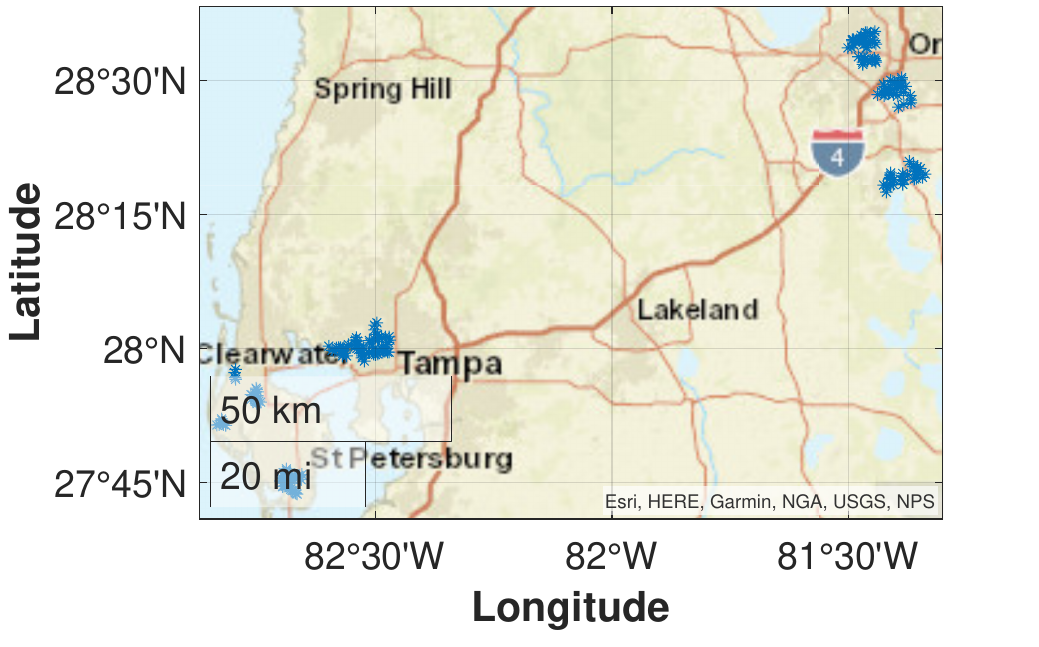}   
        \caption{}
        \label{sfig:Orlando_geo}
    \end{subfigure}

    \medskip
    \centering
    \begin{subfigure}[b]{0.47\linewidth}        
        \centering
        \includegraphics[trim = 50mm 90mm 50mm 90mm, clip,scale=0.31]{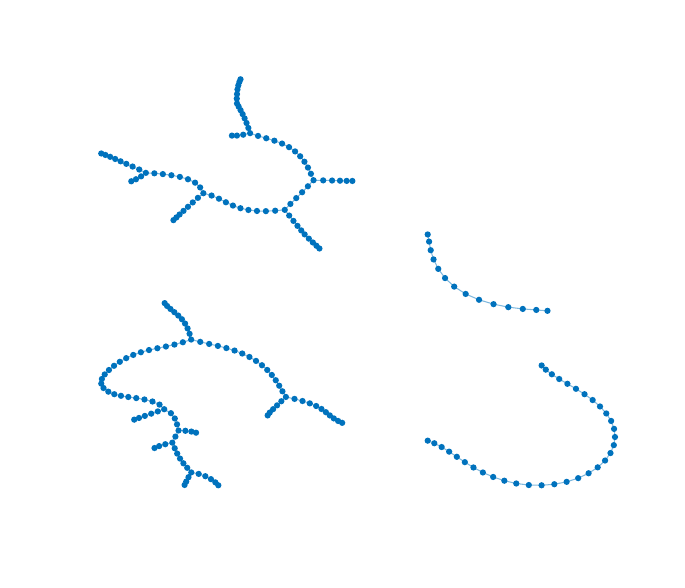}
        \caption{}
        \label{sfig:Seattle_graph}
    \end{subfigure}
    \begin{subfigure}[b]{0.47\linewidth}        
        \centering
        \includegraphics[trim = 50mm 90mm 50mm 90mm, clip,scale=0.31]{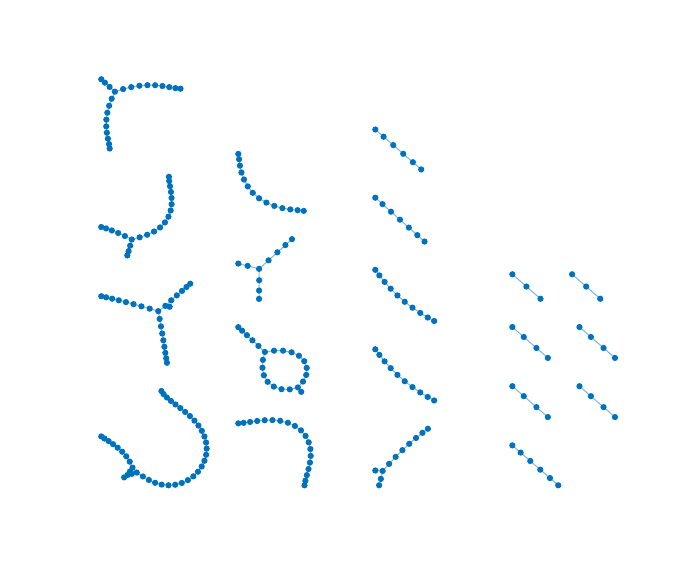}
        \caption{}
        \label{sfig:Orlando_graph}
    \end{subfigure}
    \caption{(\subref{sfig:Seattle_geo}) (\subref{sfig:Orlando_geo}) Crown Castle fiber infrastructure in Seattle, WA and Orlando, FL (\subref{sfig:Seattle_graph}) (\subref{sfig:Orlando_graph}), and their graph representations. GIS data obtained from~\cite{Sathiya2020}.}
    \label{fig:practical_network_example}
\end{figure}

%
Network nodes create BPs and share them across all their neighbors using fibers (quantum connections). Incoming qubits (from other nodes) are stored in Incoming memory, and outgoing qubits (to other nodes) are stored in Outgoing memory, as shown in Fig~\ref{sfig:HL_repeater_structure}. According to swapping policy, each node selects pairs of qubits from its memory and performs swapping. Classical swapping results are then shared with the participating nodes using fibers (classical connections) to create EPs. Losses and quantum errors may occur in fibers (due to fiber length) and in node's memory (due to storage time). Both quantum and classical information experience delay due to fiber length. 

\noindent A common quantitative performance metric for a quantum network is the \emph{EP Rate (EPR)} which is the total number of successful EPs established between the end nodes relative to the maximum possible number of EPs. Fidelity is used as a quality measure metric and it is calculated for each EP (between end nodes), relative to the generated BP. The \emph{averaged Fidelity} is obtained by averaging the Fidelity across all EPs. The propagation of qubits and classical communication messages across the network links (fibers) as well as the swapping operations introduce delay. In order to establish EPs, nodes should have enough memory to store qubits during these delay periods. 
For a single repeater, as depicted in Fig.~\ref{sfig:single_repeater}, we will focus on its expected (incoming) memory utilization (shared across all incoming inputs), $E[Q]$~\cite{Towsley2021a}, and \emph{Pairing Rate} (PR or $R$). Since a single repeater is not part of a network, PR is the performance metric where swapping is not performed as in EPR.

Several quantum network simulators~\cite{bel2024, NetSquid,QuISP,SeQUeNCe} have been introduced to evaluate the performance of quantum networks. For practical network performance evaluation, \emph{we developed a custom hybrid quantum network simulation platform that incorporates both quantum and classical network aspects}, in contrast to quantum focused simulators. The simulation platform takes into account fiber impairments, quantum losses, and classical connectivity constraints. Both quantum communications and classical packets are handled and processed jointly to enable future hybrid connectivity (when the same fiber is used for both quantum and classical communications).

In order to verify the accuracy of the simulation platform, we compare the simulation results to analytic results. Specifically, a single quantum repeater as depicted in Fig.~\ref{sfig:single_repeater}, was analyzed in~\cite{Towsley2021a,Towsley2021b,Towsley2023,Towsley2020}. It was shown that when the number of inputs $K$ is greater than two ($K >2$), the expected memory utilization of a repeater, $E[Q]$, can be obtained~\cite{Towsley2021a}. In contrast, for the DIR ($K=2$), $E[Q]$ cannot be bounded~\cite{Towsley2021b}.
For the case of a single repeater with $K>2$, we compare the analytical results of~\cite{Towsley2021a} to our simulation outputs. We show that the simulation results are within $3\%$ of the analytical results.
%

As DIRs ($K=2$) are key network components and with the understanding that the memory size is limited, the main focus should be on performance degradation rather than expected memory usage. Therefore, \emph{we introduce an analytical approach to model the PR ($R$) of DIRs using Discrete Time Markov Chain (DTMC)}. The minimum size of memory $M$ that achieves $R(M) = (1-\delta) R_\textrm{max}$ for max possible rate $R_\textrm{{max}}$ and performance degradation of $\delta, \; 0 \leq \delta \leq 1$ relative to $R_\textrm{{max}}$, is calculated for two models: on-the-fly pairing and full-memory-usage. 
We use our analytical results to show that the mean $R(M)$ matches within $0.3\%$ to the EPR simulation results.
Then, \emph{we show analytically that a small memory size is sufficient and leads to minimal performance degradation}. For instance, memory size of $M=3$ leads to only $3\%$ ($\delta=0.03$) degradation in performance (for input arrival probabilities of $p_1 = 0.2$ and $p_2 = 0.8$ from the two nodes).

When expanding to line topology networks, quantum errors are considered for both the fiber connections (phase flip dephasing) and the quantum memory (amplitude damping and phase flip dephasing).
We introduce a simple heuristic algorithm to maximize the EPR. Average Fidelity is calculated for the end-to-end EPs as a quality measure. Via simulations, we explore the EPR performance and average Fidelity of:
\begin{enumerate*}[label=(\roman*)]
    \item single or 3 repeater synthetic (artificial) networks with equal fiber lengths,
    \item single repeater synthetic network with 1:2 ratio of its fiber lengths, and
    \item single and 3 repeater practical (real world distances) networks.
\end{enumerate*}

We show that for practical single repeater networks (case (iii)) with low values of $M$, EPR performance is low due to network latency (after swapping, the classical result is sent via fiber and arrives to the receive node after some delay). For example, in the BNL$\leftrightarrow$SBU$\leftrightarrow$WBY sub-network of SCY-QNet (see Fig.~\ref{fig:SCY_google_earth}), when $M=6$ the EPR is below $0.015$. We show that for high values of $M$, performance is sensitive to the outgoing memory size, and therefore, additional memory is required to obtain an EPR that is close to the max EPR. As fiber lengths increase, qubits are stored in memory for longer durations to enable EP creation. In this case, the Fidelity mainly degrades with high amplitude damping rate (in quantum memory). For instance, for the same sub-network, we see a degradation of $0.2$ in Fidelity between several \thinspace{Hz} amplitude damping rate to hundreds of \thinspace{Hz} amplitude damping rate for low $M$, and up to $0.45$ degradation for high $M$. 

For practical multi repeater networks (case iii), \emph{EPR and average Fidelity are obtained via simulation for sub-networks of SCY-QNet (Fig.~\ref{fig:SCY_google_earth}) for two paths between Columbia (CBU) and Brookhaven National Lab (BNL)}. We observe that there is a minimum $M$ to enable EP creation due to fiber delay. We identify that large memory size is needed to obtain higher EPR due to the overall network latency. 
As the number of nodes increases or when the fiber length between nodes differs significantly, additional node memory is needed to maintain the same EPR (qubits are stored for longer duration to allow completion of entanglement across the network). For example, an equal length fiber networks with a single node or $3$ nodes require memory size of $45$ or $80$ to maintain the same EPR of $0.79$. Memory size should be above some minimum threshold to allow EP creation (EPR $>0$) as pair selection and swapping completion propagate across the network.
As fiber lengths increase, the average Fidelity degrades for high amplitude damping rate, high coherence dephasing rate and higher probability of phase flips (both in node memory and/or fibers) due to overall latency and memory storage time. The more time the qubit spends in memory or in fiber, the higher the probability of quantum error to occur. For instance, while the channel quantum error is fixed, the Fidelity drops from $0.6 - 0.7$ (low quantum error memory settings) to $0.3 - 0.35$ (high quantum error memory settings).

The paper is organized as follows. 
Section~\ref{sec:related_work} reviews related work and Section~\ref{sec:network_SR_models} introduces the single repeater and network models. Section~\ref{sec:simulation_platform} introduces the simulation platform. Section~\ref{sec:multi_input_repeater} compares simulation results to prior analysis of 3, 4, 6, 7, 9, and 10 input repeaters. Section~\ref{sec:dual_repeater} introduces analytical results for DIR and compares the analytical and simulation results. Section~\ref{sec:quantum_network_performance} introduces a heuristic algorithm to maximize the EPR of line topology network followed by simulation results for both synthetic and practical networks. Finally, Section~\ref{sec:summary} summarizes the paper and discusses directions for future work.

%
%
%
%
%
%
%
%
\section{Related work}\label{sec:related_work}
Quantum networks connect quantum nodes to each other using fibers (physical layer) and provide means of EP creation, qubit teleportation and quantum applications (upper layers).
Several protocol stack approaches are proposed~\cite{Gayane2023,Illiano2022,Kozlowski2020} to establish upper layer hierarchy -- similar to the classical world, while facing specific challenges~\cite{Cacciapuoti2020a} such as no cloning and synchronization. To overcome these, quantum multi-cast methods were proposed~\cite{Shi2006,Chen2022a} and synchronization mitigation is widely discussed~\cite{Okeke2018,Chuang2006,Shi2022}. 

In parallel, efforts have been made towards network connectivity in the physical layer and several approaches to network flow were introduced. One approach to the network flow is to create / prepare the next connected node / link while the qubit is propagating, similar to a Datagram approach~\cite{Bacciottini2024} or to establish a full connectivity flow. The latter suggests that a propagation path is created (determined) for the qubit which includes nodes that it will traverse at later steps. The network capacity~\cite{Pirandola2019} and specifically the achievable EPR, is a desired metric as high EPR means more basic ingredients for quantum algorithms (e.g., QKD). Several methods were suggested to maximize EPR: two step~\cite{Pant2019, Chen2022b, Zeng2022, Zhao2022a, Zhao2021, Shi2024}, segmented~\cite{Zhao2022b}, swapping trees~\cite{Gupta2022}, base graph~\cite{Gyongyosi2018}, maximizing QKD~\cite{Amer2020} or other~\cite{Li2021}. 

The network \emph{topologies} used for analysis are usually generic line or lattice~\cite{Pant2019}, enabling high diversity paths between nodes. This leads to a complex path selection optimization, where the decisions regarding the node to which qubits need to be sent and the selection and prioritization of qubit pairs, determines EPR. 
Nodes can be co-located in a single location (e.g., lab) or spread across multi locations. Single location deployed networks are usually used to verify specific characteristics of quantum components or small scale EP sharing~\cite{Herbauts2013}, where multi-lab testbeds start to emerge~\cite{Katramatos2024}. In the later case, dedicated dark fibers are assigned to the quantum and classical connections (usually two pairs to accommodate transmit and receive directions).
Quantum memory, located inside the repeater, is used to store incoming (from connected nodes) and outgoing (to connected nodes) qubits to compensate for time misalignment and to allow swapping and purification. Large memory capacity is currently hard to achieve and expensive to implement, and therefore, efforts are made to reduce the waiting time for swapping~\cite{Loock2019} and improve the swapping algorithm~\cite{Mastriani2023}.

An analytical approach that uses Continuous Time Markov Chain (CTMC) and DTMC was introduced in~\cite{Towsley2021a,Towsley2021b} to obtain the $E[Q]$ of the quantum memory (for incoming qubits) inside a single repeater. $E[Q]$ was calculated for bipartite matching, tripartite~\cite{Towsley2023}, and multipartite ~\cite{Towsley2020}. For $K \geq 3$, $E[Q]$ has a closed form showing that small memory size is needed on average. In contrast, for $K=2$, the memory utilization goes to infinity, and therefore, any memory will eventually be full regardless of its size. 

Several quantum network simulators were introduced~\cite{bel2024, NetSquid,QuISP,SeQUeNCe} for quantum network performance evaluation. NetSquid~\cite{NetSquid}, for example, is a popular simulator that provides quantum computing foundation and quantum algorithm driven classes. We decided to develop a custom simulation platform that incorporates the classical network along with the quantum network. The simulation platform takes into account fiber impairments, quantum losses, and classical connectivity constraints. Both quantum and classical packets are handled and processed jointly to enable future hybrid connectivity (when the same fiber is used for both quantum and classical communication).

%
%
%
%
%
%
%
%
\section{Network and single repeater models}\label{sec:network_SR_models}

\subsection{Swapping} \label{ssec:swapping}
Swapping operation is depicted in Fig.~\ref{fig:swapping}. Each party (e.g. Alice and Bob) prepares a BP and send one of its qubit (via fiber) to a repeater. The repeater performs CNOT (Controlled NOT gate) on the two qubits he receives from Alice and Bob followed by Hadamard gate (H gate). After measuring the two qubits, the repeater sends the 2 bit classical result to Bob (via fiber). Depending on the result, Bob performs $X$ gate, $Z$ gate or both on his side of the BP. The result is a shared EP between Alice to Bob.

\begin{figure}
    \centering
    \includegraphics[trim = 0mm 35mm 0mm 0mm, clip,scale=0.3]{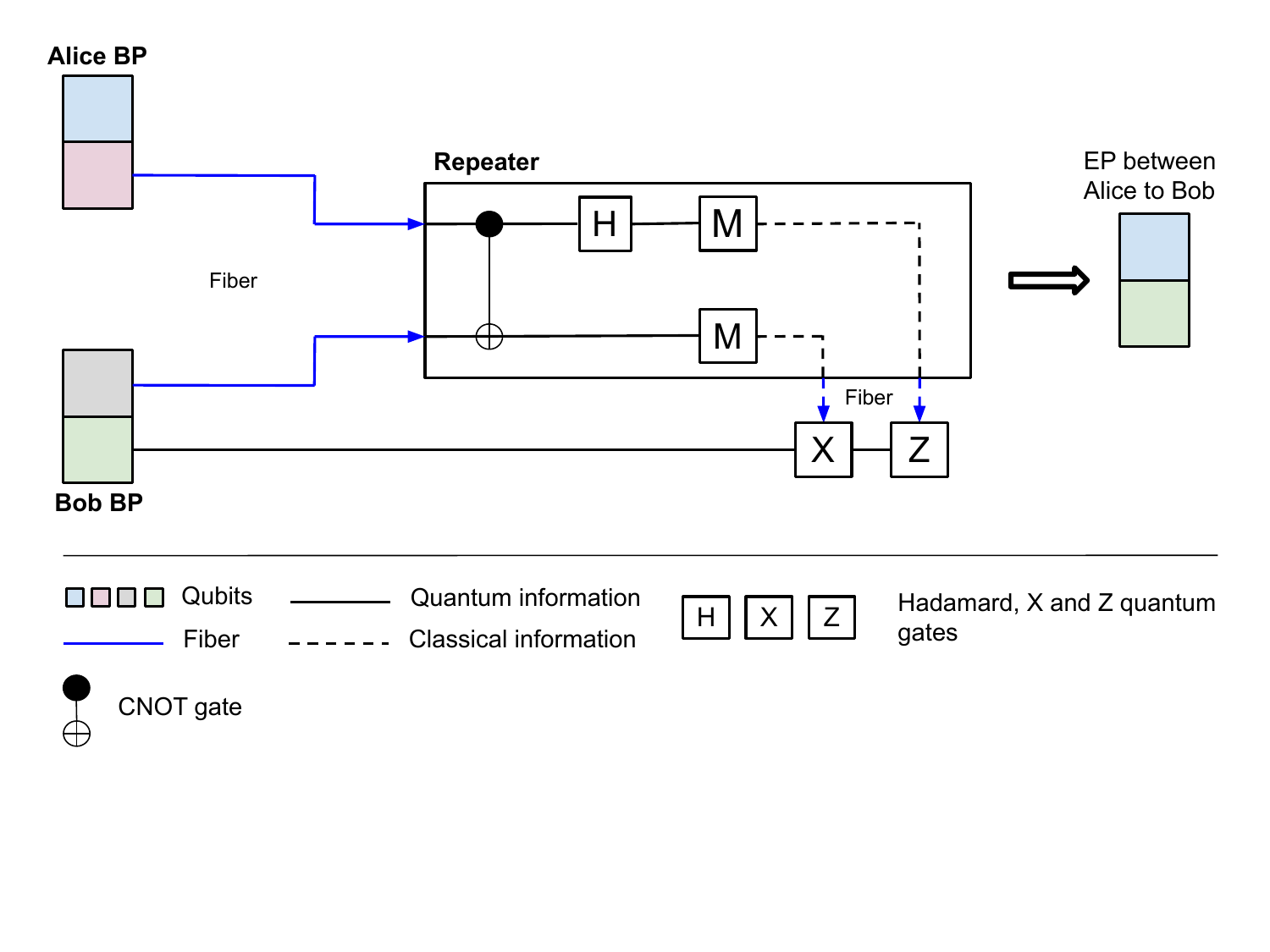}       
    \caption{Swapping algorithm flow.}
    \label{fig:swapping}
\end{figure}

\subsection{Network model} \label{ssec:network_model}

\subsubsection{General architecture} \label{sssec:general_architecture}

A given network consists of $N$ nodes, two of which are designated as the selected end nodes (i.e., Alice and Bob) at a specific time, see Fig.~\ref{sfig:arbitrary_network}. We seek EPR performance and EP average Fidelity at the end nodes.
EPR is the total number of successful EPs established between the end nodes relative to the maximum possible number of EPs. Fidelity~\cite{Nielsen_Chuang_2010} is calculated between the BP and the resulted EP. The averaged Fidelity is obtained by averaging the Fidelity across all EPs. The other $N-2$ non end nodes are repeaters, capable of performing swapping. We can view such network as a graph with $N$ nodes whose edges' weights represent the distances between nodes. In the special case that the network topology is a line with no loops, each of the $N-2$ repeaters is a DIR and the end nodes have only one edge.

When inspecting the practical networks described in Fig.~\ref{fig:SCY_google_earth} and Fig.~\ref{fig:practical_network_example}, we observe that:
\begin{enumerate*}[label=(\roman*)]
    \item most of the nodes have only two edges,
    \item it can be assumed that each node represents a repeater, and
    \item there are almost no closed loops. Hence, for most cases, the network graph (or most of its subgraphs) can be represented by a line. 
\end{enumerate*}

In this paper, we assume that the network is composed of a line with no loops. Additionally, we assume the following:
\begin{enumerate*}[label=(\roman*)]
    \item all nodes are capable of generating BPs at every time unit and can share a BP with all of their neighbors, 
    \item all repeater nodes are capable of performing swapping,
    \item newly arrived qubits replace old stored qubits once the memory is full, and
    \item if a qubit is lost (in fiber or dropped from memory), the associated qubit (part of BP) is removed from nodes' memory (incoming and/or outgoing). 
\end{enumerate*}

A quantum memory with size $M_\textrm{in}$ is used to store incoming qubits and $M_\textrm{out}$ quantum memory size is used to store outgoing qubits. Henceforth, we assume $M_\textrm{in} = M_\textrm{out}=M$. 
Each repeater matches pairs of incoming and/or outgoing qubits that originated from different nodes. Several matching options exist:
\begin{enumerate*}[label=(\roman*)]
    \item between a pair of BPs (connected to two different neighbor nodes).
    \item between BP (connected to neighbor node) and EP (entangled to non neighbor node).
    \item between a pair of EPs (entangled to two different non-neighbor nodes).
\end{enumerate*}
Swapping policy~\cite{Gupta2022,Pant2019,Zhao2022b} is used to prioritize and select the matching pairs from the available options.
 
Define a $N \times N$ dimension distance matrix $D$ whose entries $d_{ij}=d_{ji}$ are the distances between any two nodes $i,j$, $1 \leq i,j \leq N$ where $d_{ii}=0 \; \forall i$.  The physical distances $d_{ij}$ between nodes are predefined and known to all nodes.

\subsubsection{Network flow} \label{sssec:network_flow}
Fig.~\ref{fig:sim_line_flow} illustrates the network flow steps:
\begin{enumerate}[label=\protect\circled{\arabic*}]
    \item On every time step (iteration), BPs are generated in every repeater. One BP for every neighbor node.
    \item BP qubits are stored in Outgoing memory.
    \item BP qubits are sent via fiber to all neighbor nodes and arrive to their destination after some delay (due to fiber length). Some of the qubits are lost or experience quantum errors (to be discussed in Section \ref{sssec:quantum_channels}). If qubits are lost, the associated stored qubits (the other 'half' of the BP or EP) in nodes' Outgoing memory are cleared.
    \item Any incoming qubits are stored in Incoming memory.
    \item Pairs of qubits (from incoming and/or outgoing memory) are selected according to swapping policy (qubits belong to two different participating nodes). While qubits are stored, they experience additional quantum errors (to be discussed in Section \ref{sssec:quantum_memory}). If qubits are lost or dropped (replaced in memory), the associated stored qubits are cleared from other nodes memory.
    \item Each repeater performs swapping, if applicable (pairs of qubits are available for swapping), and send the classical result to one of the participating nodes. 
    \item classical swapping information is sent via fiber and arrives to its destination after some delay (to be discussed in Section \ref{sssec:classical_channels}). 
    \item Underlying information for the simulation such as lost qubits, qubits that are part of BP/EP that participate in swapping etc., is exchanged via virtual channel (management).
    \item EP is created between the participating nodes. If both the participating nodes are end nodes, a success event is counted and the Fidelity of the EP is calculated.
\end{enumerate}

\begin{figure}
    \centering
    \includegraphics[trim = 0mm 30mm 0mm 0mm, clip,scale=0.35]{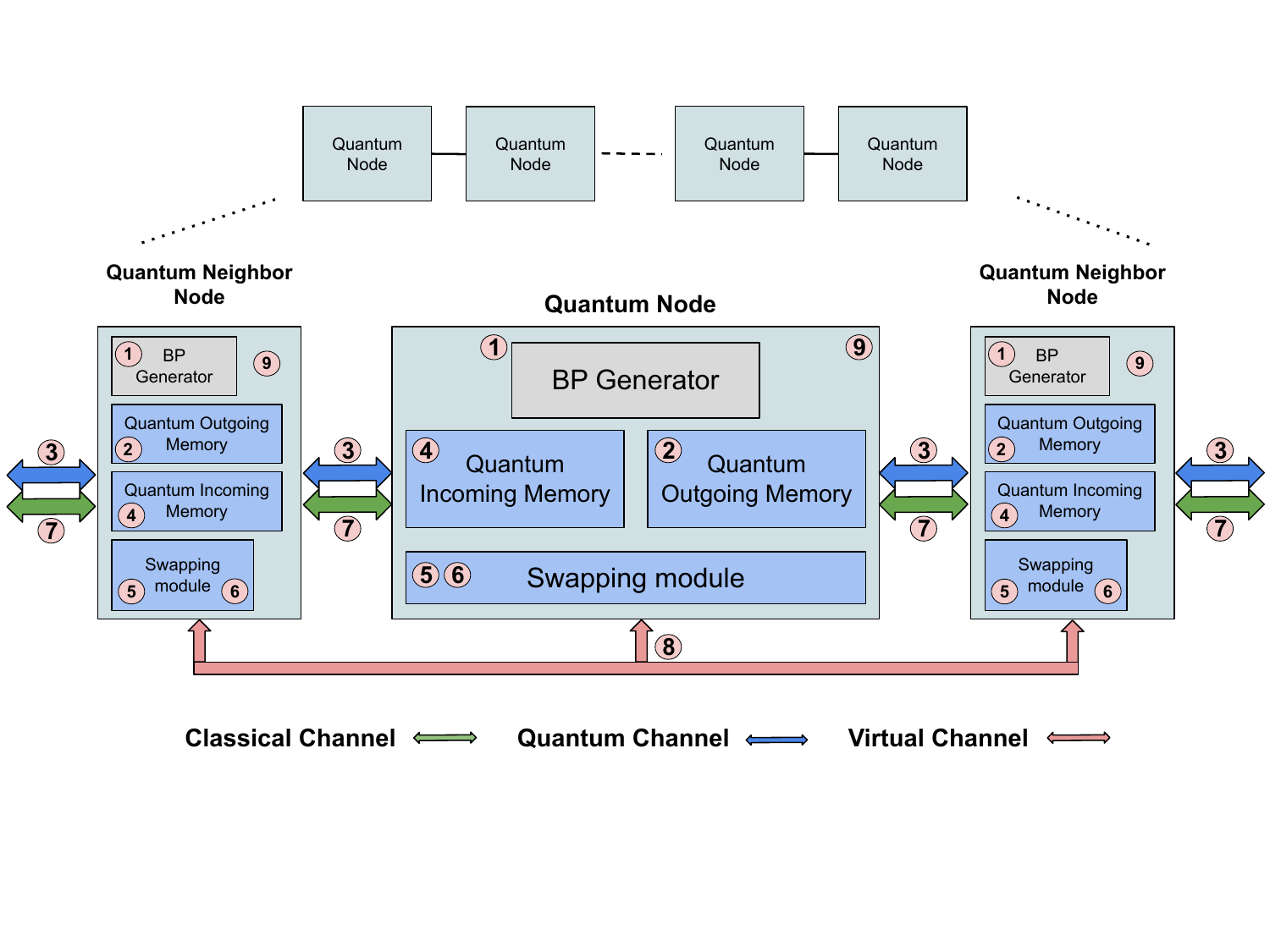}       
    \caption{Flow for line topology network.}
    \label{fig:sim_line_flow}
\end{figure}

\subsubsection{Classical channels} \label{sssec:classical_channels}
Underlying classical fiber infrastructure allows full lossless connectivity between all nodes. Due to signal propagation, the propagation delay $\tau_{ij}$ between nodes $i$ and $j$ is $\tau_{ij} = n \cdot d_{ij}/c$, $1 \leq i,j \leq N$, where $d_{ij}$ is the fiber length, $n$ is the refractive index (taken as $1.444$ for $1550$\thinspace{nm} channels) and $c$ is the speed of light ($c=3\cdot 10^8$\thinspace{m/s}). Note that $\tau_{ij}=\tau_{ji}$.

\subsubsection{Quantum channels} \label{sssec:quantum_channels}
A dedicated fiber infrastructure connects all neighbor nodes and allows sending and receiving qubits. We assume that each transmitted qubit undergoes the following:

\paragraph{Propagation delay}
Same as for classical bits, as described in Section \ref{sssec:classical_channels}.

\paragraph{Loss}
We assume a high success probability of $0.99$~\cite{VanDevender2007} for an up/down frequency conversion to/from the C-band wavelength and a loss of $0.2$\thinspace{dB/Km}. A qubit arrives from node $i$ to node $j$ with probability $p_{ij} = 0.99 \cdot 10^{-0.2 \cdot d_{ij}/10}$ for $1 \leq i,j \leq N$ that are neighbor nodes. Note that $p_{ij}=p_{ji}$.

\paragraph{Dephasing (phase flip)}
For a phase flip channel~\cite{Nielsen_Chuang_2010}, 
$P^{ij}_{\textrm{PF}}$ is Phase Flip probability for the fiber connection between node $i$ and node $j$.
We take $P^{ij}_{\textrm{PF}} = (1-\lambda^{ij}_{\textrm{PF}})/2$, where $\lambda^{ij}_{\textrm{PF}}$ is the time dependent visibility of the fiber: $\lambda^{ij}_{\textrm{PF}} = e^{-\sigma^{ij}_\phi/2}$, $\sigma^{ij}_\phi$ is the mean phase change. For example, under close to optimal conditions, for 25\thinspace{km} fiber length ($122\mu s$ delay), $\sigma_\phi=0.1$ rad ~\cite{Min2008}.

\subsubsection{Quantum memory} \label{sssec:quantum_memory}
We assume Electromagnetically Induced Transparency (EIT) quantum memory ~\cite{Fleischhauer2005, Liu2001} with limited storage time characteristics of $T_1$, population relaxation (decay of exited states to ground states), and $T_2$, coherence dephasing ~\cite{Ma2017} that are common to all quantum memories across all nodes. The associated quantum errors are as follows:

\paragraph{Amplitude damping}
Amplitude damping channel with $\gamma$ parameter is assumed ~\cite{Nielsen_Chuang_2010}. For a qubit stored in memory for time $t$, $\gamma = 1-e^{-t/T_1}$. We further define $\Gamma_1 = 1/T_1$ as the population relaxation rate.

\paragraph{Dephasing (phase flip)}
Similar to the Fiber phase flip, we define $P_{\textrm{M}} = (1-\lambda_{\textrm{M}})/2$, where $P_{\textrm{M}}$ is phase flip probability in memory and $\lambda_{\textrm{M}} = e^{-t/T_2}$. We further define $\Gamma_2 = 1/T_2$ as coherence dephasing rate.

\subsection{Single repeater model} \label{ssec:SR_model}

Consider a single repeater that is connected to $K$ neighboring nodes, as shown in Fig.~\ref{sfig:single_repeater}. In this case, since there are only connections to a single repeater, for simplicity, we use a single index for $p_{ij}$, $p_{ij}=p_{1j}=p_{k}$. In every time unit, each neighbor node $k, \; 1 \leq k \leq K$, sends a single qubit (which is part of a Bell pair, generated at node $k$) to the repeater. With probability $p_k, \; 0 < p_k < 1$, the repeater receives a qubit from the $k^\textrm{th}$ node. For the homogeneous case, where $p_k$ is identical across all nodes, we denote $P=p_k, \;\forall k$ and $P_\textrm{loss} = 1-P$.

We denote the maximum PR as $R_\textrm{max}$ (achievable by infinite $M$) and the $M$-dependent PR as $R(M)$ ($R(M) \leq R_\textrm{max}$). As in~\cite{Towsley2021a}, the number of qubits that are stored in the memory is denoted as $Q$ and the expected memory utilization of the repeater is denoted as $E[Q]$. Simulation result of the average memory utilization is denoted as $\overline{Q}$. Note that in the specific case of a single repeater, the EPR could be seen as Paring Rate (PR) as no further step (e.g., swapping) is needed after matching.

%
%
%
%
%
%
%
%
\section{Simulation platform}\label{sec:simulation_platform}
Several quantum-focused network simulators exist~\cite{bel2024, NetSquid,QuISP,SeQUeNCe}. To name a few, SeQUeNCe~\cite{SeQUeNCe} is an open source, discrete event simulator that can be customized according to usage. QUISP~\cite{QuISP} is a research group product that focuses on quantum error correction and large scale quantum networks. NetSquid~\cite{NetSquid} is a popular open source, discrete event, considered all-around quantum simulator.
In order to evaluate the performance of a given quantum network, we developed a quantum network simulation platform, that includes full network dynamics (packet flow), classical and quantum channels as well as node model. In addition, we added classical network aspects such as classical information (bits) messaging and encoding, to construct one \textit{joint} hybrid simulation. 
By jointly handling both classical and quantum packets, we established the foundation for future hybrid connectivity (when the same fiber is used for both quantum and classical communication) and quantum-classical frame structure construction,.

The simulation is written in Matlab (no Matlab toolboxes are used) and consists of two parts: quantum computing infrastructure (qubits and qubit manipulation) and the network dynamics (quantum and classical) that is built on top of it.
An object oriented approach is used to define the network components (and sub-components), for example:
\begin{enumerate*}[label=(\roman*)]
    \item Repeater or end nodes.
    \item Control node.
    \item Network entity.
\end{enumerate*}
Synchronization is maintained using distributed master clock, that advances in user defined time steps (iterations), thereby allowing to simulate any time related impairments.
A high-level view of the simulation platform appears in Fig.~\ref{fig:sim_platform_HL_view}. In general, nodes can be connected arbitrarily to each other by a configurable number of quantum channels to create any desired topology. Two nodes are designated as end nodes and an underlying classical infrastructure is used to simulate exchange of classical information. A control node manages any classical / quantum optimization (globally as well as distributing any required local optimization information).


\begin{figure}
    \centering
    \includegraphics[trim = 20mm 40mm 0mm 0mm, clip,scale=0.33]{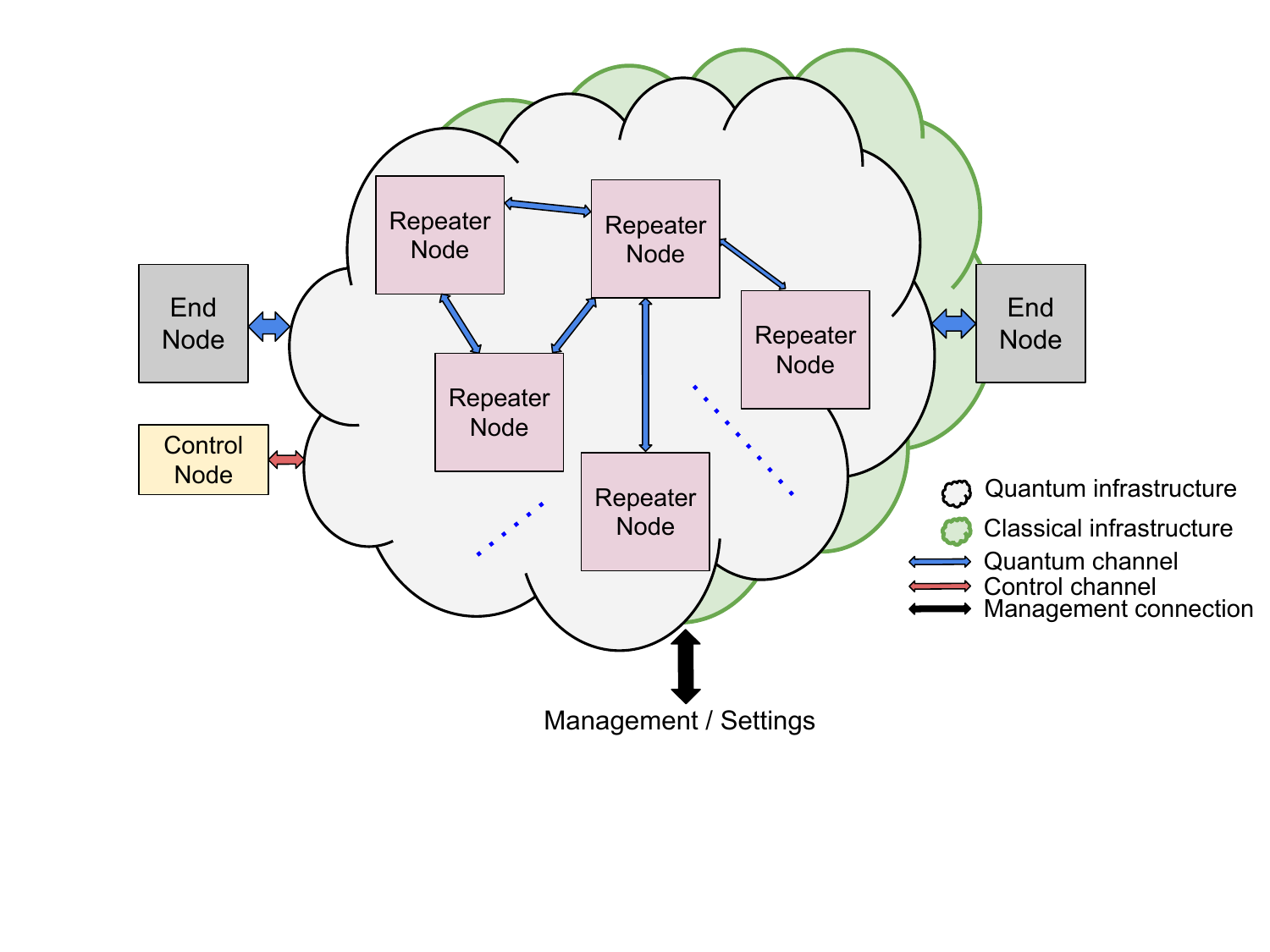}
    \caption{High level view of the simulation platform.}
    \label{fig:sim_platform_HL_view}
\end{figure}

Nodes and channels (classical, control or quantum) consist of multiple sub-components. These allow data path and control path modeling: quantum and classical channels (as discussed in Sections~\ref{sssec:classical_channels},~\ref{sssec:quantum_channels}), quantum impairments (as discussed in Section~\ref{sssec:quantum_memory}) as well as joint synchronization and control. Examples of the logical structure of selected network components: \textbf{Quantum node}: Fig.~\ref{sfig:quantum_node} shows the logical structure of a quantum node. Four different interfaces are used: quantum, and classical data paths as well as management and control paths. The node consists of quantum memory module that is used in coordination with the BP generator, the quantum interface and the quantum algorithm block. For example, qubits stored in the quantum memory could be part of swapping. A sync module is a local instance of synchronization that manages both the quantum and the classical timing to enable global network synchronization and timing error models. \textbf{Quantum channel}: Fig.~\ref{sfig:quantum_channel} shows the logical structure of a quantum channel. Qubits are stored in a First Out First In (FIFO) pipeline (corresponding to the fiber length) where associated delay and loss are applied. Additionally, quantum impairments (e.g., quantum error channels) could be considered. \textbf{Control node}: Fig.~\ref{sfig:control_node} shows the logical structure of the control node. Any network maximization/optimization problem (e.g., shortest path) is handled both for the quantum information and the classical information.

The simulation can be shared upon request.

\begin{figure}
    \centering
    \begin{minipage}{0.5\textwidth}
        \centering
        \begin{subfigure}{\textwidth}
            \includegraphics[trim = 0mm 50mm 0mm 0mm, clip,scale=0.29]{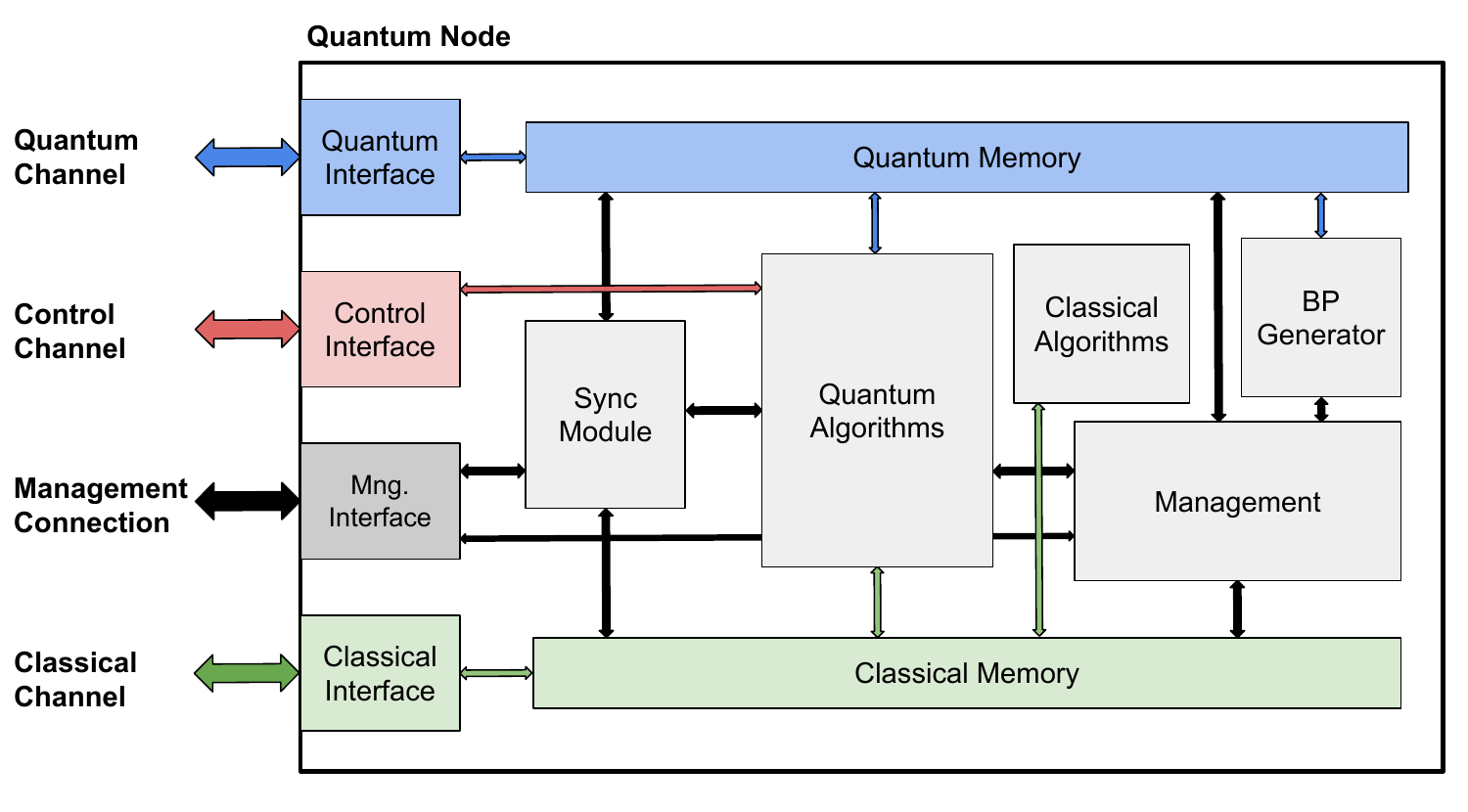}
            \caption{Logical structure view of a Quantum node.}
            \label{sfig:quantum_node}   
        \end{subfigure}
    \end{minipage}
    \hfill
    \begin{minipage}{0.45\textwidth}
        \centering
        \begin{subfigure}{\textwidth}
            \includegraphics[trim = 4mm 100mm 4mm 0mm, clip,scale=0.28]{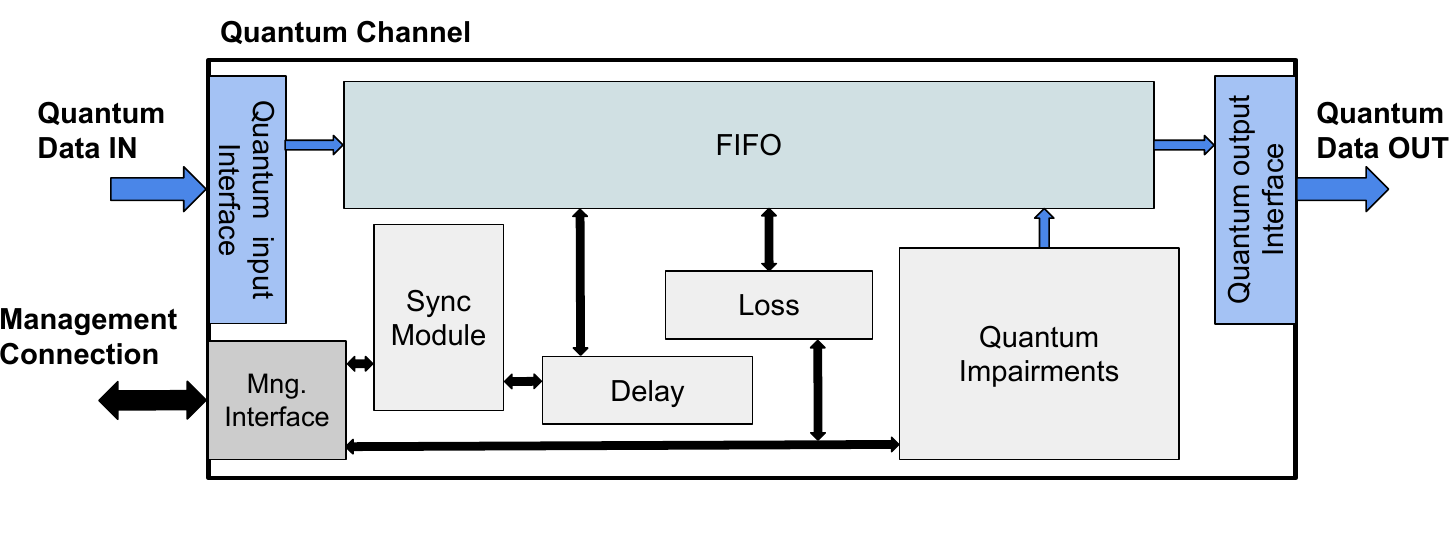}
            \caption{Logical structure view of a Quantum channel.}
            \label{sfig:quantum_channel}
        \end{subfigure}
        
        \vspace{0.2cm} 
        
        \begin{subfigure}{\textwidth}
            \includegraphics[trim = 14mm 100mm 0mm 0mm, clip,scale=0.28]{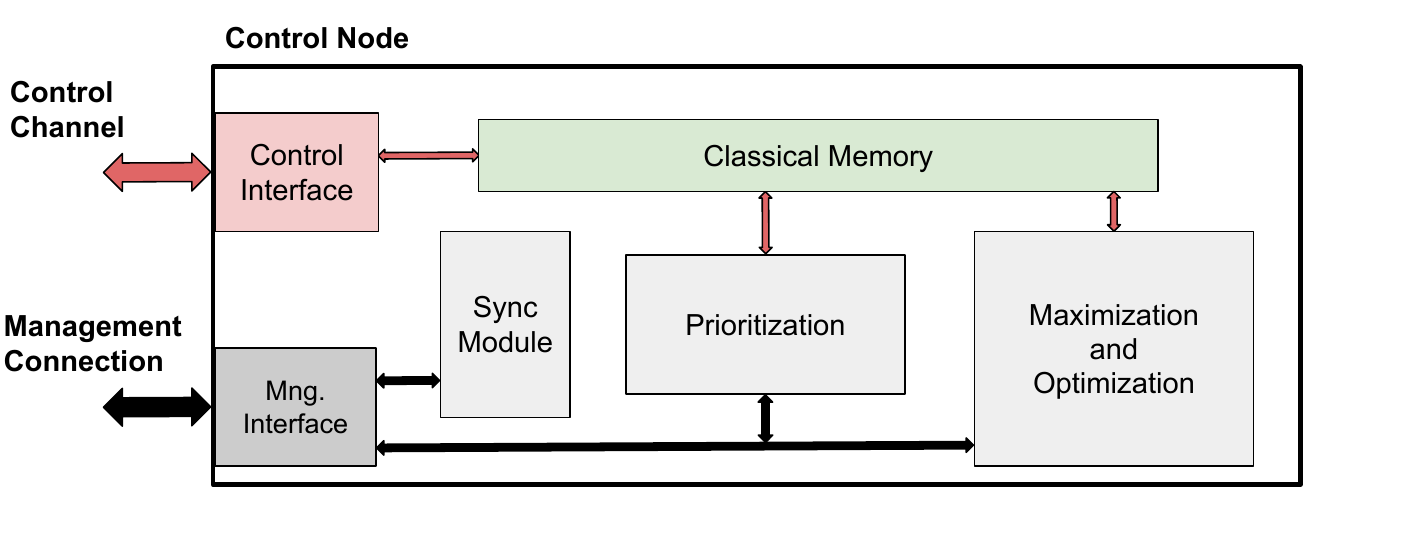}
            \caption{Logical structure view of a Control node.}
        \label{sfig:control_node}  
        \end{subfigure}
    \end{minipage}
    \caption{Logical structure of several simulation components.}
\end{figure}

%
%
%
%
%
%
%
%
\section{Repeater with more than two inputs}\label{sec:multi_input_repeater}

We will demonstrate the behavior of $E[Q]$ for the homogeneous case ($P=p_k, \;\forall k$) when $K \geq 3$. First, Fig.~\ref{sfig:k_3_memory_utilization} shows simulation results for the instantaneous repeater memory utilization for nodes 1--3 when $P=0.8$. The maximum utilization per specific node is $2$, where the maximum total utilization of the memory is $5$ (not shown in the graph). Recall that not all nodes' maximum utilization instances occur at the same time. The average utilization values per node are: $0.179$, $0.1815$, and $0.1745$ for nodes 1--3, respectively, which results in $\overline{Q} = 0.5350$.

\begin{figure}
    \centering
    \begin{subfigure}[b]{0.47\linewidth}        
        \centering
        \includegraphics[trim =0mm 0mm 0mm 0mm, clip,scale=0.20]{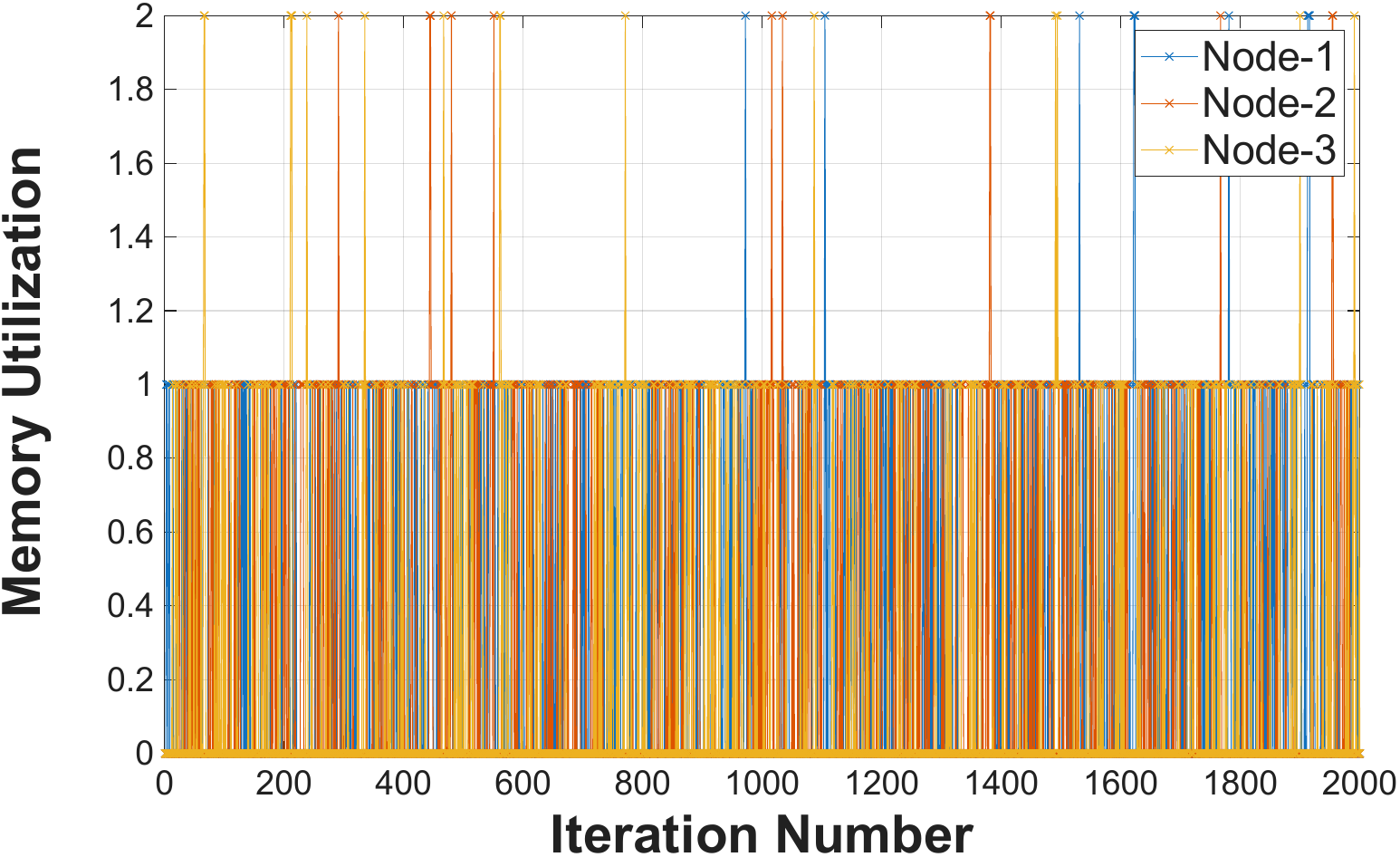}    
        \caption{$K=3$}
        \label{sfig:k_3_memory_utilization}
    \end{subfigure}
    \begin{subfigure}[b]{0.47\linewidth}        
        \centering
        \includegraphics[trim = 0mm 0mm 0mm 0mm, clip,scale=0.20]{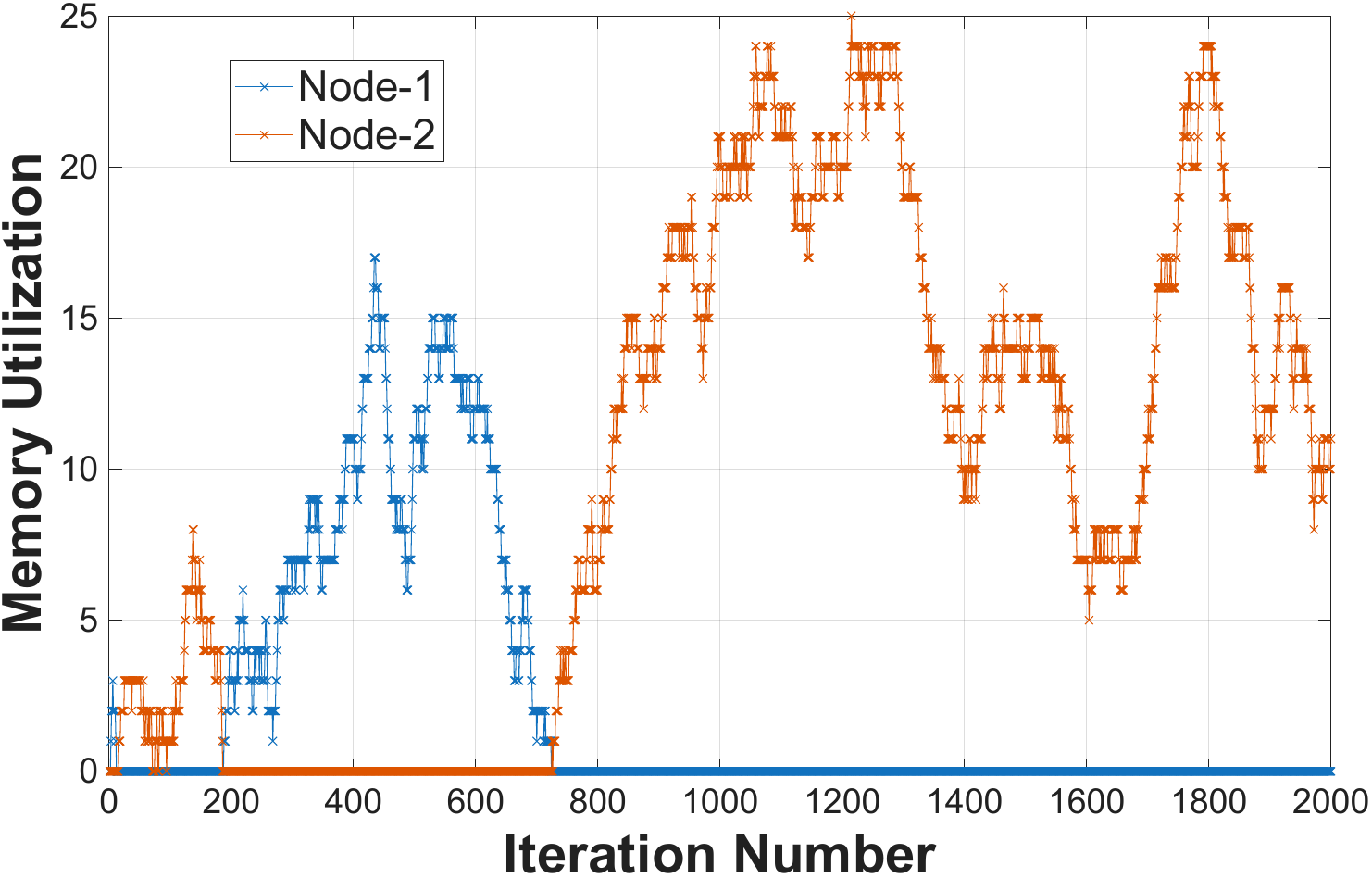} 
        \caption{$K=2$}
        \label{sfig:k_2_memory_utilization}
    \end{subfigure}

    \caption{Instantaneous memory utilization of a single repeater with (\subref{sfig:k_3_memory_utilization}) $K=3$ nodes and (\subref{sfig:k_2_memory_utilization}) $K=2$ nodes, for $P=0.8$.}
    \label{fig:memory_utilization_K2_K3}
\end{figure}

Utilization analysis for $E[Q]$ was introduced in~\cite{Towsley2021a,Towsley2021b} using DTMC. It was shown in~\cite{Towsley2021b} that taking $\beta_0$ as the root of $f(\beta) = (\beta P + P_\textrm{loss})^{k-1}  (P + \beta P_\textrm{loss}) - \beta$, $0 < \beta <1$, $E[Q]$ can be written as $E[Q] = \frac{1+\beta_0}{2(1-\beta_0)}$. Fig.~\ref{fig:SR_E_Q_comparison} shows a comparison between the analytical $E[Q]$ and our simulation results $\overline{Q}$ for various number of nodes, $K$, as a function of $P$. Each simulation point is the average of 5 independent runs of 1,800 iterations. Simulation results within $3\%$ of the analytical results.

\begin{figure}
    \centering
    \begin{subfigure}[b]{0.32\linewidth}        
        \centering
        \includegraphics[trim = 20mm 85mm 20mm 80mm, clip,scale=0.29]{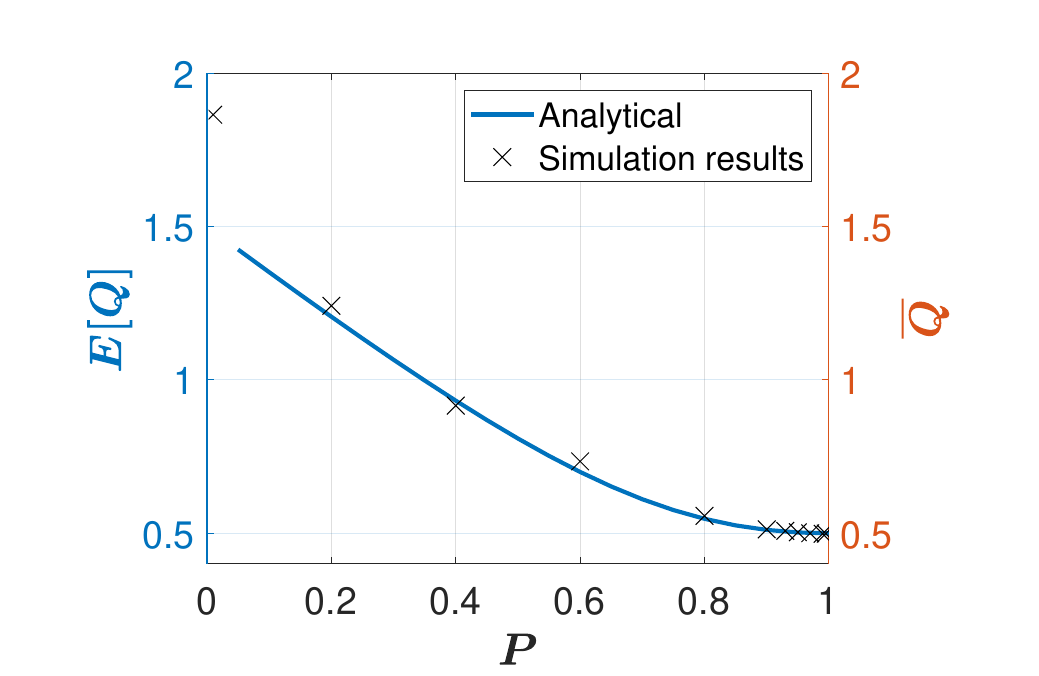}   
        \caption{$K=3$}
    \end{subfigure}
    \begin{subfigure}[b]{0.32\linewidth}        
        \centering
        \includegraphics[trim = 20mm 85mm 20mm 80mm, clip,scale=0.29]{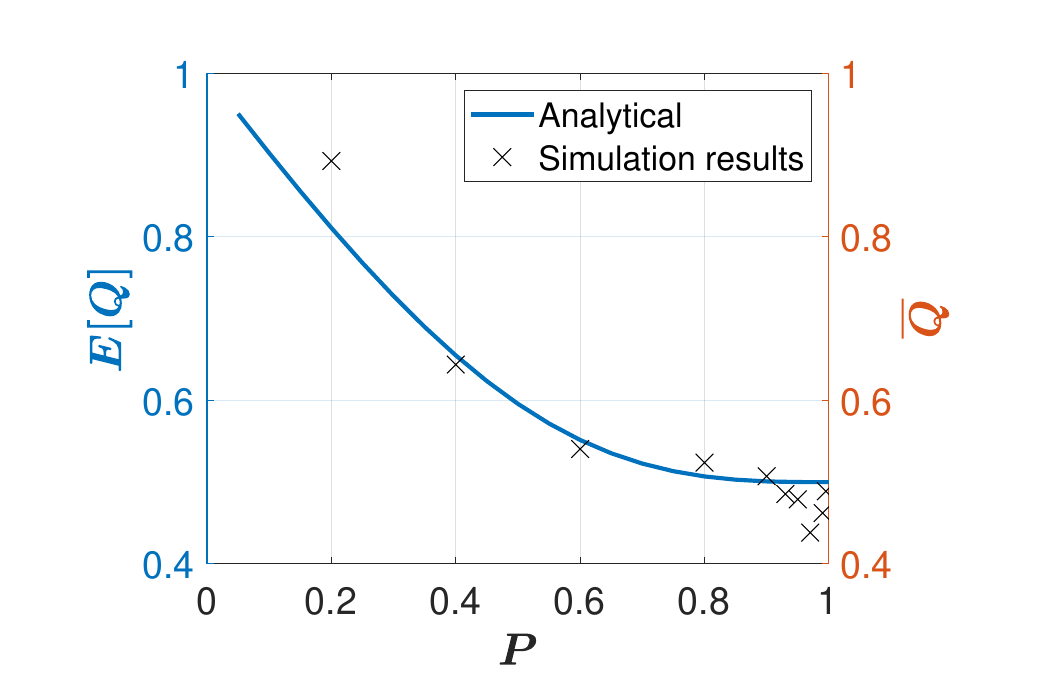}
        \caption{$K=4$}
    \end{subfigure}
    \begin{subfigure}[b]{0.32\linewidth}        
        \centering
        \includegraphics[trim = 20mm 85mm 20mm 80mm, clip,scale=0.29]{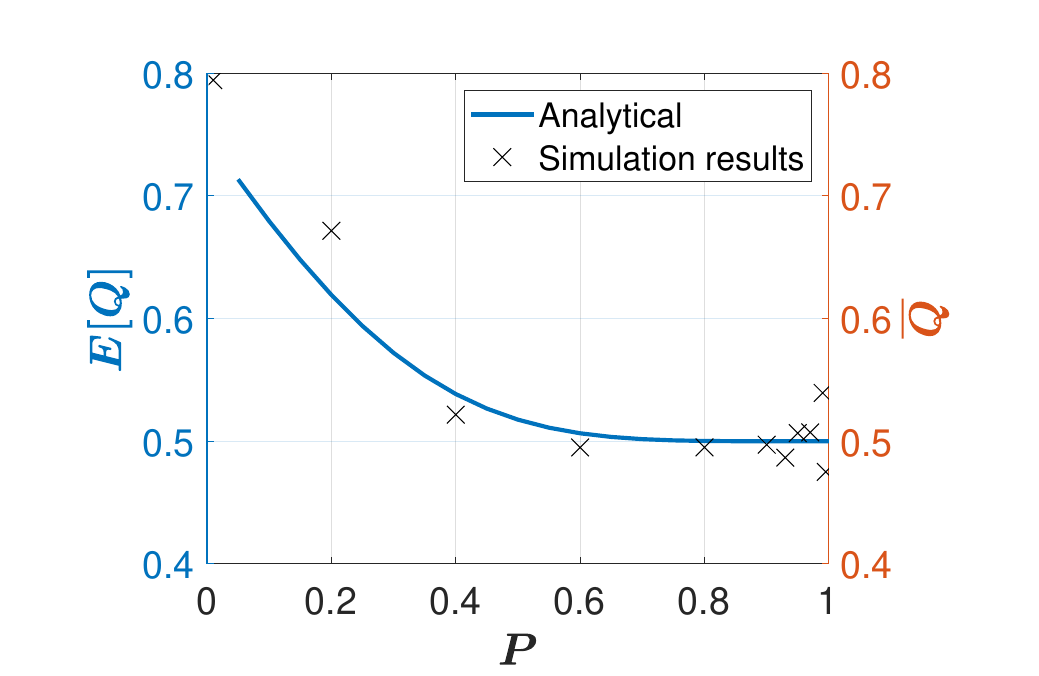}
        \caption{$K=6$}
    \end{subfigure}
    \medskip
    \centering
    \begin{subfigure}[b]{0.32\linewidth}        
        \centering
        \includegraphics[trim = 20mm 85mm 20mm 80mm, clip,scale=0.29]{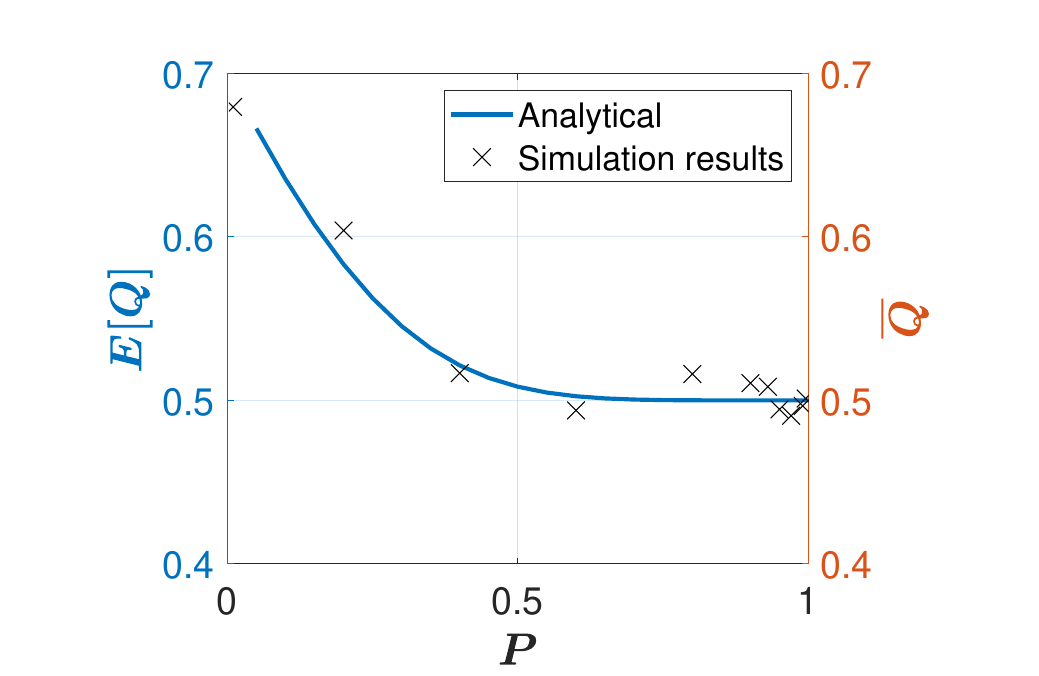}   
        \caption{$K=7$}
    \end{subfigure}
    \begin{subfigure}[b]{0.32\linewidth}        
        \centering
        \includegraphics[trim = 20mm 85mm 20mm 80mm, clip,scale=0.29]{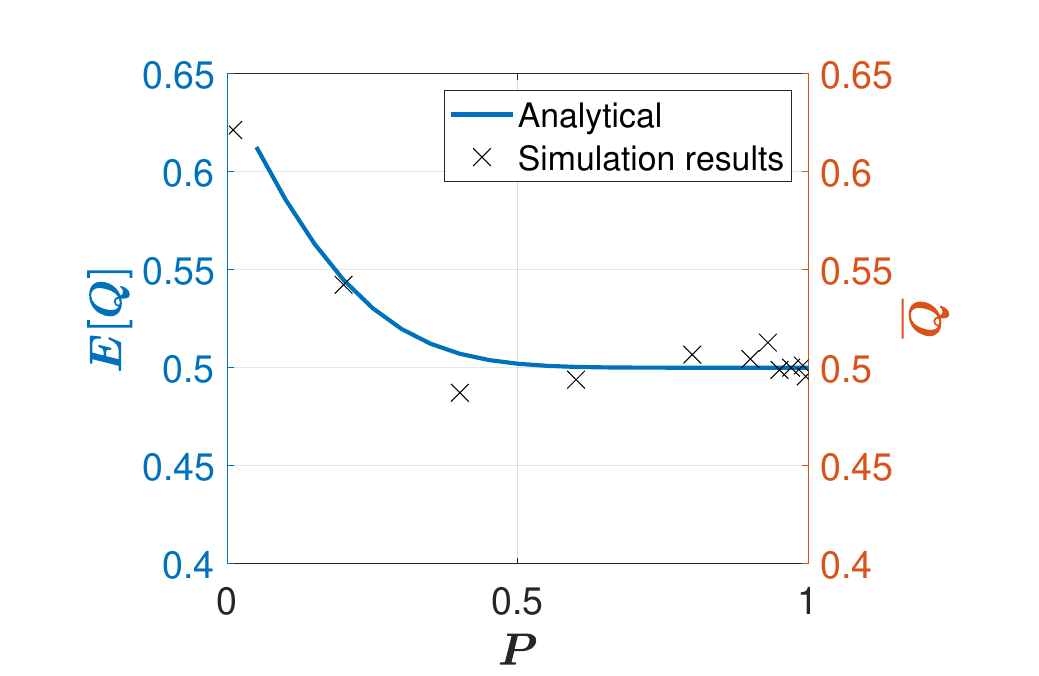}
        \caption{$K=9$}
    \end{subfigure}
    \begin{subfigure}[b]{0.32\linewidth}        
        \centering
        \includegraphics[trim = 20mm 85mm 20mm 80mm, clip,scale=0.29]{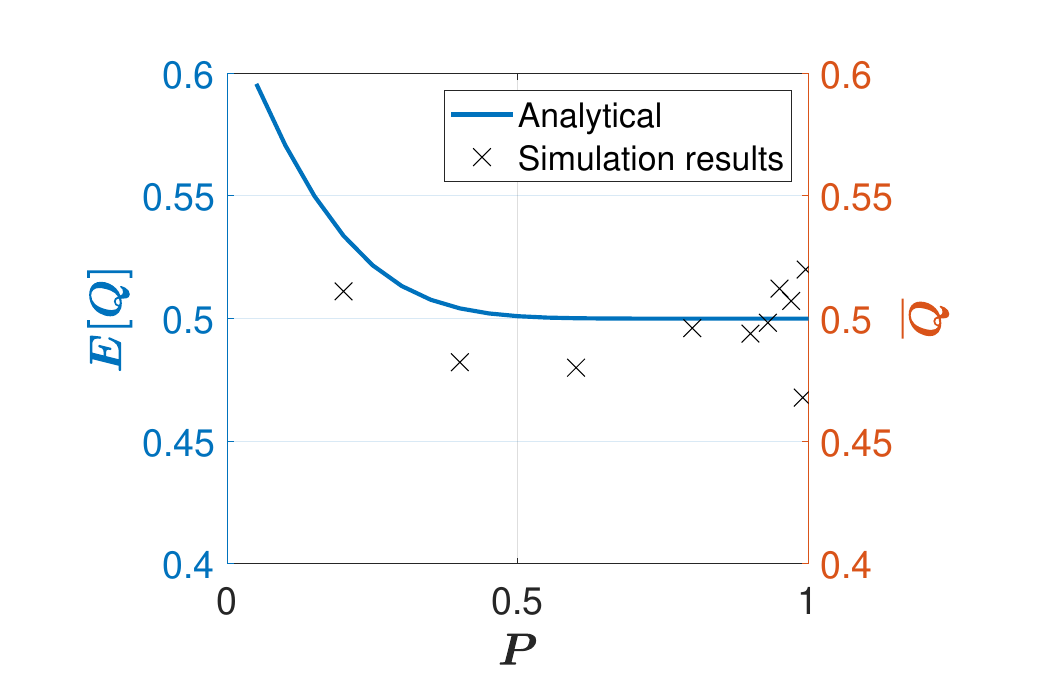}
        \caption{$K=10$}
    \end{subfigure}
    \caption{Analytical $E[Q]$ (expected memory utilization) and simulated $\overline{Q}$ (average memory utilization) as a function of $P$ for various number of neighbor nodes $K$ (3, 4, 6, 7, 9, and 10).}
    \label{fig:SR_E_Q_comparison}
\end{figure}

%
%
%
%
%
%
%
%
\section{Dual Input Repeater (DIR)}\label{sec:dual_repeater}

Obtaining the performance of DIRs ($K=2$) is desirable, since as mentioned in Section \ref{sec:intro} they will be common in most practical quantum networks. However, the $E[Q]$ analysis above~\cite{Towsley2021a,Towsley2021b} does not apply to $K=2$ and the expected memory utilization in this case goes to infinity. Fig.~\ref{sfig:k_2_memory_utilization} shows simulation results for the repeater memory utilization for $K=2$ when $P=0.8$. In this simulation instance, the maximum utilization per a specific node is $(17,25)$, where the maximum total utilization of the memory is $26$. The average utilization values per node are: $(2.21, 9.82)$, which results in $\overline{Q} = 12.03$.  
Clearly, other simulation instances will yield different average utilization statistics where inevitably, as the simulation continues, the memory will be full (as shown in~\cite{Towsley2021b}). 

The important consideration for $K=2$ repeater is not its memory utilization. Rather, \emph{it is the performance tradeoff given some memory size}. Namely, for a memory size $M$, one would like to obtain the performance degradation in terms of PR. We can now formulate the problem statement.

\begin{problem*}
    Given a memory size $M$, $R_\textrm{max}$ and performance degradation $\delta, \; 0 \leq \delta \leq 1$ relative to $R_\textrm{{max}}$, what is the minimum size of $M$ that achieves the rate $R(M)$ given by:
    
    \begin{equation}\label{eq:R_M_delta_problem}
        R(M) = (1-\delta) R_{\textrm{max}}
    \end{equation}
    
\end{problem*}

The maximum rate can be written as $\frac{qKp}{W}$ (see~\cite{Towsley2021a}), where $q$ denotes matching efficiency and $W$ denotes the number of qubits being matched. Since we have $K=2$, $W=2$ and we assume $q=1$, we obtain:
\begin{equation}
    R_{\textrm{max}} = \textrm{min}(p_1, p_2)
\end{equation}
\noindent since input probabilities may be different.

We will use DTMC to address the above problem and to obtain analytical closed form for $R(M)$. First, 'on-the-fly' pairing model will be analyzed in Section~\ref{ssec:on_the_fly}, where the memory is used only, if no pairing could be made. Then, a full-memory-usage model in which all qubits are first stored, will be discussed in Section~\ref{ssec:full_mem_usage}. The on-the-fly pairing model is simpler to analyze and will provide a close approximation to the performance of the full-memory-usage model.

\subsection{On-the-fly pairing model} \label{ssec:on_the_fly}

Since only two nodes are connected to the DIR, the memory always holds qubits only from a single input node - otherwise a pair will be created. There are two situations in which a pair will be created:

\begin{itemize}
    \item Two qubits arrive at the same time and are paired immediately.
    \item A single qubit arrives from one of the nodes. If a pairing qubit (opposite node) resides in memory, that qubit will be retrieved and a pair will be created.
\end{itemize}

 Denote the state of the a memory as $\staterep{x}{y}$, where $x$ denotes the number of qubits in memory that originate from $N_1$ and $y$ denotes the number of qubits in memory that originate from $N_2$. Clearly, either one or both of $x$ or $y$ are zero.

The state space $\boldsymbol{S}$ could be written as:


\begin{equation}
        \boldsymbol{S} =  \left\{  \staterep{0}{M},  \staterep{0}{M-1}, \staterep{0}{M-2}, \dots , \staterep{0}{1}, \staterep{0}{0},  
         \staterep{1}{0}, \dots, \staterep{M-2}{0}, \staterep{M-1}{0}, \staterep{M}{0} \right\}
\end{equation}

Let

\begin{itemize}
    \item $\alpha$ = Probability of $\{$qubit arrived from $N_1$ but no qubit arrived from $N_2\}$.
    \item $\beta$ = Probability of $\{$qubit arrived from $N_2$ but no qubit arrived from $N_1\}$.
    \item $\sigma$ = Probability of $\{$qubits arrived from both $N_1$ and $N_2\}$.
    \item $\gamma$ = Probability of $\{$no qubit arrived$\}$.
\end{itemize}

we have


\begin{equation}
        \alpha      =p_1(1-p_2), \;\; \beta       =p_2(1-p_1),\;\;
        \sigma      =p_1 p_1, \;\;     \gamma    =(1-p_1)(1-p_2)        
\end{equation}

\noindent and we define $\rho$ as 

\begin{equation}
        \rho \equiv \frac{\alpha}{\beta}        
\end{equation}

Fig.~\ref{fig:on_the_fly_DTMC_structure} depicts the DTMC structure for the on-the-fly pairing model.

\begin{figure}
    \centering
    \includegraphics[trim = 10mm 155mm 10mm 0mm, clip,scale=0.64]{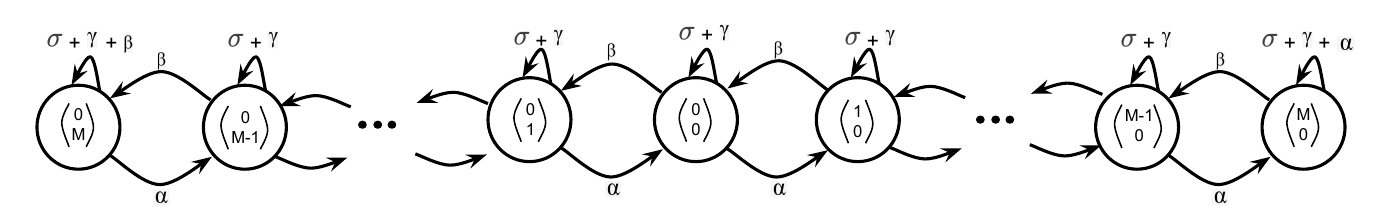}
    \caption{DTMC structure for on-the-fly pairing model, DIR.}
    \label{fig:on_the_fly_DTMC_structure}
\end{figure}

\begin{lemma} \label{lemma:on_the_fly}
For on-the-fly paring model, $R(M)$ is given by:

\begin{equation} 
    R(M) =  \pi_{\staterep{0}{0}}( \beta L_M + \alpha L_{-M} ) +  \sigma
\end{equation}

where $\pi_{\staterep{0}{0}}$ and $L_{\pm M}$ are given by~\eqref{eq:pi_00_OTF} and ~\eqref{eq:L_M_OTF}, respectively. 

\end{lemma}

\begin{proof}

    The transition probabilities for $1 \leq x \leq M$, \; $1 \leq y \leq M$:
    
    \begin{equation}
        \begin{aligned}
            P \left[\staterep{x}{0} \Big\lvert \staterep{x}{0} \right] &= P \left[\staterep{0}{y} \Big\lvert \staterep{0}{y} \right] = \gamma + \sigma, \;\; 0 \leq x,y < M\\
            P \left[\staterep{M}{0} \Big\lvert \staterep{M}{0} \right] &= \gamma + \sigma + \alpha, 
            &&P \left[\staterep{0}{M} \Big\lvert \staterep{0}{M} \right] = \gamma + \sigma + \beta\\
            P \left[\staterep{x-1}{0} \Big\lvert \staterep{x}{0} \right] &= \beta, \;\; 1 \leq x \leq M,\;\;\;
            &&P \left[\staterep{x+1}{0} \Big\lvert \staterep{x}{0} \right] = \alpha,\;\; 0 \leq x \leq M-1\\
            P \left[\staterep{0}{y-1} \Big\lvert \staterep{0}{y} \right] &= \alpha, \;\;1 \leq y \leq M,\;\;\;
            &&P \left[\staterep{0}{y+1} \Big\lvert \staterep{0}{y} \right] = \beta, \;\;0 \leq y \leq M-1
        \end{aligned}
    \end{equation}
    
    Denote the steady state distribution of a state $s$ as $\pi_s$. Note, that a pair is created anytime that a qubit is retrieved from the memory or when both qubits arrive at the same time (probability $\sigma$).
    
    We can write $R(M)$ as:
    
    \begin{equation} \label{eq:R_M_OTF} 
        R(M) = \sum_{0 < x \leq M} \pi_{\staterep{x}{0}} \beta + \sum_{0 < y \leq M} \pi_{\staterep{0}{y}} \alpha +  \sigma
    \end{equation}
    
    The associated detailed balance equation, $1 \leq x \leq M$:
    
    \begin{equation} \label{eq:detailed_bal_OTF_right}
        \pi_{\staterep{x}{0}} =   \rho \cdot \pi_{\staterep{x-1}{0}} 
    \end{equation}
    
    where $\rho \equiv \frac{\alpha}{\beta}$. Recursively substituting ~\eqref{eq:detailed_bal_OTF_right} for all $x$ we get
    
    \begin{equation} \label{eq:detailed_bal_OTF_right2}
        \pi_{\staterep{x}{0}} = \rho^{x} \pi_{\staterep{0}{0}}
    \end{equation}

    Similarly, $1 \leq y \leq M$    
       
    \begin{equation} \label{eq:detailed_bal_OTF_left}
        \pi_{\staterep{0}{y}} = \rho^{-1} \cdot \pi_{\staterep{0}{y-1}} = \rho^{-y} \pi_{\staterep{0}{0}}
    \end{equation}
    
    Since $\pi_s$ is a distribution it must be normalized:
    
    \begin{equation}
        \sum_{s \in \boldsymbol{S}} \pi_s = \sum_{1 \leq x \leq M} \pi_{\staterep{x}{0}} + \sum_{1 \leq y \leq M} \pi_{\staterep{0}{y}} + \pi_{\staterep{0}{0}}= 1
    \end{equation}
    
    or
    
    \begin{equation} \label{eq:pi_00_OTF}
        \pi_{\staterep{0}{0}} = \left(1+ \sum_{0 \leq x \leq M} \rho^{x} -1  + \sum_{0 \leq y \leq M} \rho^{-y} - 1 \right)^{-1}
    \end{equation}
    
    Let us denote
    
    \begin{equation} \label{eq:L_M_OTF}
            L_M =  \sum_{0 \leq x \leq M} \rho^{x} -1 ,\qquad
            L_{-M} = \sum_{0 \leq y \leq M} \rho^{-y} -1
    \end{equation}

    \noindent for $\rho = 1$ ($p_1=p_2$, homogeneous case), we get
    
    \begin{equation} 
            L_M =  \sum_{0 \leq x \leq M} \rho^{x} -1 = M, \qquad
            L_{-M} = \sum_{0 \leq y \leq M} \rho^{-y} -1 =M
    \end{equation}
    
    for $0 < \rho < 1$, using the fact that the sum of geometric series is $a_1 \frac{1-r^n}{1-r}$ where the $n\textrm{th}$ term of the series $a_n = a_1 r^{n-1}$ (note that we re-arranged to have $a_1 = 1$) we get
    
    \begin{equation}
            L_M =  \sum_{0 \leq x \leq M} \rho^{x} -1 = \frac{1-\rho^{M}}{1-\rho} - 1, \qquad
            L_{-M} = \sum_{0 \leq y \leq M} \rho^{-y} -1 = \frac{1-\rho^{-M}}{1-\rho^{-1}} - 1
    \end{equation}
    
    Now, ~\eqref{eq:R_M_OTF} can be written as
    
    \begin{equation} \label{eq:R_M_OTF_final}
        R(M) =  \pi_{\staterep{0}{0}}( \beta L_M + \alpha L_{-M} ) +  \sigma
    \end{equation}

\end{proof}

Recall, that we want an expression for $M(\delta)$. For $\rho = 1$ case, we can obtain a closed form expression by substituting~\eqref{eq:pi_00_OTF} into ~\eqref{eq:R_M_OTF_final}:

\begin{equation} \label{eq:M_delta_rho_1_OTF}
    M(\delta) =  \left\lceil \frac{1}{2} \frac{(\delta - 1)R + \sigma}{(1-\delta)R - \beta - \sigma} \right\rceil
\end{equation}

For $0 < \rho < 1$,  $M(\delta)$ is obtained numerically. Since $M(\delta)$ represents memory size, we have to take the upper bound to the closest integer. Results will be discussed in Section \ref{ssec:full_mem_usage}.

\subsection{full-memory-usage model} \label{ssec:full_mem_usage}

Here, we assume that every qubit that arrives is stored in the memory. Practically, this must be done since not all qubits will arrive at the same time (due to different distances from the repeater, probability of arrival and origin node BP generation rate). A pair will be created if matching could be made from qubits that reside in the memory. Fig. ~\ref{fig:full_memory_DTMC_structure} depicts the DTMC structure for the full-memory-usage model. 

\begin{figure}
    \centering
    \includegraphics[trim = 0mm 140mm 0mm 0mm, clip,scale=0.62]{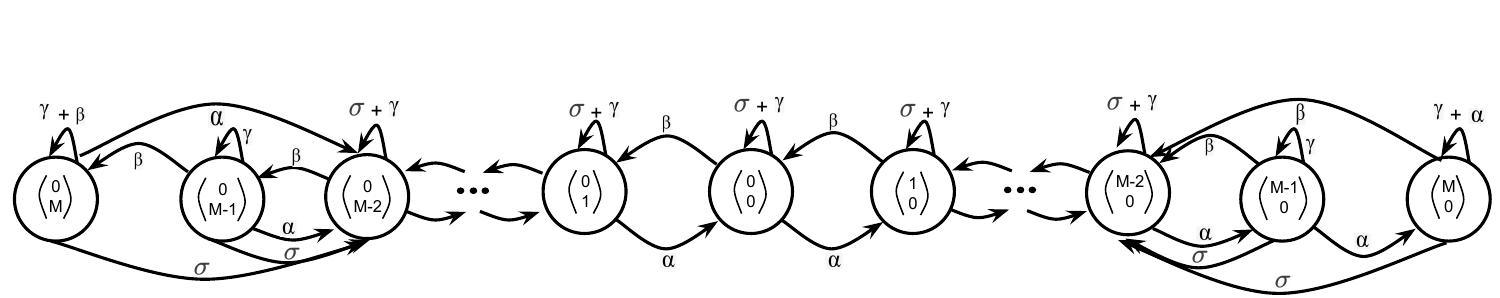} 
    \caption{DTMC structure for full-memory-usage pairing model, DIR.}
    \label{fig:full_memory_DTMC_structure}
\end{figure}

Note, that there are few differences comparing to the on-the-fly model. We describe them below for a memory that is full or almost full by qubits originated from $N_1$. Similarly, the same claims are true (symmetrical) for a memory that is full or almost full by qubits originated from $N_2$.

\begin{itemize}
    \item States depicted are steady states: any transients (see below) have already occurred.
    \item If the memory is full from qubits originated from $N_1$:
    \begin{itemize}
        \item A new input from either $N_1$ or $N_2$ replaces the oldest qubit currently resides in the memory. If the new input is from $N_2$, a matching is made and two memory slots become empty. If the new input is from $N_1$, the memory stays full.
        \item If two inputs arrive (from $N_1$ and $N_2$), they replace the two oldest qubits currently resides in the memory. A matching is made and two memory slots become empty.
    \end{itemize}    
    \item If the memory has only one empty slot (single qubit storage), where all $M-1$ slots are occupied by $N_1$ qubits:
    \begin{itemize}
        \item A new input from $N_2$ is placed in the memory followed by a matching and two memory slots become empty. Alternatively, a new input from $N_1$ is placed in the memory and the memory becomes full.
        \item If two inputs arrive (from $N_1$ and $N_2$), one of them replaces the oldest qubit currently resides in the memory. A matching is made and two memory slots become empty.
    \end{itemize}    
\end{itemize}

\begin{lemma} \label{lemma:full_mem_usage}
For full-memory-usage model, $R(M)$ is given by:

\begin{equation}
        R(M) = \pi_{\staterep{0}{0}} \Big[ (\beta + \sigma)( L_{M-2} +  C_1 +  C_2) 
         + (\alpha + \sigma) (L_{-M+2} + D_1 +  D_2) + \sigma \Big]
\end{equation}

where $\pi_{\staterep{0}{0}}$ and $L_{\pm (M-2)}$ are given by ~\eqref{eq:pi_00_full_mem_final} and~\eqref{eq:parameters_full_mem}, respectively. 

\end{lemma}

\begin{proof}
    The transition probabilities:
    
    \begin{equation}
        \begin{aligned}
            P \left[\staterep{x}{0} \Big\lvert \staterep{x}{0} \right] & = \gamma + \sigma, \;\;\; 0 \leq x \leq M-2,  
            &&P \left[\staterep{M-1}{0} \Big\lvert \staterep{M-1}{0} \right] =  \gamma \\
            P \left[\staterep{x-1}{0} \Big\lvert \staterep{x}{0} \right] &= \beta, \;\;\; 1 \leq x \leq M-1,
            &&P \left[\staterep{M}{0} \Big\lvert \staterep{M}{0} \right] = \gamma + \alpha \\  
            P \left[\staterep{x+1}{0} \Big\lvert \staterep{x}{0} \right] &= \alpha,\;\;\; 0 \leq x \leq M-1,
            &&P \left[\staterep{M-2}{0} \Big\lvert \staterep{M}{0} \right] = \beta, \\
            P \left[\staterep{M-2}{0} \Big\lvert \staterep{M-1}{0} \right] &=  P \left[\staterep{M-2}{0} \Big\lvert \staterep{M}{0} \right] = \sigma\\
        \end{aligned}
    \end{equation}

    and

    \begin{equation}
        \begin{aligned}   
            P \left[\staterep{0}{y} \Big\lvert \staterep{0}{y} \right] & = \gamma + \sigma, \;\;\; 0 \leq y \leq M-2,
            &&P \left[\staterep{0}{M-1} \Big\lvert \staterep{0}{M-1} \right] = \gamma \\
            P \left[\staterep{0}{y-1} \Big\lvert \staterep{0}{y} \right] &= \alpha, \;\;\; 1 \leq y \leq M-1,
            &&P \left[\staterep{0}{M} \Big\lvert \staterep{0}{M} \right] = \gamma + \beta \\ 
            P \left[\staterep{0}{y+1} \Big\lvert \staterep{0}{y} \right] &= \beta, \;\;\; 0 \leq y \leq M-1,
            &&P \left[\staterep{0}{M-2} \Big\lvert \staterep{0}{M} \right] = \alpha, \\            
            P \left[\staterep{0}{M-2} \Big\lvert \staterep{0}{M-1} \right] &= P \left[\staterep{0}{M-2} \Big\lvert \staterep{0}{M} \right] = \sigma\\
        \end{aligned}
    \end{equation}
    
    $R(M)$ is defined as before, ~\eqref{eq:R_M_OTF}. The associated detailed balance equations:
    
    \begin{equation}  \label{eq:detailed_bal_full_mem_right}
        \begin{aligned}
            \pi_{\staterep{x}{0}} &= \rho^{x} \pi_{\staterep{0}{0}}, \;\;\; 1 \leq x \leq M-2 \\
            \pi_{\staterep{M-1}{0}} &= \left(\frac{\alpha}{\sigma + \beta}\right) \rho^{M-2} \pi_{\staterep{0}{0}},
            &&\pi_{\staterep{M}{0}} = \left(\frac{\alpha^{2}}{\sigma + \beta}\right) \rho^{M-2} \pi_{\staterep{0}{0}}
        \end{aligned}
    \end{equation}
    
    and
    
    \begin{equation} \label{eq:detailed_bal_full_mem_left}
        \begin{aligned}
            \pi_{\staterep{0}{y}} &= \rho^{-y} \pi_{\staterep{0}{0}}, \, 1 \leq y \leq M-2 \\
            \pi_{\staterep{0}{M-1}} &= \left(\frac{\beta}{\sigma + \alpha}\right) \rho^{-M+2} \pi_{\staterep{0}{0}},
            &&\pi_{\staterep{0}{M}} = \left(\frac{\beta^{2}}{\sigma + \alpha}\right) \rho^{-M+2} \pi_{\staterep{0}{0}}
        \end{aligned}
    \end{equation}
    
    We can find $\pi_{\staterep{0}{0}}$ using normalization as before:

    \begin{equation} \label{eq:pi_00_full_mem_part1}
            \sum_{s \in \boldsymbol{S}} \pi_s = \sum_{1 \leq x \leq M} \pi_{\staterep{x}{0}} + \sum_{1 \leq y \leq M} \pi_{\staterep{0}{y}} 
           + \pi_{\staterep{0}{0}}= 1
    \end{equation}
    
    Using~\eqref{eq:detailed_bal_full_mem_right}, ~\eqref{eq:detailed_bal_full_mem_left} we can write:

    \begin{equation}
            \sum_{1 \leq x \leq M} \pi_{\staterep{x}{0}} = \sum_{1 \leq x \leq M-2} \rho^{x} \pi_{\staterep{0}{0}} 
            + \left(\frac{\alpha}{\sigma + \beta}\right) \rho^{M-2} \pi_{\staterep{0}{0}}  +
            \left(\frac{\alpha^{2}}{\sigma + \beta}\right) \rho^{M-2} \pi_{\staterep{0}{0}}
    \end{equation}
    
    \begin{equation}
        \begin{aligned}
            \sum_{1 \leq x \leq M} \pi_{\staterep{x}{0}} &=  \pi_{\staterep{0}{0}} (L_{M-2} +C_1 + C_2)
        \end{aligned}
    \end{equation}
    
    similarly
    
    \begin{equation}
        \begin{aligned}
            \sum_{1 \leq y \leq M} \pi_{\staterep{0}{y}} &=  \pi_{\staterep{0}{0}} (L_{-M+2} +D_1 + D_2)
        \end{aligned}
    \end{equation}
    
    Substituting in~\eqref{eq:pi_00_full_mem_part1} we get
    
    \begin{equation} \label{eq:pi_00_full_mem_final}
        \begin{aligned}
            \pi_{\staterep{0}{0}} = \Big(L_{M-2} +C_1 + C_2 + L_{-M+2} +D_1 + D_2 + 1 \Big) ^{-1}
        \end{aligned}
    \end{equation}
    
    where we denoted    

    \begin{equation} \label{eq:parameters_full_mem}
        \begin{aligned}
            L_{M-2} &= \sum_{0 \leq x \leq M-2} \rho^{x} - 1, 
            &&L_{-M+2} = \sum_{0 \leq y \leq M-2} \rho^{-y} - 1 \\
            C_1 &= \left(\frac{\alpha}{\sigma + \beta}\right) \rho^{M-2}, 
            &&C_2 = \left(\frac{\alpha^{2}}{\sigma + \beta}\right) \rho^{M-2} \\
            D_1 &= \left(\frac{\beta}{\sigma + \alpha}\right) \rho^{-M+2}, 
            &&D_2 = \left(\frac{\beta^{2}}{\sigma + \alpha}\right) \rho^{-M+2} \\
        \end{aligned}
    \end{equation}

    For $\rho = 1$, ~\eqref{eq:parameters_full_mem} becomes
    
    \begin{equation}
        \begin{aligned}
            L_{M-2} &= M-2, 
            && L_{-M+2} = M-2\\
            C_1 &=  \left(\frac{\alpha}{\sigma + \beta}\right), 
            &&C_2 =  \left(\frac{\alpha^{2}}{\sigma + \beta}\right)  \\
            D_1 &=  \left(\frac{\beta}{\sigma + \alpha}\right), 
            && D_2=  \left(\frac{\beta^{2}}{\sigma + \alpha}\right)\\
        \end{aligned}
    \end{equation}
    
    and for $0 < \rho < 1$ we have for $L_{M-2}$ and $L_{-M+2}$:
    
    \begin{equation}
        L_{M-2} =  \frac{1-\rho^{M-2}}{1-\rho} - 1, \qquad
        L_{-M+2} =  \frac{1-\rho^{-M+2}}{1-\rho^{-1}} - 1
    \end{equation}
    
    We can now write~\eqref{eq:R_M_OTF} as
    
    \begin{equation} \label{eq:R_M_full_mem_final}
            R(M) = \pi_{\staterep{0}{0}} \Big[ (\beta + \sigma)( L_{M-2} +  C_1 +  C_2) 
              + (\alpha + \sigma) (L_{-M+2} + D_1 +  D_2) + \sigma \Big]
    \end{equation}

    $M(\delta)$ is obtained numerically for a given $\rho$.
    
\end{proof}

The PR performance is collected by simulation that uses the full-memory-usage model (e.g., all qubits must be stored before pairing occurs).
Simulation data collected using 30 runs for each $M$, 10,000 iteration per run and an equal length distance between the two nodes and the repeater. 
Fig.~\ref{sfig:R_M_comparison} shows the simulated PR (mean PR) for memory sizes between 4 to 30 and the analytical results, $R(M)$, for both on-the-fly model~\eqref{eq:R_M_OTF_final} and full-memory-usage model~\eqref{eq:R_M_full_mem_final}. A probability of $P=0.8$ is shown (corresponds to a fiber length of $4,630$\thinspace{m}). The simulation and the full-memory-usage model analytical results for $R(M)$ are within $0.3\%$ of error. On-the-fly $R(M)$ provides a good approximation at high memory sizes, but as expected, its predicted performance in those cases is higher than what is achieved. 
Similarly, we collect PR performance for several selections of $P$ (starting from low probability $P=0.1$, corresponds to long length fiber, and up to high probability $P=0.9$, corresponds to short length fiber). Fig~\ref{sfig:SR_P_comparison} shows the mean PR (obtained as in Fig.~\ref{sfig:R_M_comparison}) and full-memory-usage model~\eqref{eq:R_M_full_mem_final} results for memory sizes between 4 to 30. Results of the simulation and the full-memory-usage model are in agreement for all selected $P$ (within $0.3\%$ of error).

\begin{figure}
    \centering
    \begin{subfigure}[b]{0.47\linewidth}        
        \centering
        \includegraphics[trim = 0mm 68mm 10mm 65mm, clip,scale=0.31]{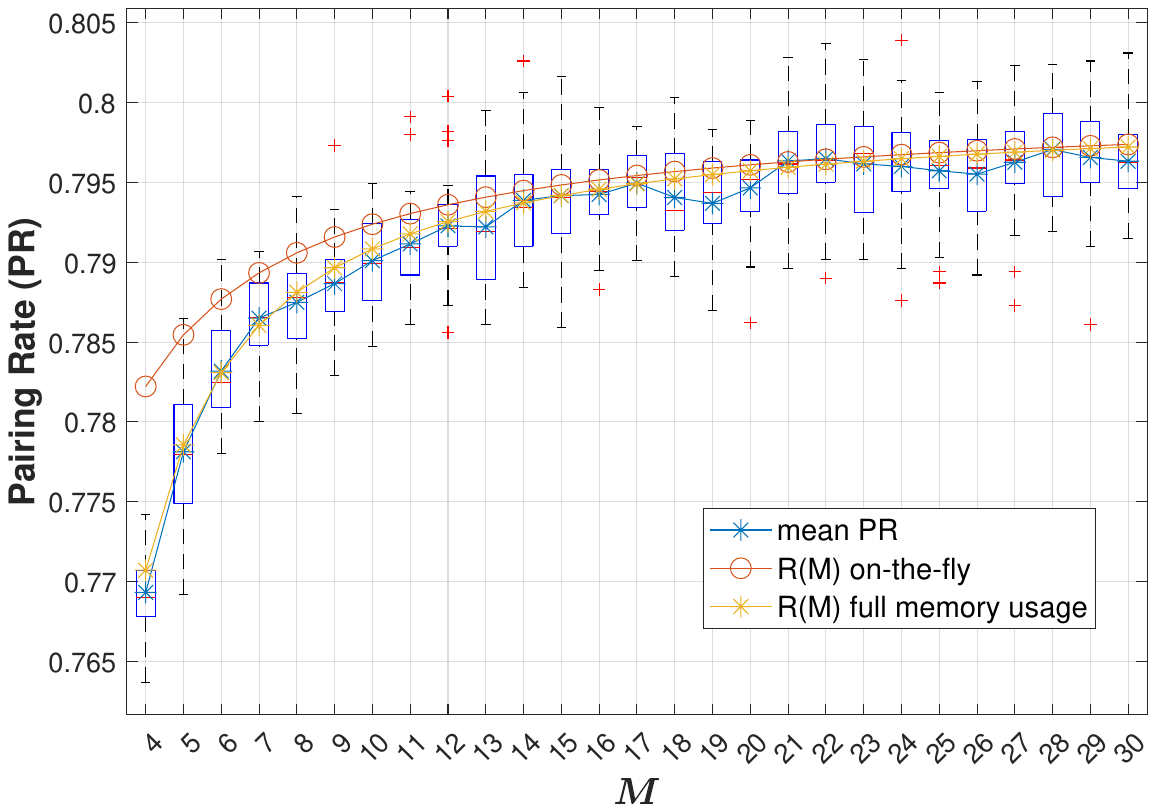}     
        \caption{}
        \label{sfig:R_M_comparison}
    \end{subfigure}
    \begin{subfigure}[b]{0.47\linewidth}        
        \centering
        \includegraphics[trim = 0mm 0mm 0mm 0mm, clip,scale=0.26]{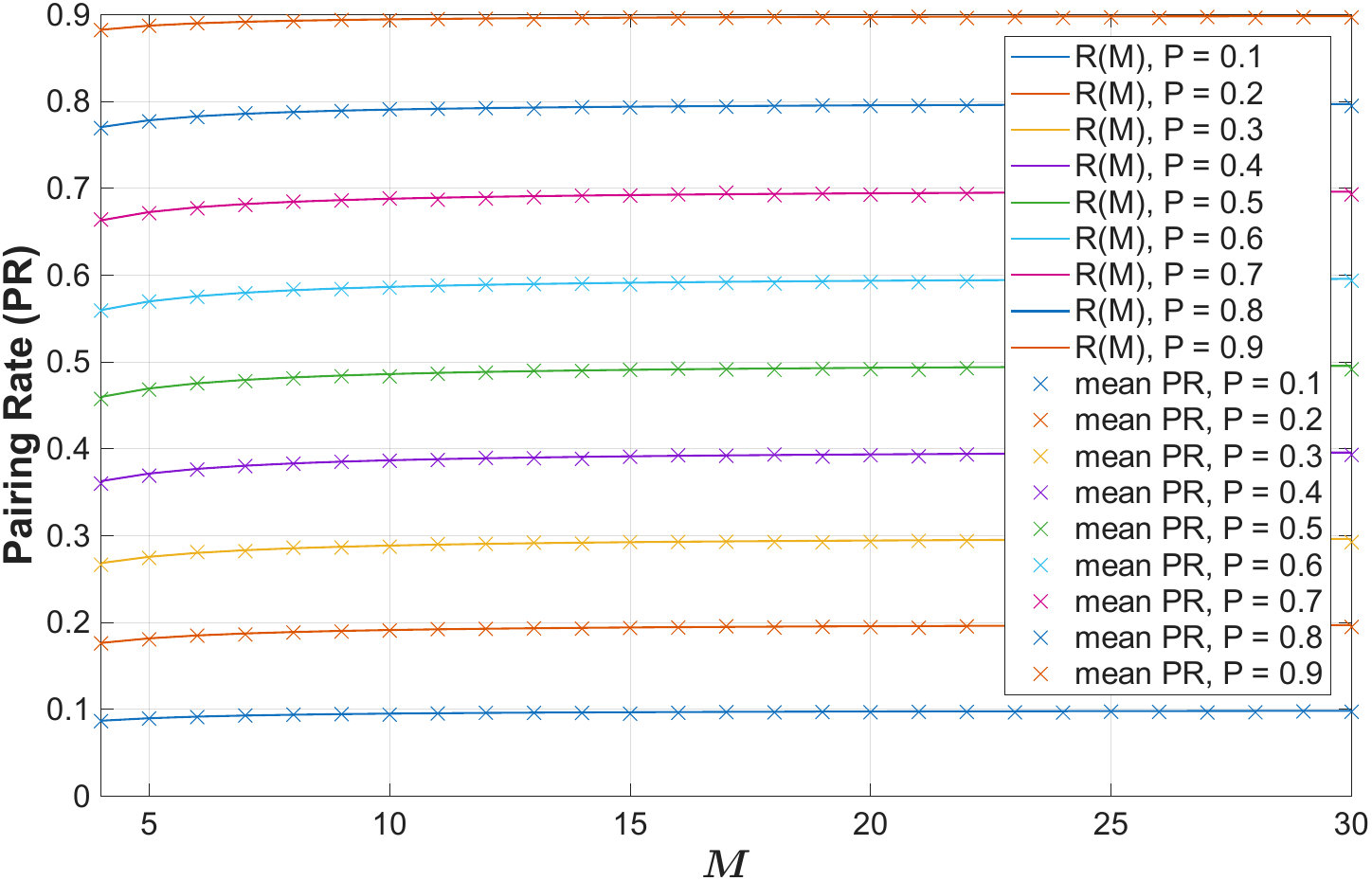}  
        \caption{}
        \label{sfig:SR_P_comparison}
    \end{subfigure}

    \caption{(\subref{sfig:R_M_comparison}) PR performance for incoming memory sizes between $4$ to $30$: simulation results (mean PR) and analytical $R(M)$, calculated for two pairing models. The simulation always uses the memory to store any incoming qubit before any operation. Simulation data collected using 30 runs for each $M$, 10,000 iteration per run, $P=0.8$. Box notation: bottom and top are the $25^\textrm{th}$ and $75^\textrm{th}$ percentiles, red line is the median, whiskers show the distance between each end of the box to maximum or minimum values and red + sign indicates outliers (distance that is more than 1.5 of whiskers). (\subref{sfig:SR_P_comparison}) Mean PR and analytical full-memory-usage model $R(M)$ results for different selection of $P$.}
\end{figure}

For obtaining $M(\delta)$, we solve \eqref{eq:R_M_delta_problem} numerically for $\delta$ using $R(M)$ for the two pairing models given in \eqref{eq:R_M_OTF_final} and \eqref{eq:R_M_full_mem_final}.
Fig.~\ref{fig:M_delta} shows the required $M$ value as a function of $\delta$ (before and after applying ceiling) for several values of $p_1$ and $p_2$ for the two pairing models. When the difference between $p_1$ and $p_2$ is large, smaller memory size is sufficient to get most of the performance. For instance, when $p_1=0.2, p_2=0.8$, $M=3$ is needed for $3\%$ performance degradation ($\delta = 0.03$) comparing to $M=4$ in $p_1=0.4, p_2=0.8$ case. When $P=p_1=p_2$, lower values of $P$ require more memory. Note, that since the memory resource is expensive, a trade-off of few percentage degradation is sufficient to reduce the memory size significantly.

\begin{figure}
    \centering
    \begin{subfigure}[b]{0.32\linewidth}        
        \centering
        \includegraphics[trim = 20mm 80mm 20mm 80mm, clip,scale=0.29]{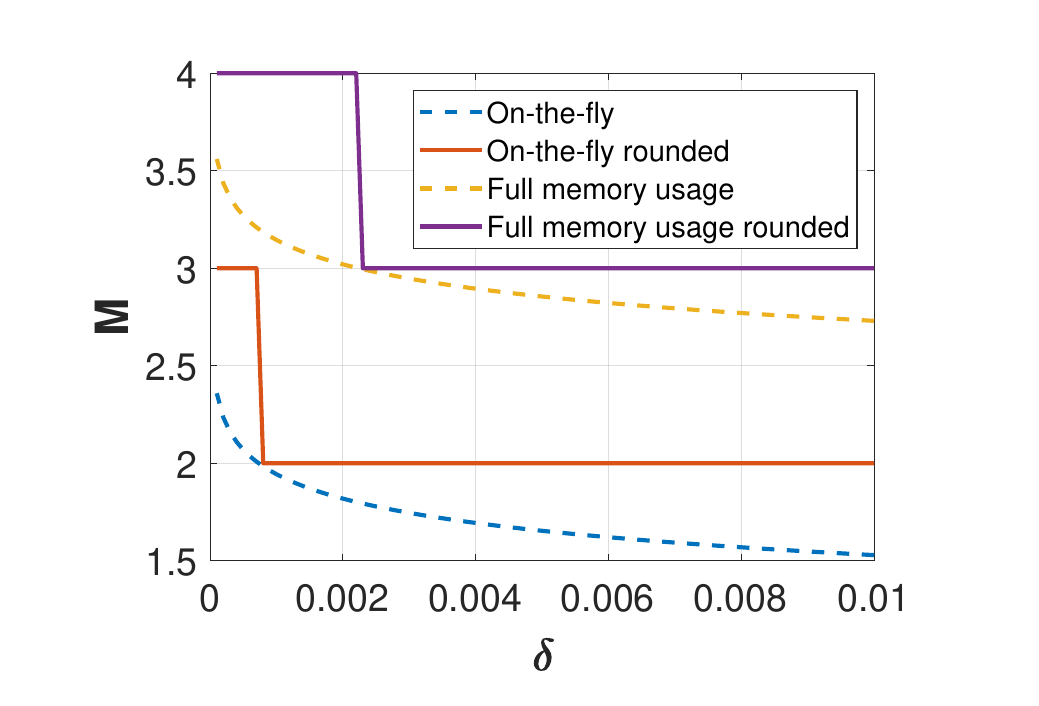}   
        \caption{$p_1 = 0.2$, $p_2 = 0.8$}
    \end{subfigure}
    \begin{subfigure}[b]{0.32\linewidth}        
        \centering
        \includegraphics[trim = 20mm 80mm 20mm 80mm, clip,scale=0.29]{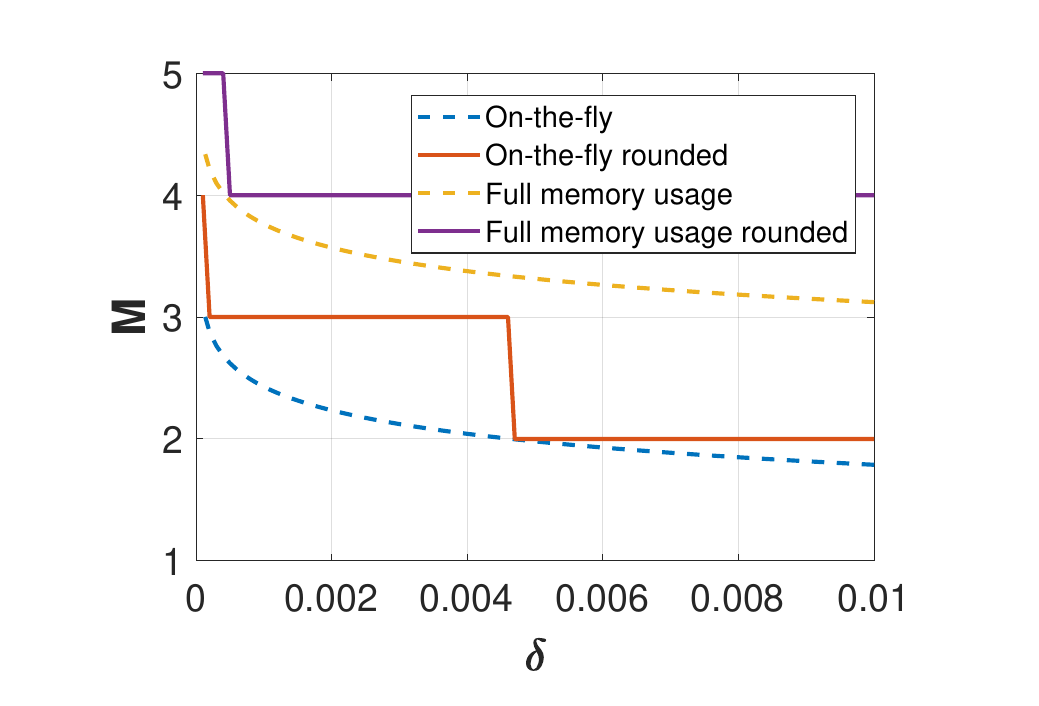}
        \caption{$p_1 = 0.4$, $p_2 = 0.8$}
    \end{subfigure}
    \begin{subfigure}[b]{0.32\linewidth}        
        \centering
        \includegraphics[trim = 20mm 80mm 20mm 80mm, clip,scale=0.29]{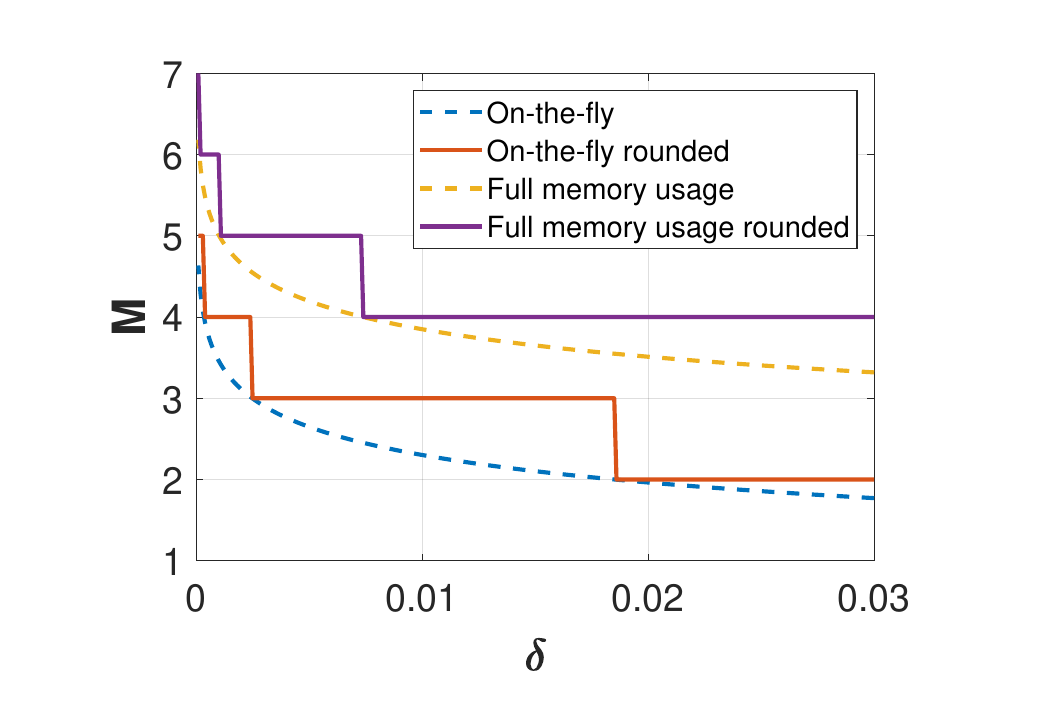}
        \caption{$p_1 = 0.6$, $p_2 = 0.8$}
    \end{subfigure}
    \medskip
    \centering
    \begin{subfigure}[b]{0.32\linewidth}        
        \centering
        \includegraphics[trim = 20mm 80mm 20mm 80mm, clip,scale=0.29]{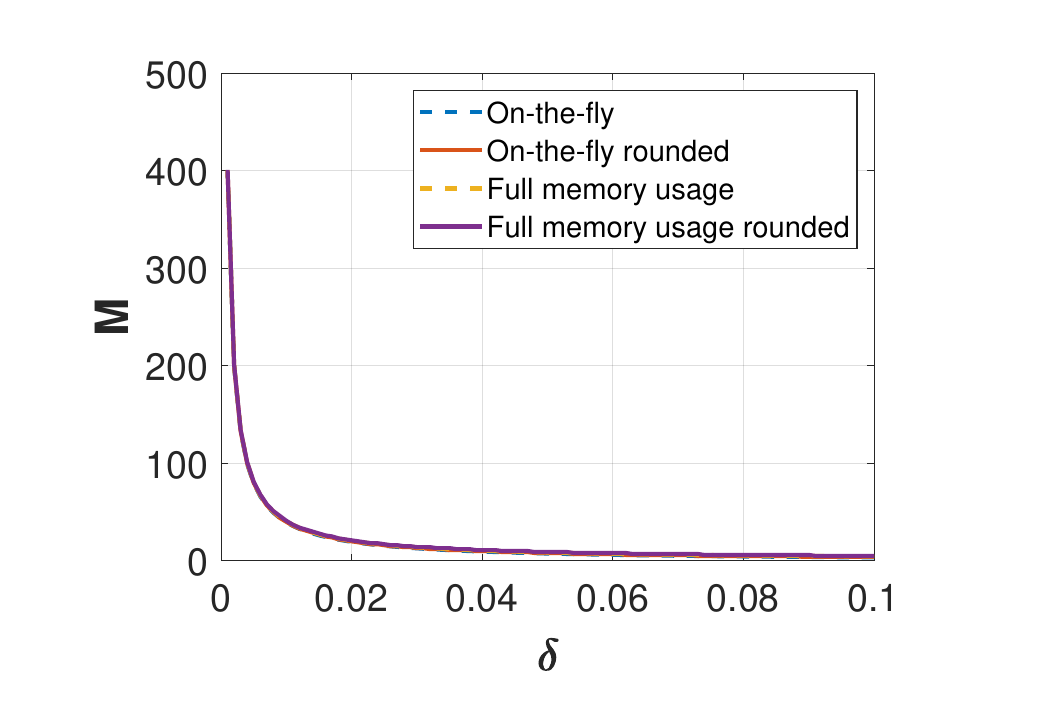}   
        \caption{$p_1 = 0.2$, $p_2 = 0.2$}
    \end{subfigure}
    \begin{subfigure}[b]{0.32\linewidth}        
        \centering
        \includegraphics[trim = 20mm 80mm 20mm 80mm, clip,scale=0.29]{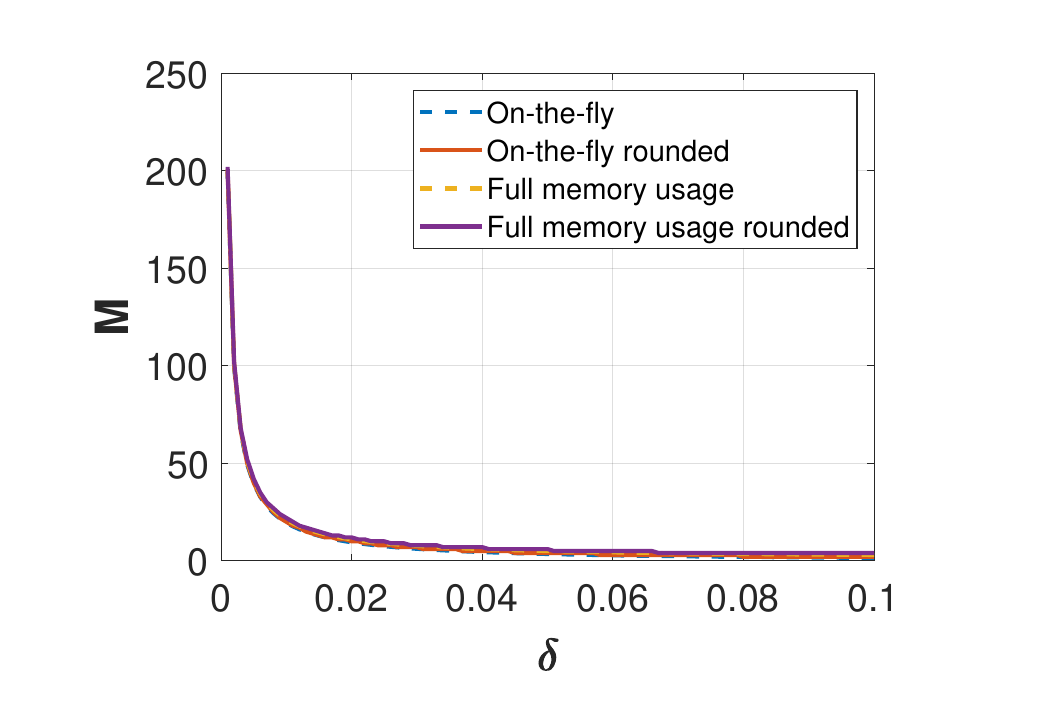}
        \caption{$p_1 = 0.6$, $p_2 = 0.6$}
    \end{subfigure}
    \begin{subfigure}[b]{0.32\linewidth}        
        \centering
        \includegraphics[trim = 20mm 80mm 20mm 80mm, clip,scale=0.29]{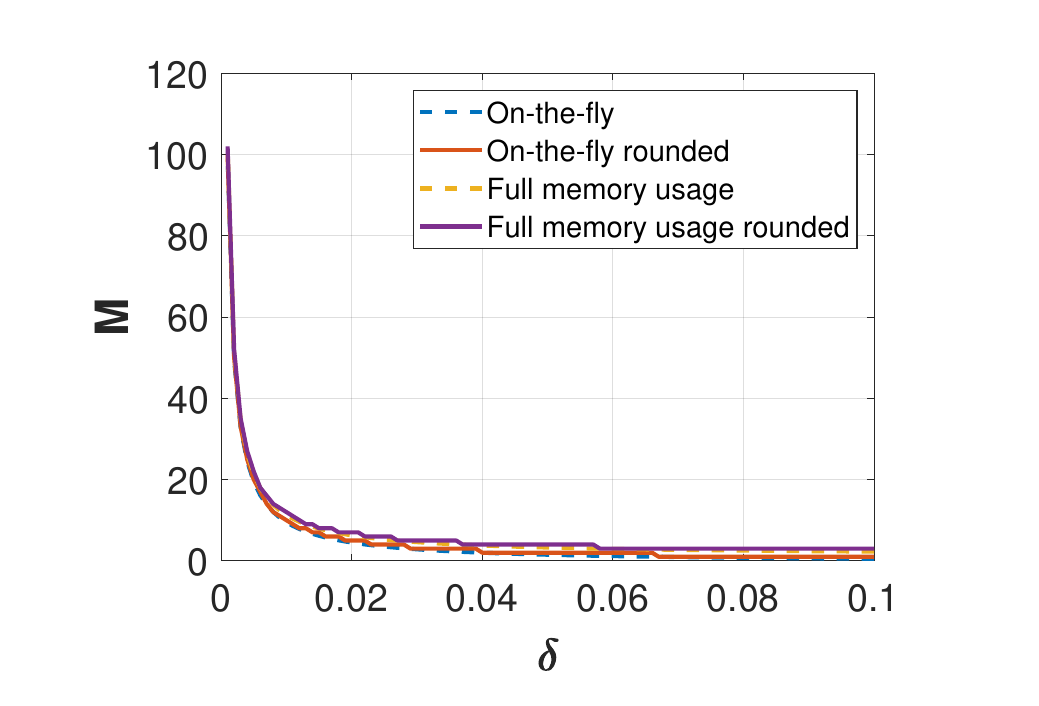}
        \caption{$p_1 = 0.8$, $p_2 = 0.8$}
    \end{subfigure}
    \caption{Numerical results of $M$ as a function of $\delta$ for several values of $p_1$ and $p_2$. Results are given before and after applying rounding (ceiling) for the two pairing models.}
    \label{fig:M_delta}
\end{figure}

%
%
%
%
%
%
%
%
\section{Quantum network performance} \label{sec:quantum_network_performance}
In this section, we first discuss a simple swapping policy followed by quantum error settings. Performance evaluation of several line topology networks (synthetic and practical), vary by fiber distances, number of repeaters and quantum error settings, are presented.

\subsection{Repeater swapping policy} \label{ssec:repeater_swapping_policy}
As mentioned in Section~\ref{sssec:general_architecture}, pair selection is done according to the repeater's policy - which may require local or global based optimization. For line topology, we use a simple, greedy swapping policy that enables us to focus on EPR and Fidelity metrics.

A repeater $r$ lists all $V$ valid qubits for swapping: qubits that are entangled with different nodes. Note, that some qubits are incoming qubits and some outgoing (part of BP the repeater generated). Let $\{Q_1 \dots Q_V \}$ denote the valid qubits list. For each valid qubit $Q^{(w)}_i$, there is associated distance, $d_{wr}$ and elapsed time $t_i$ since that qubit was stored in memory. $w, \; 1 \leq w \leq N, \; w \neq r$ denotes the entangled node (with node $r$) of this qubit.
Given two valid qubits for swapping, $Q^{(l)}_i, Q^{(c)}_j$, we define the pre-swap distance as $d_{lr} + d_{cr}$. Post swap, the entangled qubit is shared across nodes $l,c$ with distance $d_{lc}$. As the network topology is a line, we always seek to increase the post swap distance (until we get $d_{AB}$): $d_{lc} > d_{lr} + d_{cr}$ or we will not execute the swapping. 

On every time unit, each node generates several BP (same as the number of its neighbors - one for edge node or two for repeater node). Each repeater node may try to perform swapping so EP could eventually be shared across the end nodes (Alice, $A$, and Bob, $B$). We suggest the following swap policy Algorithm~\ref{alg:swapping} for maximizing the EPR. 

\begin{algorithm}
\scriptsize
\caption{Repeater $r$ swapping policy}\label{alg:swapping}
\KwData{$\{Q_1 \dots Q_V \}$, $D$}
\KwResult{Matched pairs of qubits for swapping}
\While{$V>1$}{
    
    \uIf{$\exists \{Q^{(A)}_i, Q^{(B)}_j \} \in V$}{
        select $\{Q^{(A)}_i, Q^{(B)}_j \}$ as a match \;
        remove $\{Q^{(A)}_i, Q^{(B)}_j \}$ from $\{Q_1 \dots Q_V \}$
    }
    \uElseIf{$\exists Q^{(A)}_i \in V$ or $\exists Q^{(B)}_j  \in V$}{
        choose qubit $Q^{(A)}_i$ or $Q^{(B)}_j$ \;
        choose the second qubit $Q^{(w)}_z$ that maximize post swap distance \;
        remove selected qubits from $\{Q_1 \dots Q_V \}$
    }
    \uElse{
        match pair of qubits whose post swap distance is max (elapsed time breaks tie) \;
         remove selected qubits from $\{Q_1 \dots Q_V \}$       
    }

    update $\{Q_1 \dots Q_V \}$
}

\end{algorithm}

If more than one valid qubit exists, we have three options:
\begin{enumerate*}[label=(\roman*)]
    \item repeater $r$ is entangled with both Alice and Bob and these qubits are valid for swapping. In this case, we pair them for swapping,
    \item only one valid qubit is entangled with Alice or Bob. We choose this qubit and match it to a valid qubit that maximizes post swapping distance. Note, that if more than one option exists, then elapsed time breaks tie (random selection among identical elapsed times),
    \item there are no valid qubits that are entangled with Alice or Bob. In this case, we search for a pair that maximizes post swapping distance.
\end{enumerate*}
After every step, the selected qubits are removed from the valid qubits list and matching continues until $V \le 1$. 
%

\subsection{SCY-QNet network graph} \label{ssec:SCY_network_graph}
We consider the SCY-QNet network depicted in Fig.~\ref{fig:SCY_google_earth} whose subgraphs that connect Columbia University (CBU) to Brookhaven National Laboratory (BNL) are depicted in Fig.~\ref{sfig:SCY_line_3R_sketch_option_1} (path from BNL to CBU via Stony Brook University (SBU), Westbury (WBY), and Brooklyn (BKN)) and Fig.~\ref{sfig:SCY_line_3R_sketch_option_2} (path from BNL to CBU via Commack (CMK), Garden City (GCY), and BKN). BKN, SBU, and WBY repeaters are observed as DIR with their associated distances. For example, BKN as DIR has fiber length of 35,000\thinspace{m} and 62,000\thinspace{m} to neighbor nodes.  

\begin{figure}
    \centering
    \begin{subfigure}[b]{0.47\linewidth}        
        \centering
        \includegraphics[trim = 30mm 150mm 30mm 0mm, clip,scale=0.25]{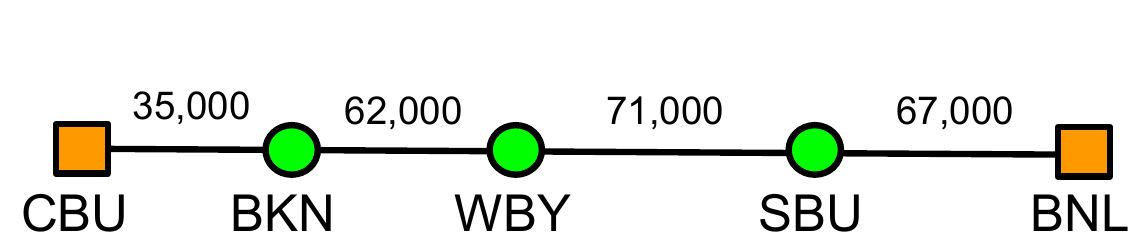}     
        \caption{}
        \label{sfig:SCY_line_3R_sketch_option_1}
    \end{subfigure}
    \begin{subfigure}[b]{0.47\linewidth}        
        \centering
        \includegraphics[trim = 30mm 150mm 30mm 0mm, clip,scale=0.25]{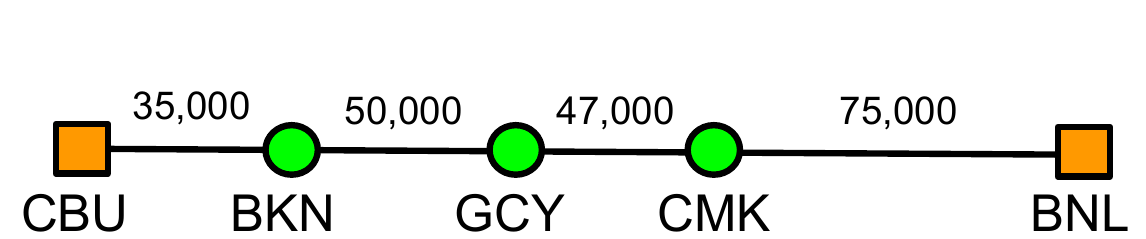}  
        \caption{}
        \label{sfig:SCY_line_3R_sketch_option_2}
    \end{subfigure}

    \caption{Subgraph CBU-BNL of SCY-QNet depicted in Fig.~\ref{fig:SCY_google_earth}: three repeater line network. Distances (weights) are given in meters. (\subref{sfig:SCY_line_3R_sketch_option_1}) path option 1. (\subref{sfig:SCY_line_3R_sketch_option_2}) path option 2.}
    \label{fig:SCY_line_3R_sketch}
\end{figure}

\subsection{Quantum error settings} \label{ssec:quantum_error_settings}

\subsubsection{$\lambda_{\textrm{PF}}$ settings} \label{sssec:lambda_PF_settings}
In continue to our discussion in Section~\ref{sssec:quantum_channels}, phase noise optimization and phase compensation impact $\lambda_{\textrm{PF}}$, where optimized fiber setup could achieve values of $\lambda_{\textrm{PF}}(25)=0.995$~\cite{Min2008} for fiber length of 25\thinspace{km}. We assume several values for $\lambda_{\textrm{PF}}(25)$ ($0.95$, $0.8$ and $0.6$) to reflect optimized fiber infrastructure down to high phase noise connections. The values of $\lambda_{\textrm{PF}}$ for different fiber length $d_{ij}$ (in \thinspace{km}) are scaled from $\lambda_{\textrm{PF}}(25)$: $\lambda^{ij}_{\textrm{PF}} = e^{d_{ij} \cdot ln(\lambda_{\textrm{PF}}(25))/25}$, as we assume the same relative phase noise along the fiber. For example, $\lambda^{ij}_{\textrm{PF}}$ values are given in Table ~\ref{tab:lambda_PF} for SCY-QNet path option 1 (Fig.~\ref{sfig:SCY_line_3R_sketch_option_1}) as well as for other selected fiber lengths.

\begin{table}[h]
    \centering
    \caption{$\lambda_{\textrm{PF}}$ for selected fiber length $d_{ij}$ (in \thinspace{km}), scaled from $\lambda_{\textrm{PF}}(25)$}
    \label{tab:lambda_PF}
    \scalebox{0.6}{
    \begin{tabular}{|c|c|c|c|c|c|c|}
        \toprule
        \diagbox[width=\dimexpr \textwidth/8+2\tabcolsep\relax, height=1cm]{ $\lambda_{\textrm{PF}}(25)$ }{$d_{ij}$(\thinspace{km})}
                   & 4.63 & 9.26 & 35 & 62 & 71 & 67 \\
        \hline
        \midrule 
        0.95 & 0.9905 & 0.9812 & 0.9307 & 0.8806 & 0.8716 & 0.8644 \\
        0.80 & 0.9595 & 0.9207 & 0.7317 & 0.5750 & 0.5499 & 0.5306 \\
        0.60 & 0.9097 & 0.8276 & 0.4891 & 0.2817 & 0.2544 & 0.2344 \\
        \bottomrule
    \end{tabular}}
\end{table}

\subsubsection{$\Gamma_1$ and $\Gamma_2$ settings} \label{sssec:G1_G2_settings}
EIT based quantum memory~\cite{Fleischhauer2005, Liu2001} architecture can differ (e.g., warm vapor, cold atoms) and lead to various $T_1$ and $T_2$ values~\cite{Phillips2001, Eisaman2005}. In order to demonstrate very stable rates (high $T_1$ and $T_2$ values) to non stable rates, we vary $\Gamma_1$ and $\Gamma_2$ pairs as follows: $[\Gamma_1, \Gamma_2]$ \thinspace{Hz} = $[10,000,100,000], [333,3,333],[100,333],[20,50], [20,200],[2,10]$.

\subsection{Performance of line topology network with a single repeater} \label{ssec:SR_network_performance}
Fig.~\ref{fig:line_1R_EPR_Fidelity} shows EPR and average fidelity as a function of $M$ for the selected $[\Gamma_1,\Gamma_2]$ and $\lambda_{\textrm{PF}}(25)$ quantum error settings. Various fiber lengths are considered: common $4,630$\thinspace{m} ($P=0.8$) fiber length and 1:2 ratio ($4,630$\thinspace{m} and $9,260$\thinspace{m}) fiber lengths. BPs are generated every $10$ \thinspace{$\mu s$}.

\noindent For common fiber length, Fig.~\ref{sfig:line_1_1_EPR} and Fig.~\ref{sfig:line_1_1_Fidelity} show the EPR and average fidelity, respectively. EPR performance is relatively close across $[\Gamma_1,\Gamma_2]$ and $\lambda_{\textrm{PF}}(25)$, as expected, since quantum noise impacts the fidelity but not the achieved EPR. We identify several observations for EPR: 
\begin{enumerate*}[label=(\roman*)]
    \item Low $M$ performance is low in comparison to a single DIR (Fig.~\ref{sfig:R_M_comparison}). The difference is due to the fact that after swapping, the classical result is send via fiber and arrives to the receive node after some delay. During this time, participating qubits are still stored in quantum memory whose size $M$ may be too small to hold newly arrived qubits.
    \item For $M=45$, EPR is between $0.7$ to $0.8$. We need additional memory to approach $R_{\textrm{max}}$ as fiber delay is high and $M$ size is not sufficient to hold the qubits while swapping is executed. Moreover, outgoing quantum memory holds the qubits at the end nodes until EP is created. As BP are generated on every time step, new qubits replace oldest qubits in the memory. In the case that outgoing memory size is not large enough, qubits are dropped before swapping is completed. This is more prominent for 1:2 ratio fiber lengths.
\end{enumerate*}
Average Fidelity is better than $0.9$ for most quantum error cases where $[\Gamma_1, \Gamma_2]$ \thinspace{Hz} = $[10,000,100,000], [333,3,333]$ demonstrate Average Fidelity of sub $0.7$ and $0.4$, respectively. Fidelity is more sensitive to high memory errors (high $[\Gamma_1, \Gamma_2]$) than fiber errors (lower $\lambda_{\textrm{PF}}(25)$). As $M$ increases, for high memory errors we observe a slight decline of the Fidelity. Since qubits spend more time in memory, there is a higher probability that quantum errors occur. 

\noindent For the 1:2 fiber length, Fig.~\ref{sfig:line_1_2_EPR} and Fig.~\ref{sfig:line_1_2_Fidelity} show the EPR and average fidelity, respectively. Here, due to the large difference in fiber lengths, we need a significant increase in memory size ($M=650$), in order to achieve close to $R_{\textrm{max}}$ performance. In the repeater, qubits originated from the shorter fiber connection, wait longer in memory as qubits from the longer fiber connection arrive with low probability. Additionally, EPR varies between $[\Gamma_1,\Gamma_2]$ and $\lambda_{\textrm{PF}}(25)$ settings as $M$ increases. At this point, we are  sensitive to the size of the outgoing memory. Average Fidelity decreases as $M$ increases, similar to the common fiber length case, but in a larger scale due to the size of $M$ (more than 15 times bigger than common fiber length case).

\begin{figure}
    \centering
    \begin{subfigure}[b]{0.45\linewidth}        
        \centering
        \includegraphics[trim = 0mm 0mm 0mm 0mm, clip,scale=0.23]{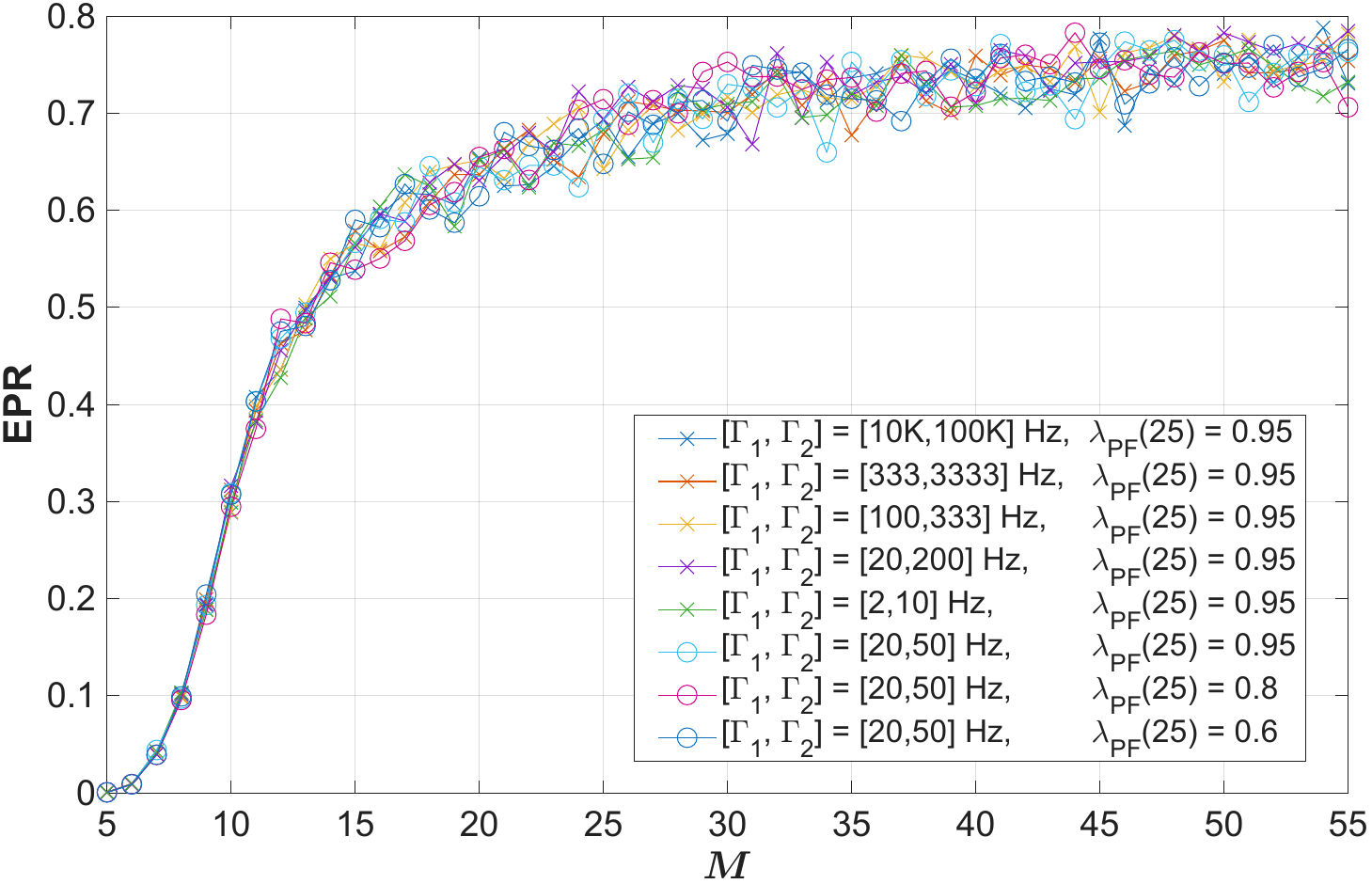} 
        \caption{EPR, $D_1=D_2$ = $4,630$\thinspace{m}.}
        \label{sfig:line_1_1_EPR}
    \end{subfigure}
    \begin{subfigure}[b]{0.45\linewidth}        
        \centering
        \includegraphics[trim = 0mm 0mm 0mm 0mm, clip,scale=0.23]{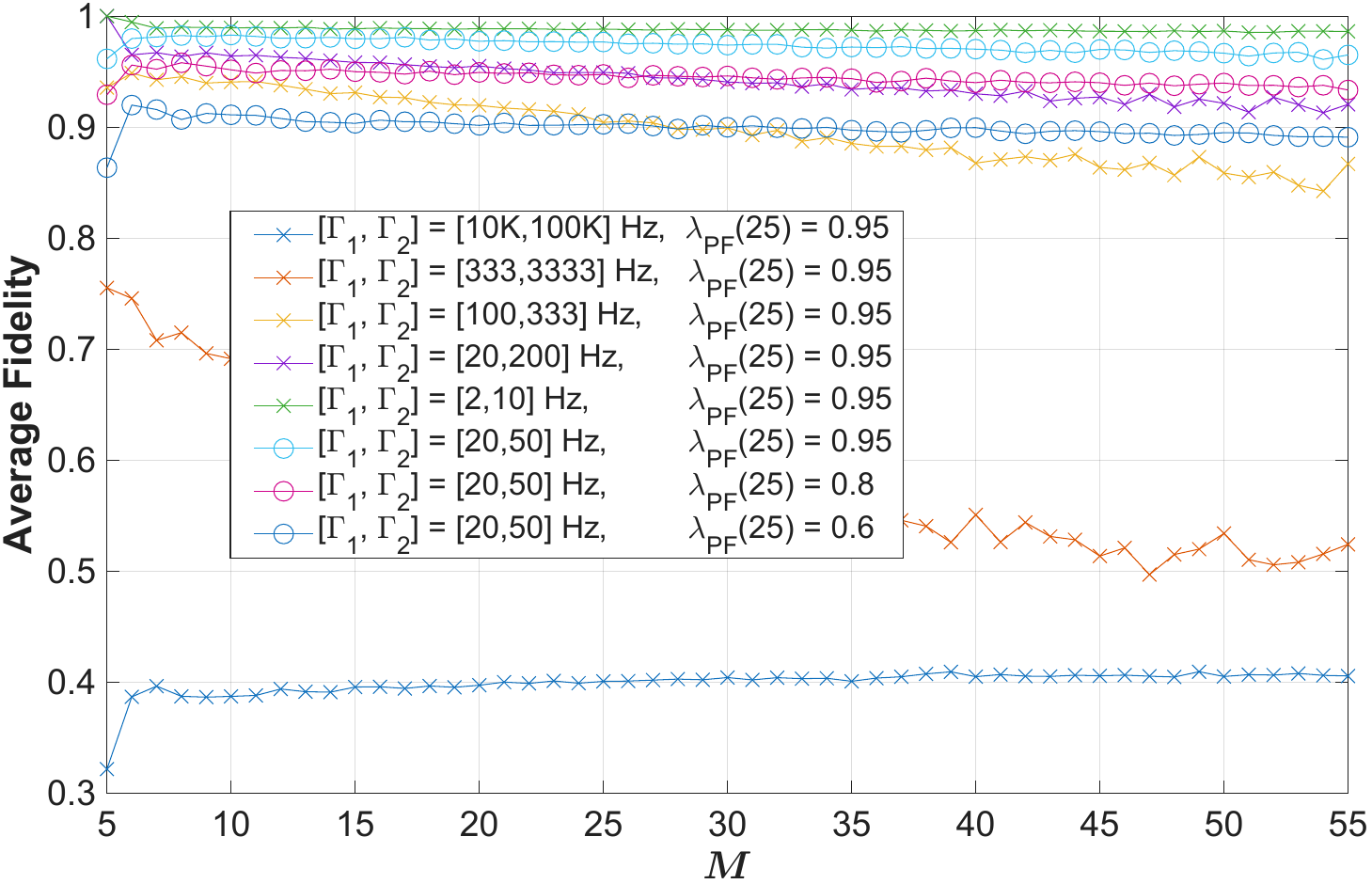} 
        \caption{Average Fidelity, $D_1=D_2$ = $4,630$\thinspace{m}.}
        \label{sfig:line_1_1_Fidelity}
    \end{subfigure}
    \medskip
    \vspace{1mm}
    \centering
    \begin{subfigure}[b]{0.45\linewidth}        
        \centering
        \includegraphics[trim = 0mm 0mm 0mm 0mm, clip,scale=0.23]{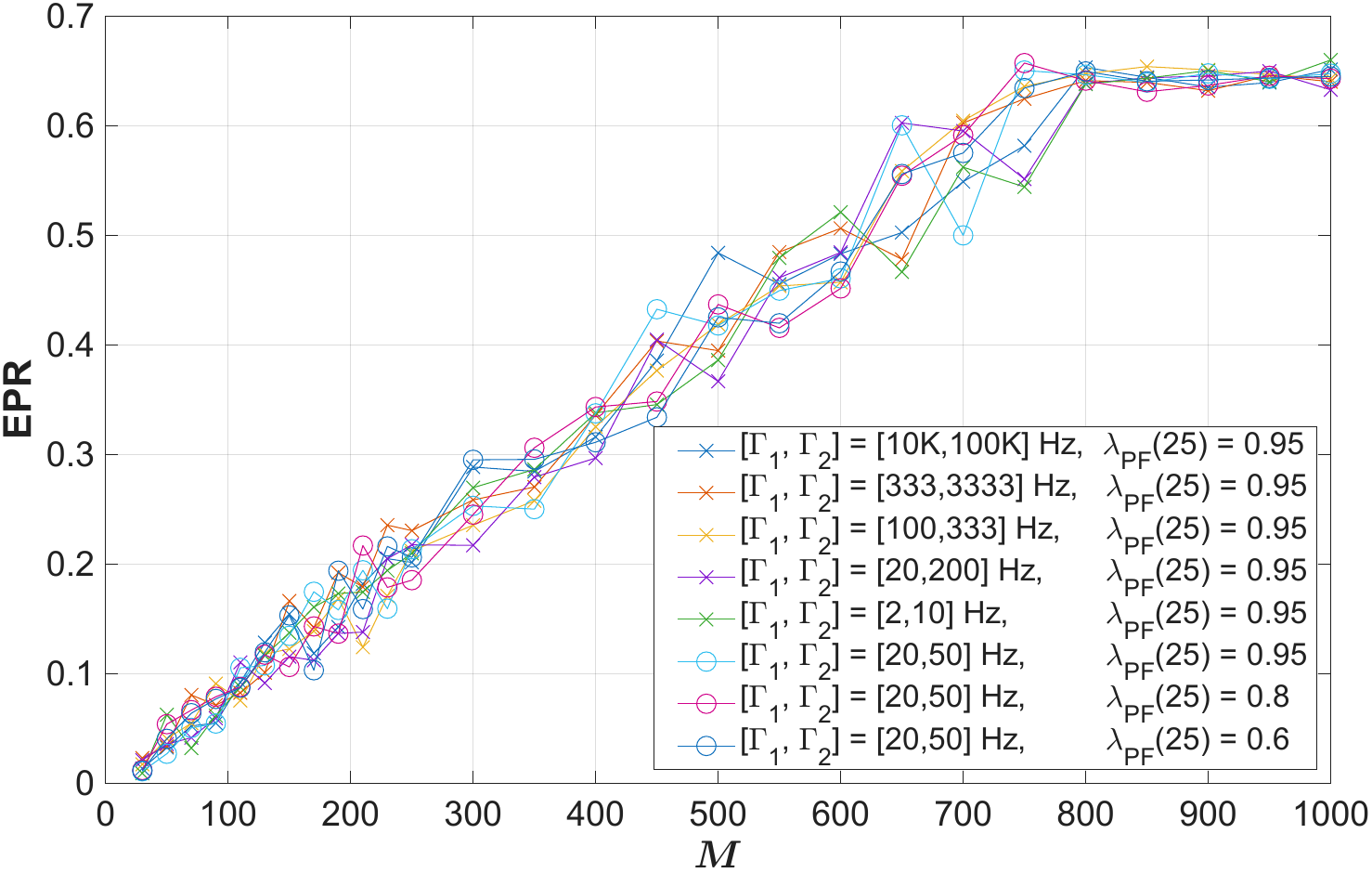}
        \caption{EPR, $D_1=4,630$\thinspace{m}, $D_2=9,260$\thinspace{m}.}
        \label{sfig:line_1_2_EPR}
    \end{subfigure}
    \begin{subfigure}[b]{0.45\linewidth}        
        \centering
        \includegraphics[trim = 0mm 0mm 0mm 0mm, clip,scale=0.23]{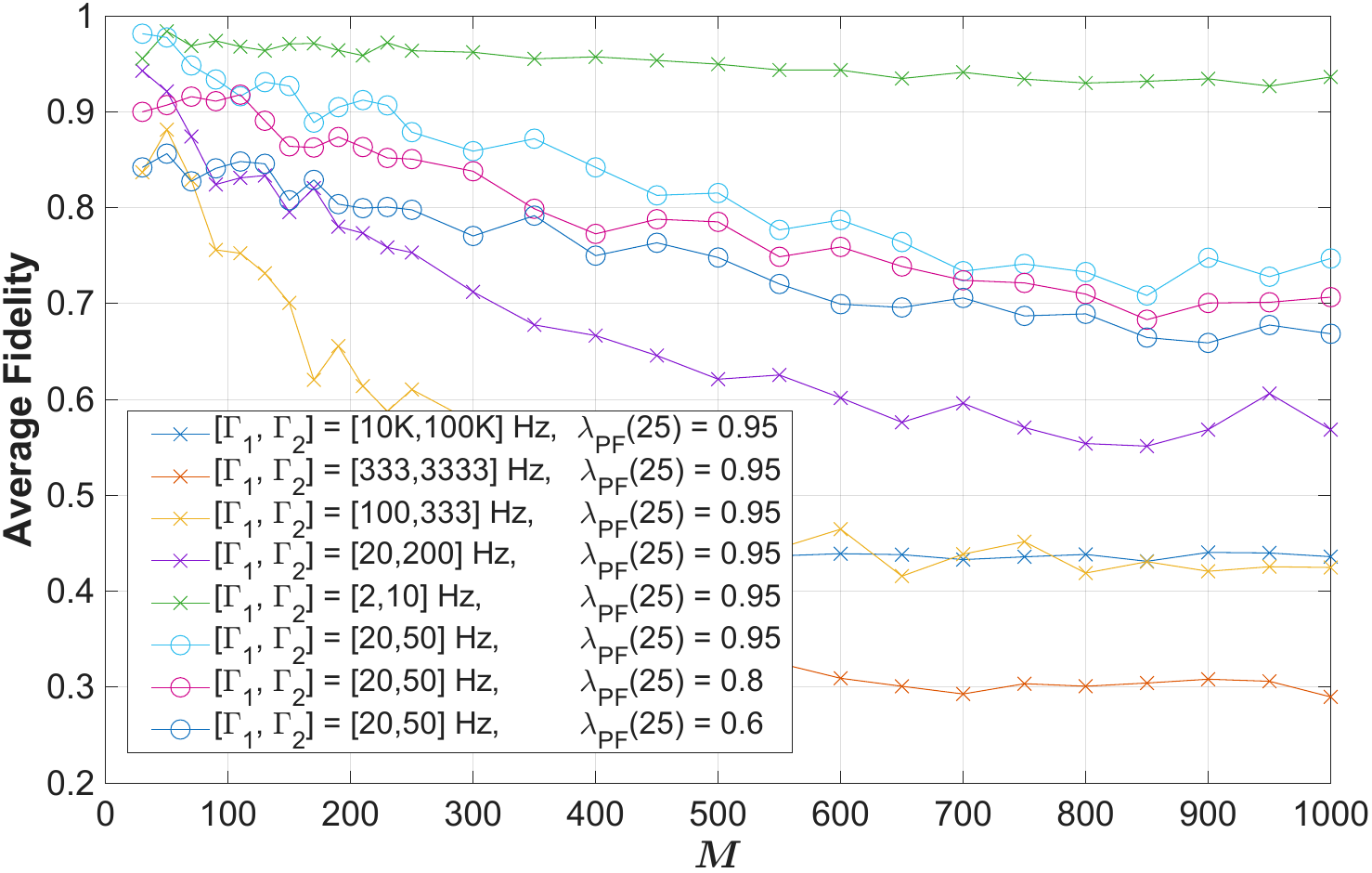}
        \caption{Average Fidelity, $D_1=4,630$\thinspace{m}, $D_2=9,260$\thinspace{m}.}
        \label{sfig:line_1_2_Fidelity}
    \end{subfigure}
    \caption{Line topology network with a single repeater: EPR and average fidelity as a function of $M$.}
    \label{fig:line_1R_EPR_Fidelity}
\end{figure}

We now consider the sub networks between BNL$\leftrightarrow$SBU$\leftrightarrow$WBY, WBY$\leftrightarrow$BKN$\leftrightarrow$CBU and SBU$\leftrightarrow$WBY$\leftrightarrow$BKN (as depicted in Fig.~\ref{sfig:SCY_line_3R_sketch_option_1}), where a repeater is located in SBU, BKN and WBY in each of the sub networks.

\noindent Fig.~\ref{sfig:BNL_SBU_WBY_EPR}, Fig.~\ref{sfig:WBY_BKN_CBU_EPR} and Fig.~\ref{sfig:SBU_WBY_BKN_EPR} show the EPR performance as a function of $M$ (between $5$ to $25$) for the selected $[\Gamma_1,\Gamma_2]$ and $\lambda_{\textrm{PF}}(25)$ quantum error settings. The calculated $R_{\textrm{max}}$ is $0.0376$, $0.057$ and $0.0376$ for Fig.~\ref{sfig:BNL_SBU_WBY_EPR}, Fig.~\ref{sfig:WBY_BKN_CBU_EPR} and Fig.~\ref{sfig:SBU_WBY_BKN_EPR}, respectively. 
As we saw before, for low $M$, EPR performance is low. However, when $M$ increases, performance is very close to $R_{\textrm{max}}$. For example, for WBY$\leftrightarrow$BKN$\leftrightarrow$CBU (Fig.~\ref{sfig:WBY_BKN_CBU_EPR}), we achieve EPR of sub $0.05$ comparing to $R_{\textrm{max}}=0.057$ (around $8\%$ difference). An interesting question arises: why can we get comparable performance for $M=25$ for these sub networks, where for the 1:2 ratio single repeater network we had to go to $M=650$ or so ? the reason is that here, $p_k$ is low enough, although the fiber lengths (delays) are high. Hence, the memory doesn't get fully utilized and stored qubits are not replaced frequently. Another important observation rises when we compare the required $M$ as a function of $\delta$ (for DIR, using given $p_k$ probabilities) for PR to the sub networks EPR performance. The former, predicts achieving $R_{\textrm{max}}$ with $M=22$ and $M=14$ for $\delta=0.002$. Clearly, swapping delays and sending the classical results back to the receiving node, increase the needed size of $M$.

\noindent Fig.~\ref{sfig:BNL_SBU_WBY_Fidelity}, Fig.~\ref{sfig:WBY_BKN_CBU_Fidelity} and Fig.~\ref{sfig:SBU_WBY_BKN_Fidelity} show the Average Fidelity as a function of $M$ (between $5$ to $25$) for the selected $[\Gamma_1,\Gamma_2]$ and $\lambda_{\textrm{PF}}(25)$ quantum error settings. As fiber lengths are in the order of tenth of \thinspace{km}, both the delay and the time duration that a qubit is stored in memory impact the Fidelity. For example, in sub network BNL$\leftrightarrow$SBU$\leftrightarrow$WBY (Fig.~\ref{sfig:BNL_SBU_WBY_Fidelity}), we see a degradation of $0.2$ in Fidelity between $[\Gamma_1, \Gamma_2]$ \thinspace{Hz} = $[2,10]$ and $[\Gamma_1, \Gamma_2]$ \thinspace{Hz} = $[100,333]$ for low $M$, and up to $0.45$ for high $M$. Similarly, when $\lambda_{\textrm{PF}}(25)=0.6$, we see a degradation of $0.2$ to $0.3$ in Fidelity comparing to $\lambda_{\textrm{PF}}(25)=0.95$. Note, that for settings of $[\Gamma_1, \Gamma_2]$ \thinspace{Hz} = $[10,000,100,000]$, the Fidelity is higher than $[\Gamma_1, \Gamma_2]$ \thinspace{Hz} = $[333,3,333]$, where one expect it to be lower. The reason for this behavior is that $\Gamma_1$ rate is high and due to the the long time the qubit is stored in memory, there is a high probability to lose population (move to $|0\rangle$ state). In this case, for BP of the form $\frac{1}{\sqrt{2}}\left(|00\rangle + |11\rangle\right)$, the Fidelity is $0.5$.

\begin{figure}
    \centering
    \begin{subfigure}[b]{0.45\linewidth}        
        \centering
        \includegraphics[trim = 0mm 0mm 0mm 0mm, clip,scale=0.23]{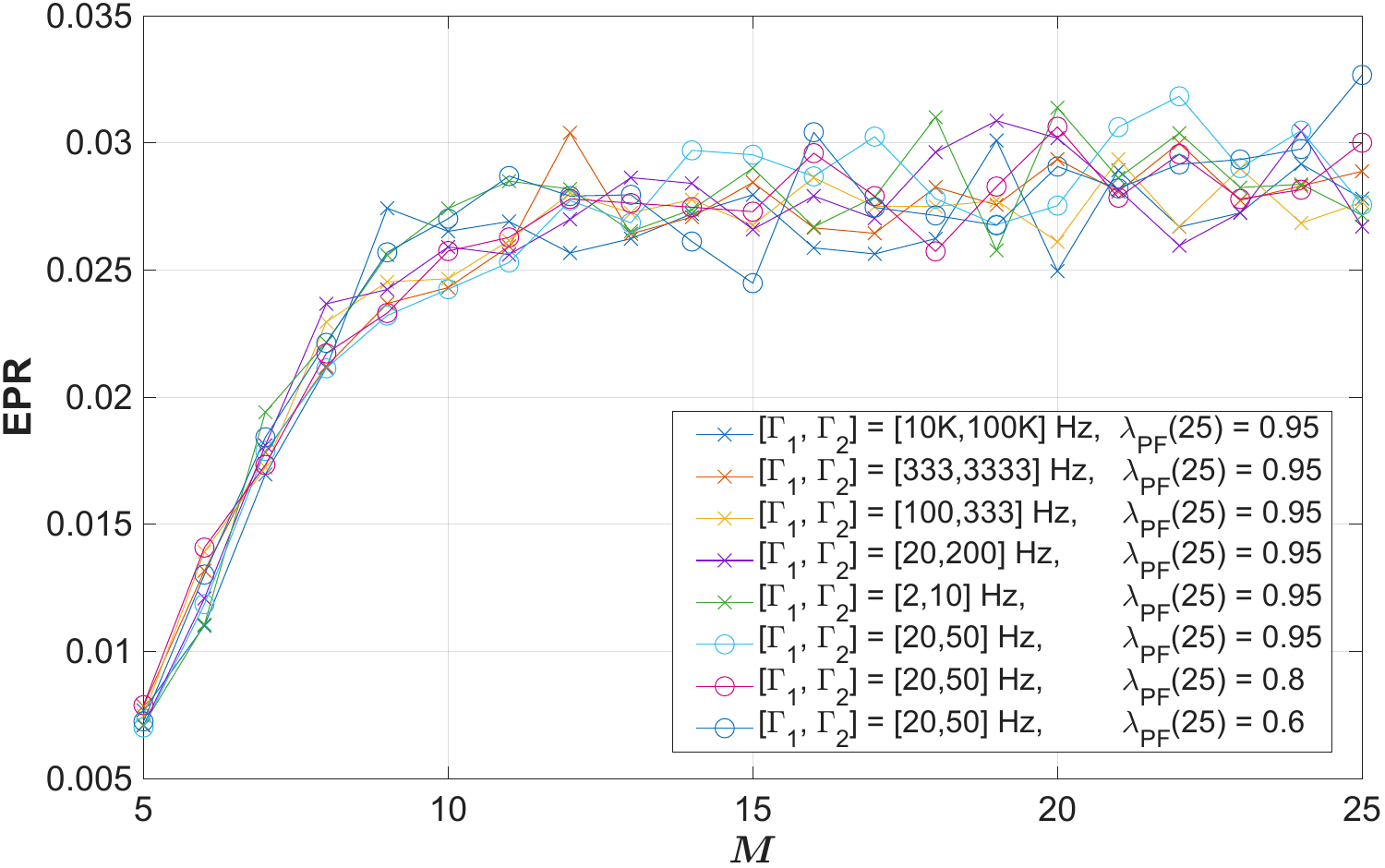} 
        \caption{BNL$\leftrightarrow$SBU$\leftrightarrow$WBY, EPR performance}
        \label{sfig:BNL_SBU_WBY_EPR}
    \end{subfigure}
    \begin{subfigure}[b]{0.45\linewidth}        
        \centering
        \includegraphics[trim = 0mm 0mm 0mm 0mm, clip,scale=0.23]{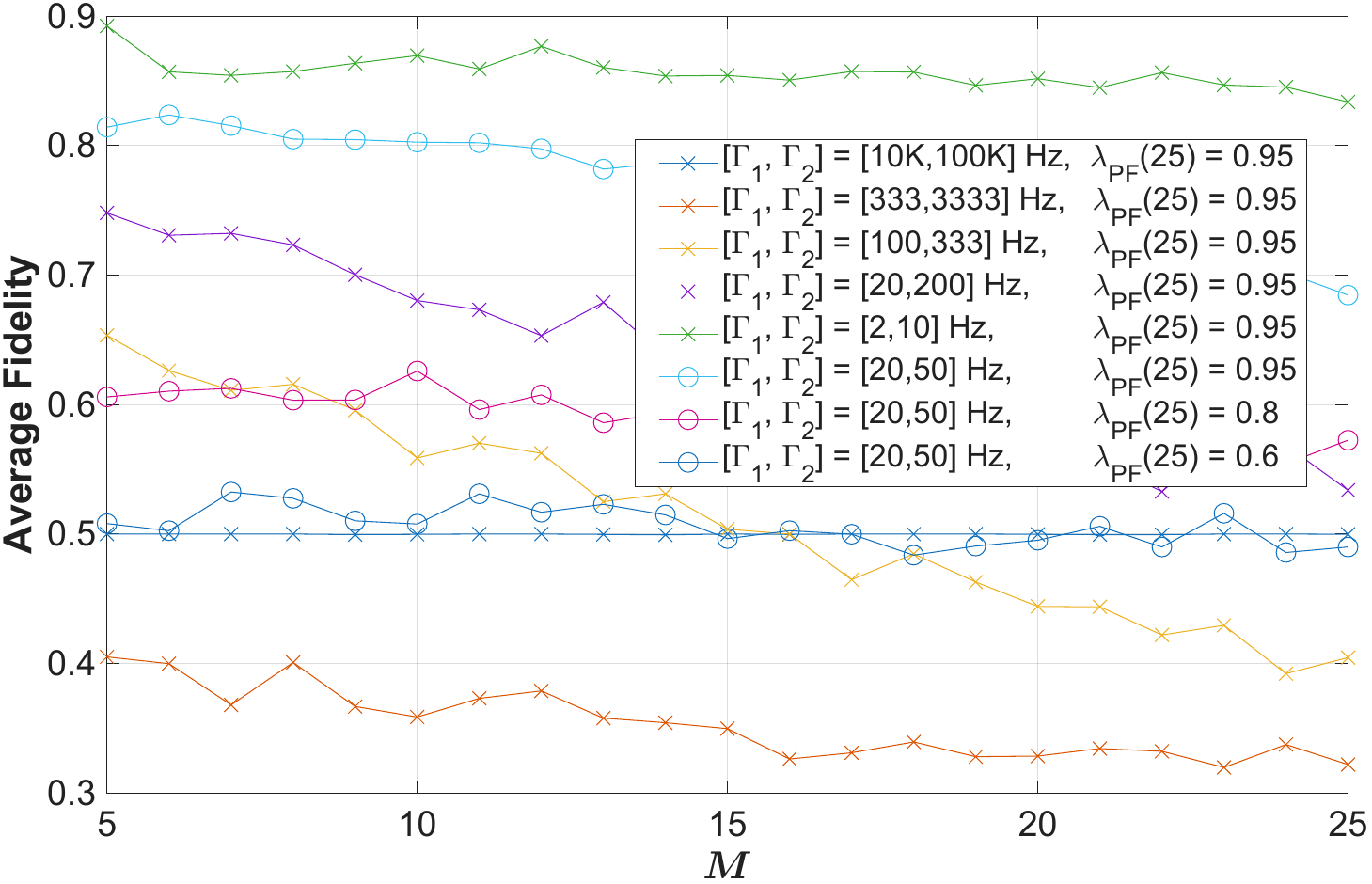} 
        \caption{BNL$\leftrightarrow$SBU$\leftrightarrow$WBY, Average Fidelity}
        \label{sfig:BNL_SBU_WBY_Fidelity}
    \end{subfigure}
    \medskip
    \vspace{1mm}
    \centering
    \begin{subfigure}[b]{0.45\linewidth}        
        \centering
        \includegraphics[trim = 0mm 0mm 0mm 0mm, clip,scale=0.23]{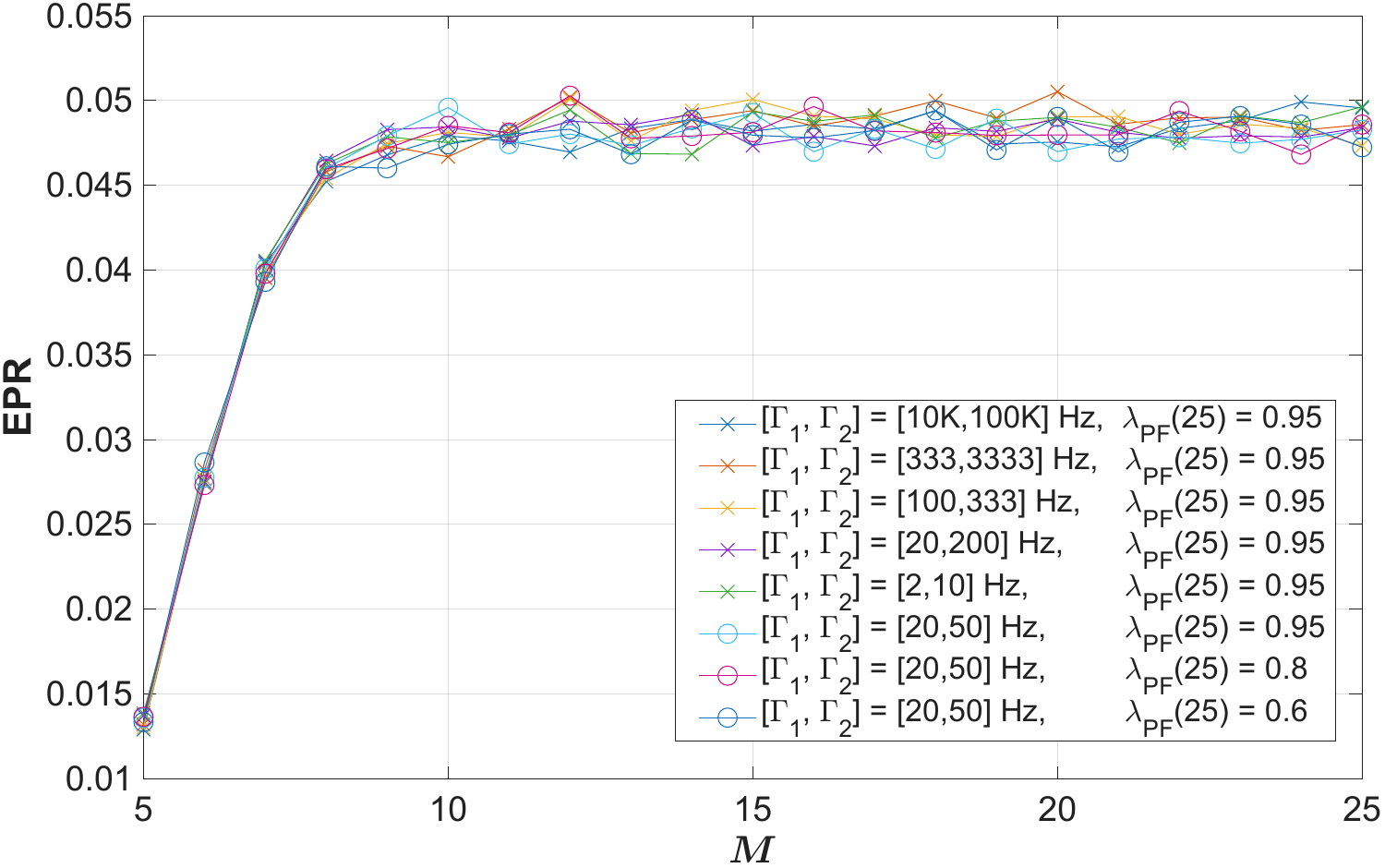}
        \caption{WBY$\leftrightarrow$BKN$\leftrightarrow$CBU, EPR performance}
        \label{sfig:WBY_BKN_CBU_EPR}
    \end{subfigure}
    \begin{subfigure}[b]{0.45\linewidth}        
        \centering
        \includegraphics[trim = 0mm 0mm 0mm 0mm, clip,scale=0.23]{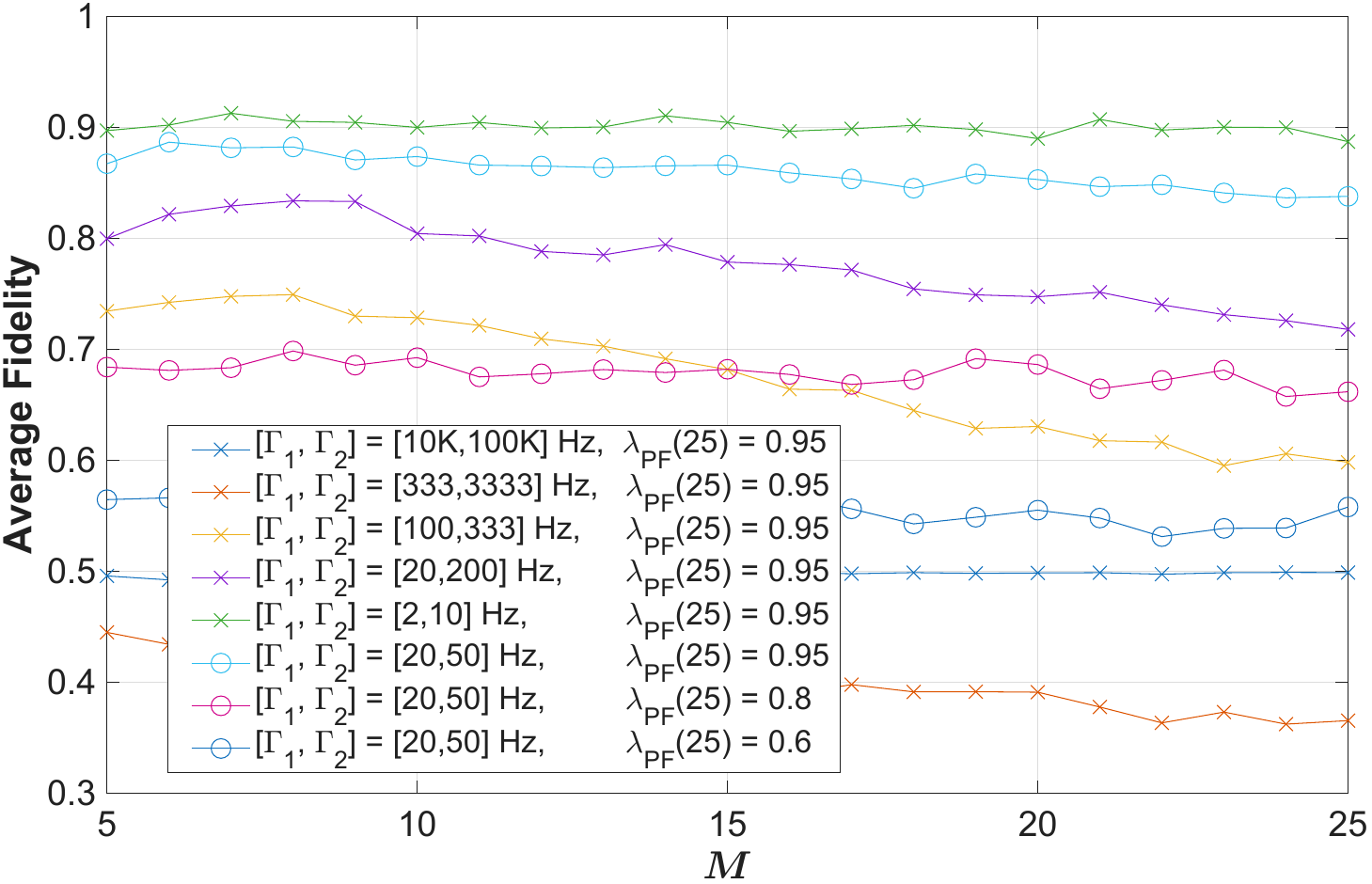}
        \caption{WBY$\leftrightarrow$BKN$\leftrightarrow$CBU, Average Fidelity}
        \label{sfig:WBY_BKN_CBU_Fidelity}
    \end{subfigure}
    \medskip
    \vspace{1mm}
    \centering
    \begin{subfigure}[b]{0.45\linewidth}        
        \centering
        \includegraphics[trim = 0mm 0mm 0mm 0mm, clip,scale=0.23]{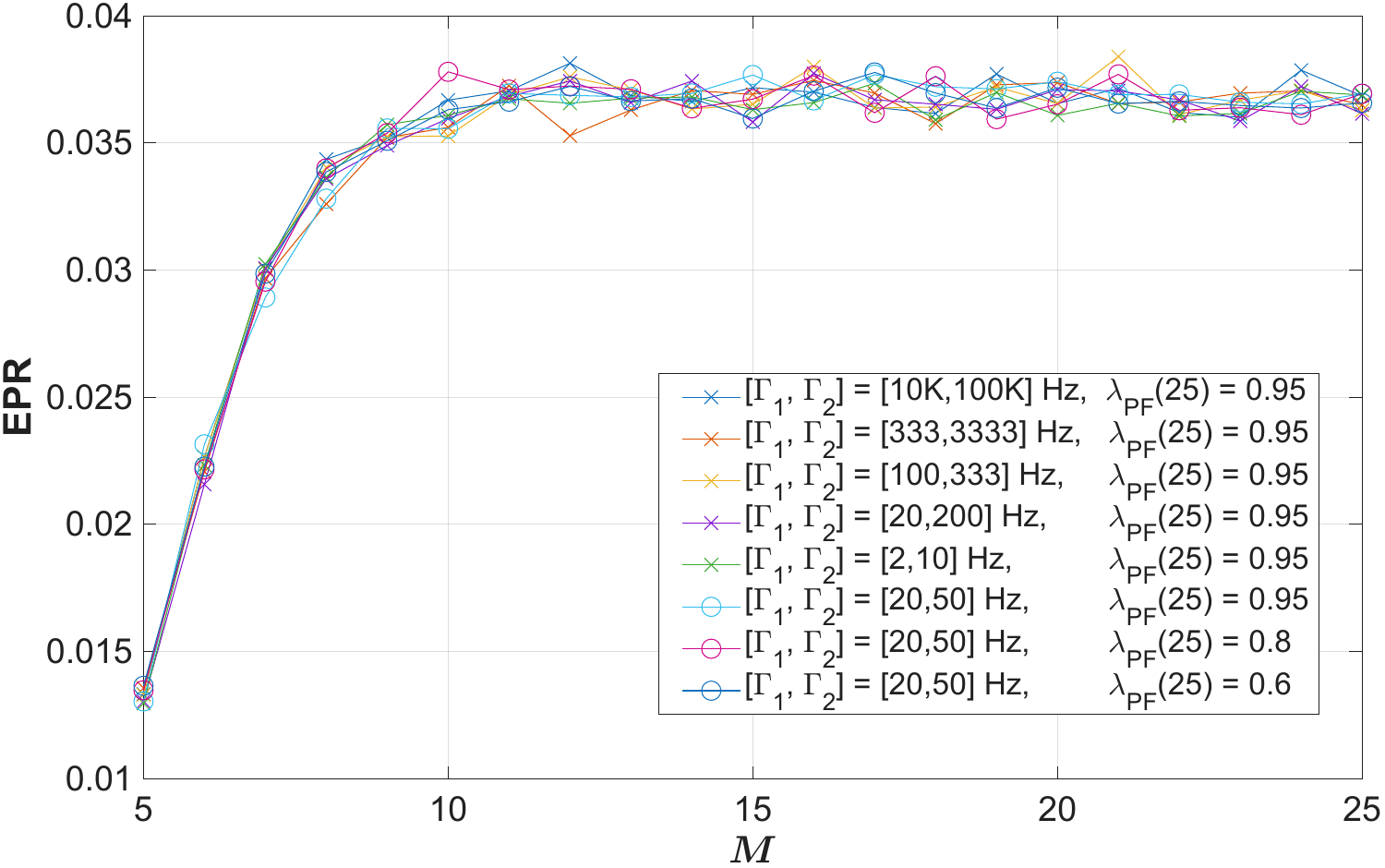}
        \caption{SBU$\leftrightarrow$WBY$\leftrightarrow$BKN, EPR performance}
        \label{sfig:SBU_WBY_BKN_EPR}
    \end{subfigure}
    \begin{subfigure}[b]{0.45\linewidth}        
        \centering
        \includegraphics[trim = 0mm 0mm 0mm 0mm, clip,scale=0.23]{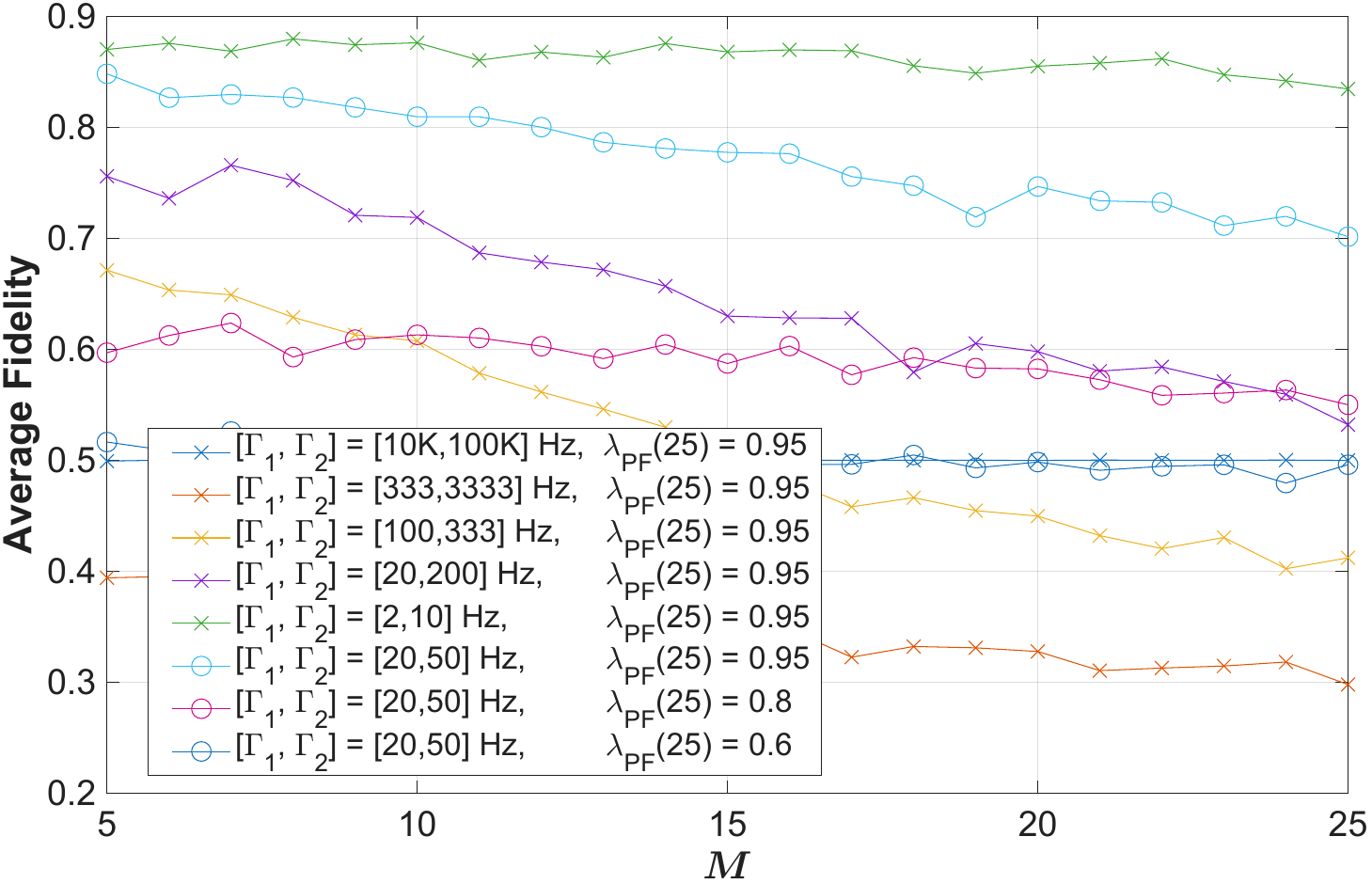}
        \caption{SBU$\leftrightarrow$WBY$\leftrightarrow$BKN, Average Fidelity}
        \label{sfig:SBU_WBY_BKN_Fidelity}
    \end{subfigure}
    \caption{EPR performance and average Fidelity as a function of $M$ for the selected $[\Gamma_1,\Gamma_2]$ and $\lambda_{\textrm{PF}}(25)$ quantum error settings (4 runs for each $M$, 20,000 iterations per run): SCY-QNet sub networks in Fig.~\ref{sfig:SCY_line_3R_sketch_option_1}: (\subref{sfig:BNL_SBU_WBY_EPR}) - (\subref{sfig:BNL_SBU_WBY_Fidelity}) BNL$\leftrightarrow$SBU$\leftrightarrow$WBY, (\subref{sfig:WBY_BKN_CBU_EPR}) - (\subref{sfig:WBY_BKN_CBU_Fidelity}) WBY$\leftrightarrow$BKN$\leftrightarrow$CBU, and (\subref{sfig:SBU_WBY_BKN_EPR}) - (\subref{sfig:SBU_WBY_BKN_Fidelity}) SBU$\leftrightarrow$WBY$\leftrightarrow$BKN.}
    \label{fig:SCY_1_sub_paths_EPR_Fidelity}
\end{figure}

\subsection{Performance of line topology network with three repeaters} \label{ssec:3R_network_performance}
First, let us consider common fiber length $4,630$\thinspace{m} ($P=0.8$) for a line network with 3 repeaters.
Fig.~\ref{fig:Line_3R_EPR_fidelity} shows EPR and average fidelity as a function of $M$ (between $50$ to $80$) for the selected $[\Gamma_1,\Gamma_2]$ and $\lambda_{\textrm{PF}}(25)$ quantum error settings. BPs are generated every $10$ \thinspace{$\mu s$}.
Recall, that for multi repeater network, the swapping classical results accumulate more latency, as several steps are needed to establish an EP between the end nodes.
Initial observation for EPR performance (Fig.~\ref{sfig:line_3R_EPR}) is that we need a large size memory to get closer to $R_{\textrm{max}}=0.8$. This is expected as network dynamics for multi repeater network requires that stored qubits will be available across the nodes, in a quantity that will enable high EPR. In other words, as $M$ getting larger, EPR can potentially increase to the point that overall losses and delays limit the EP creation. On the other hand, the minimum required $M$ is significantly higher ($M>35$, not shown) than a single repeater network, since the memory (incoming and outgoing) should have enough depth to store a qubit for the overall network delays. Below this minimum $M$, qubits are replaced (by new qubits) before a single EP can be created between the end nodes. 

\noindent Average Fidelity (Fig.~\ref{sfig:line_3R_fidelity}) is similar to a single repeater line network, where we observe significant number of losing population events (Fidelity of $0.5$) for high $\Gamma_1$.

\begin{figure}
    \centering
    \begin{subfigure}[b]{0.47\linewidth}        
        \centering
        \includegraphics[trim = 0mm 0mm 0mm 0mm, clip,scale=0.23]{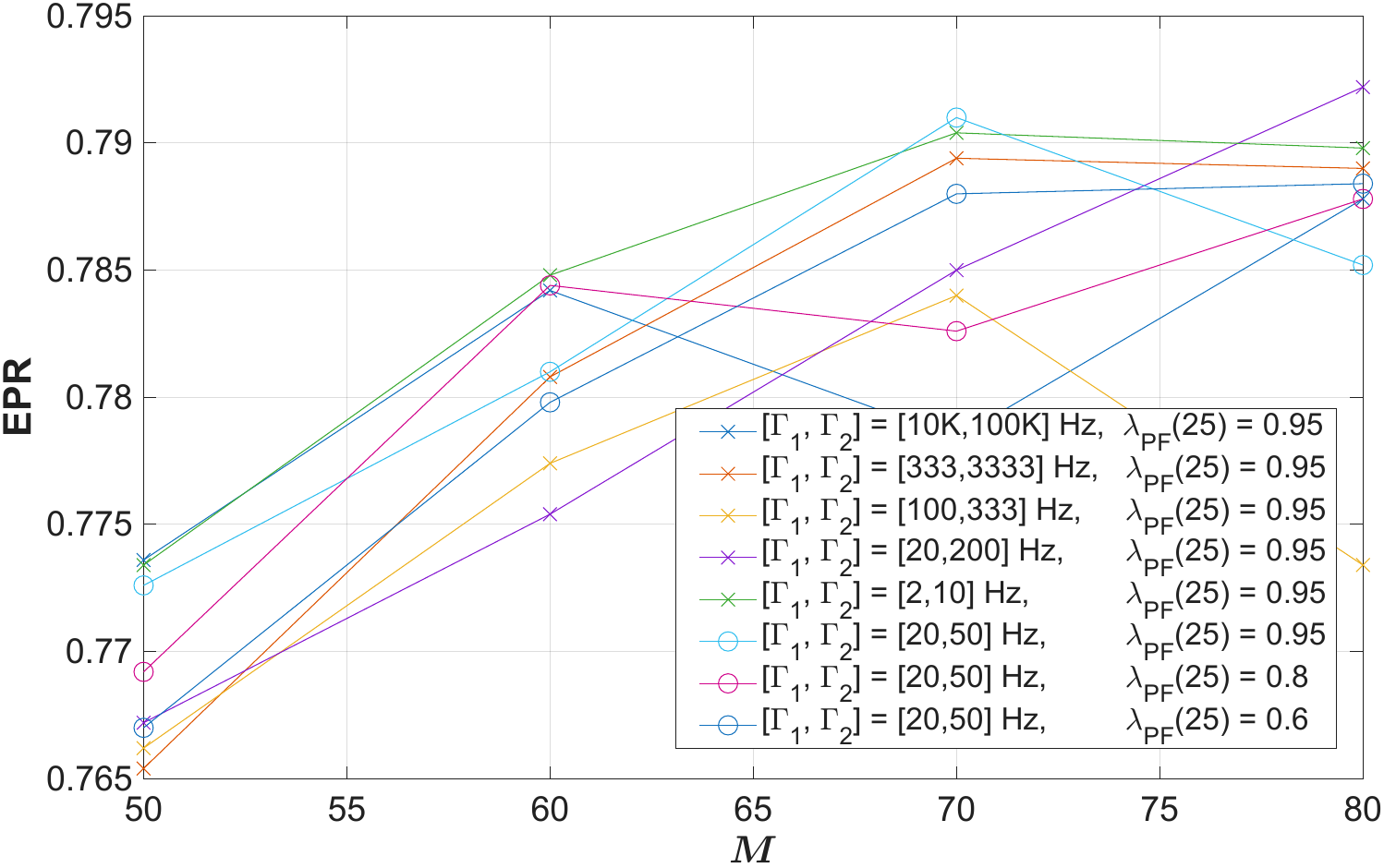}   
        \caption{}
        \label{sfig:line_3R_EPR}
    \end{subfigure}
    \begin{subfigure}[b]{0.47\linewidth}        
        \centering
        \includegraphics[trim = 0mm 0mm 0mm 0mm, clip,scale=0.23]{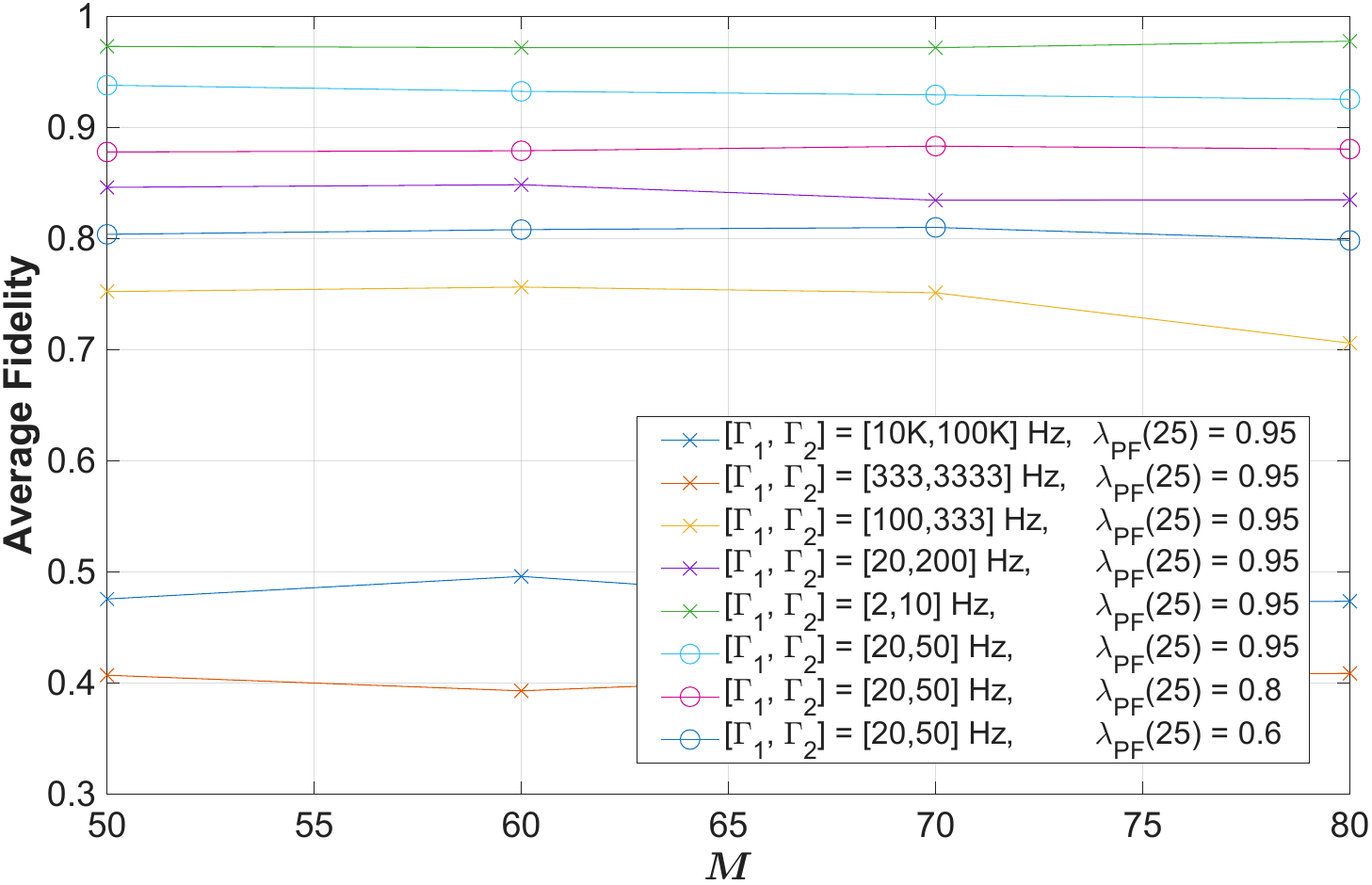}    
        \caption{}
        \label{sfig:line_3R_fidelity}
    \end{subfigure}

    \caption{Line topology with three repeaters, equal distances ($4,630$\thinspace{m}) between nodes. (\subref{sfig:line_3R_EPR}) EPR performance and (\subref{sfig:line_3R_fidelity}) Average Fidelity as a function of $M$ for the selected $[\Gamma_1,\Gamma_2]$ and $\lambda_{\textrm{PF}}(25)$ quantum error settings.}
    \label{fig:Line_3R_EPR_fidelity}
\end{figure}

Finally, let us consider the SCY-QNet network depicted in Fig.~\ref{fig:SCY_google_earth} and its two subgraphs that connect CBU to BNL (Fig.~\ref{sfig:SCY_line_3R_sketch_option_1} and Fig.~\ref{sfig:SCY_line_3R_sketch_option_2}). Fig.~\ref{fig:SCY_paths_EPR_Fidelity} shows EPR and average fidelity as a function of $M$ (between $10$ to $60$) for the selected $[\Gamma_1,\Gamma_2]$ and $\lambda_{\textrm{PF}}(25)$ quantum error settings. BPs are generated every $100$ \thinspace{$\mu s$}.

\noindent Fig.~\ref{sfig:SCY_path_option_1_EPR} and Fig.~\ref{sfig:SCY_path_option_2_EPR} show the EPR performance of path option 1 and path option 2, respectively. EPR performance is on par between the two path options. For example, for path option 1, EPR is between $0.011$ to $0.014$  for $M=30$, while for path option 2, EPR is between $0.01$ to $0.015$ for the same $M$. 
Here, in contrast to 3 repeater network with common fiber length, we see EPR for lower $M$ values. For instance, for path option 2, EPR is between $0.006$ to $0.008$  for $M=10$. As the minimum needed $M$ should have enough depth due to network latency (fiber lengths are long), the loss probability is high enough to actively reduce it.

\noindent Fig.~\ref{sfig:SCY_path_option_1_Fidelity} and Fig.~\ref{sfig:SCY_path_option_2_Fidelity} show the average Fidelity of path option 1 and path option 2, respectively. We observe a significant degradation in Fidelity when either $\lambda_{\textrm{PF}}(25)$ decreases or $[\Gamma_1,\Gamma_2]$ increase. For example, for path option 1, Fidelity drops from $0.6-0.7$ ($[\Gamma_1, \Gamma_2]$ \thinspace{Hz} = $[2,10]$) to $0.3-0.35$ ($[\Gamma_1, \Gamma_2]$ \thinspace{Hz} = $[100,333]$) for a fixed $\lambda_{\textrm{PF}}(25)=0.95$ value. Similarly, Fidelity drops from $0.35-0.45$ ($\lambda_{\textrm{PF}}(25)=0.95$) to $0.3-0.35$ ($\lambda_{\textrm{PF}}(25)=0.8$) for a fixed $[\Gamma_1, \Gamma_2]$ \thinspace{Hz} = $[20,50]$ values. Additionally, the probability to lose population becomes more evident as from $[\Gamma_1, \Gamma_2]$ \thinspace{Hz} = $[20,200]$ there are more events of Fidelity $0.5$.

\begin{figure}
    \centering
    \begin{subfigure}[b]{0.45\linewidth}        
        \centering
        \includegraphics[trim = 0mm 0mm 0mm 0mm, clip,scale=0.23]{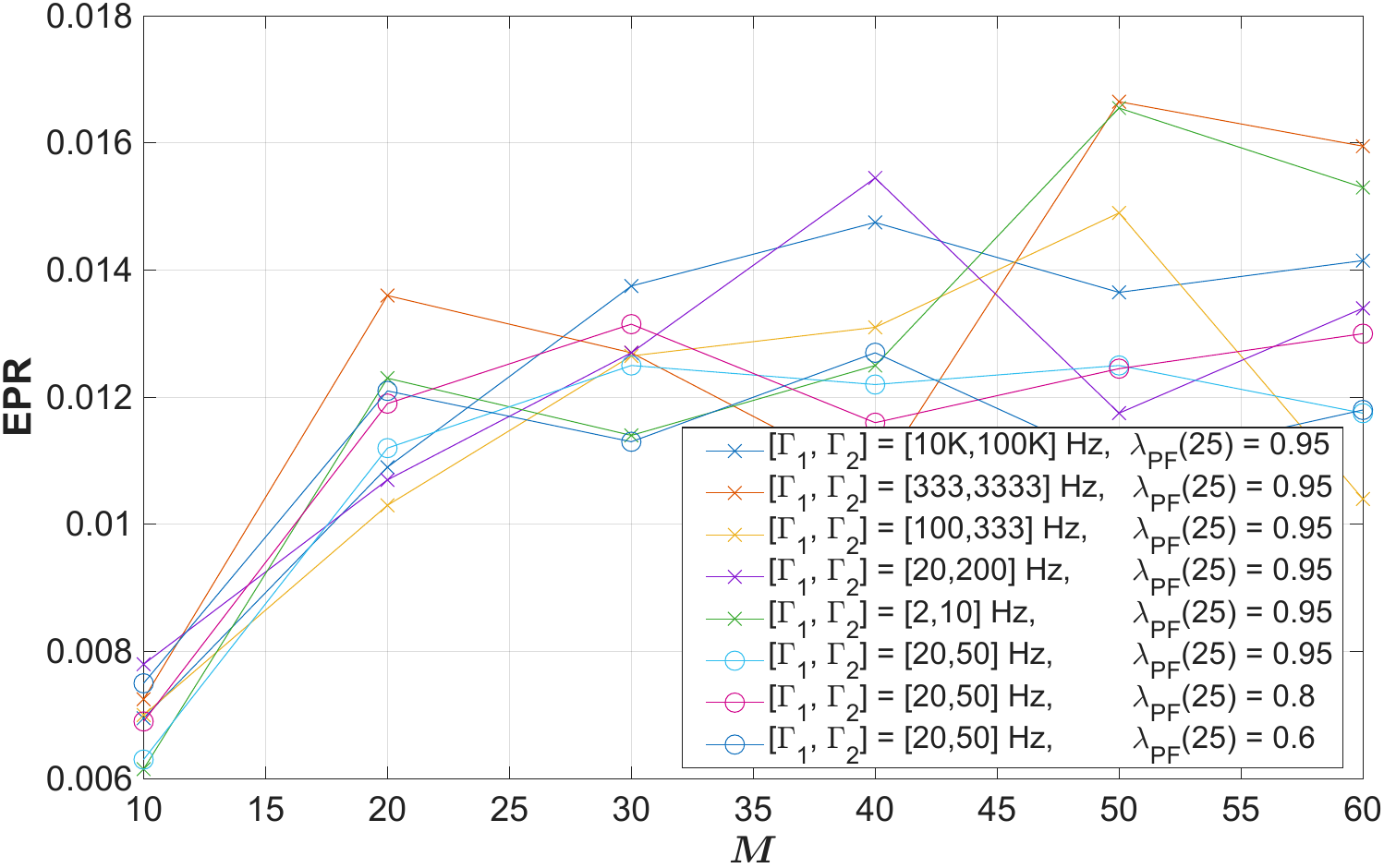} 
        \caption{SCY path option 1, EPR performance}
        \label{sfig:SCY_path_option_1_EPR}
    \end{subfigure}
    \begin{subfigure}[b]{0.45\linewidth}        
        \centering
        \includegraphics[trim = 0mm 0mm 0mm 0mm, clip,scale=0.23]{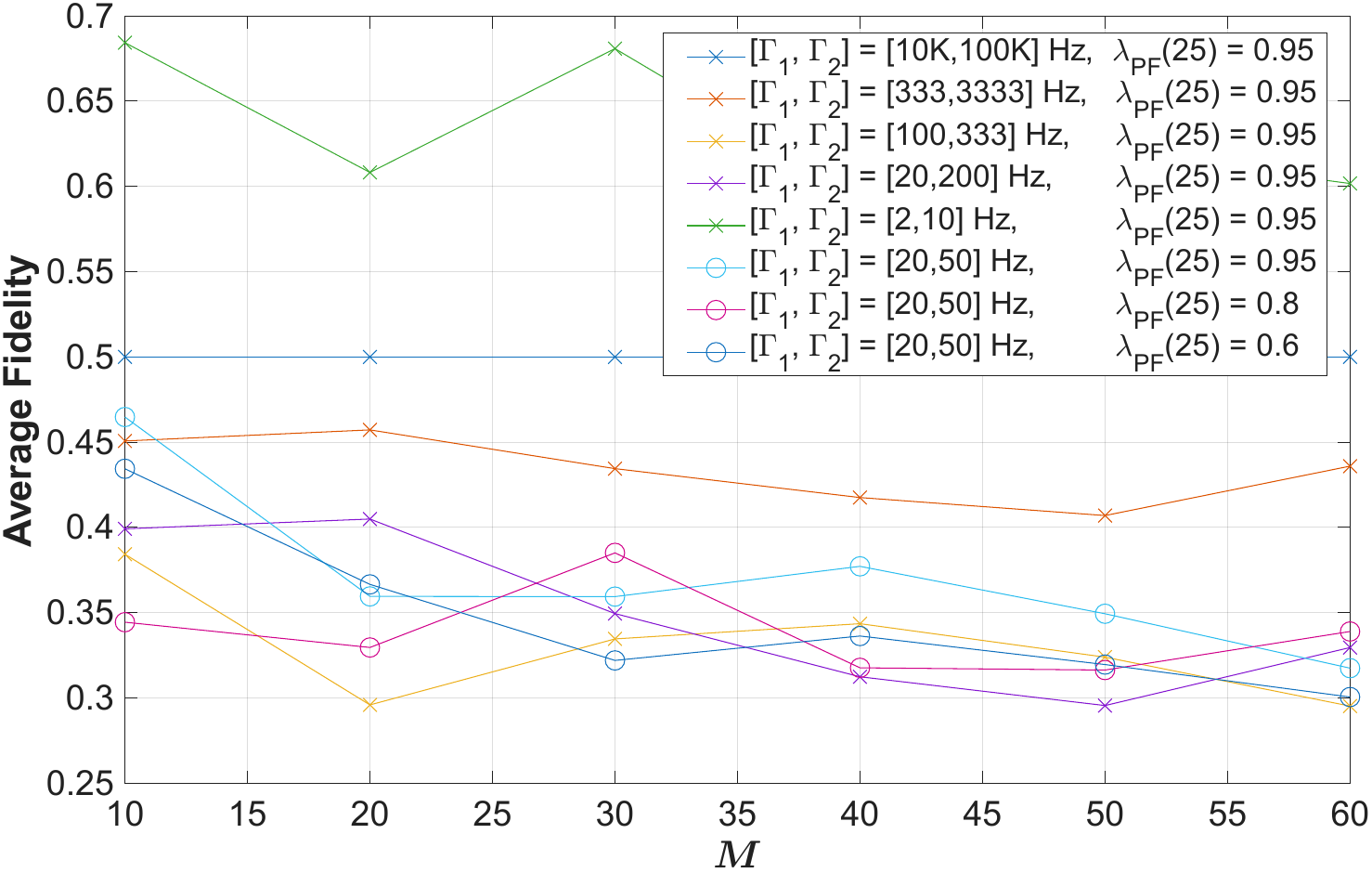} 
        \caption{SCY path option 1, Average Fidelity}
        \label{sfig:SCY_path_option_1_Fidelity}
    \end{subfigure}
    \medskip
    \centering
    \begin{subfigure}[b]{0.45\linewidth}        
        \centering
        \includegraphics[trim = 0mm 0mm 0mm 0mm, clip,scale=0.23]{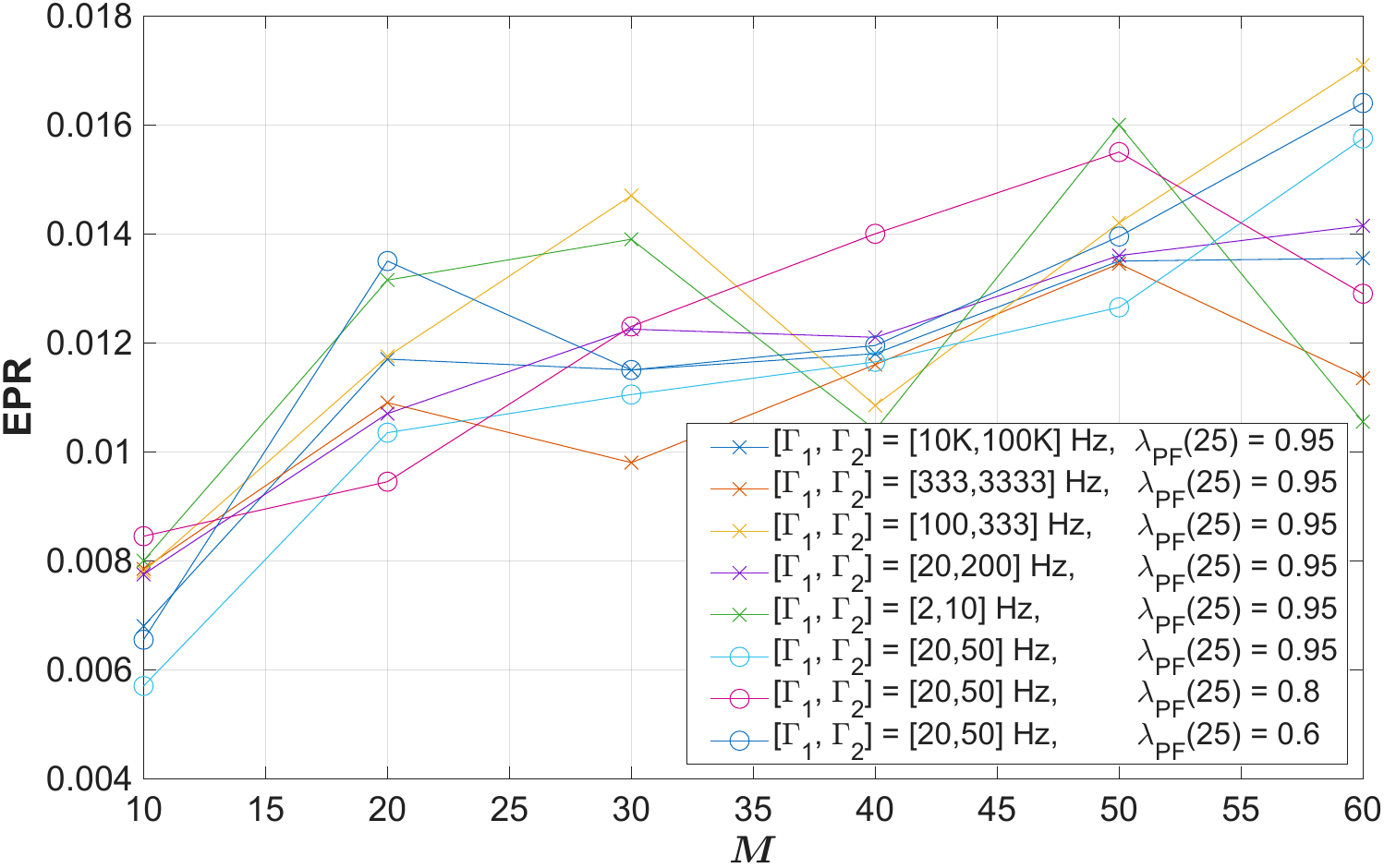}
        \caption{SCY path option 2, EPR performance}
        \label{sfig:SCY_path_option_2_EPR}
    \end{subfigure}
    \begin{subfigure}[b]{0.45\linewidth}        
        \centering
        \includegraphics[trim = 0mm 0mm 0mm 0mm, clip,scale=0.23]{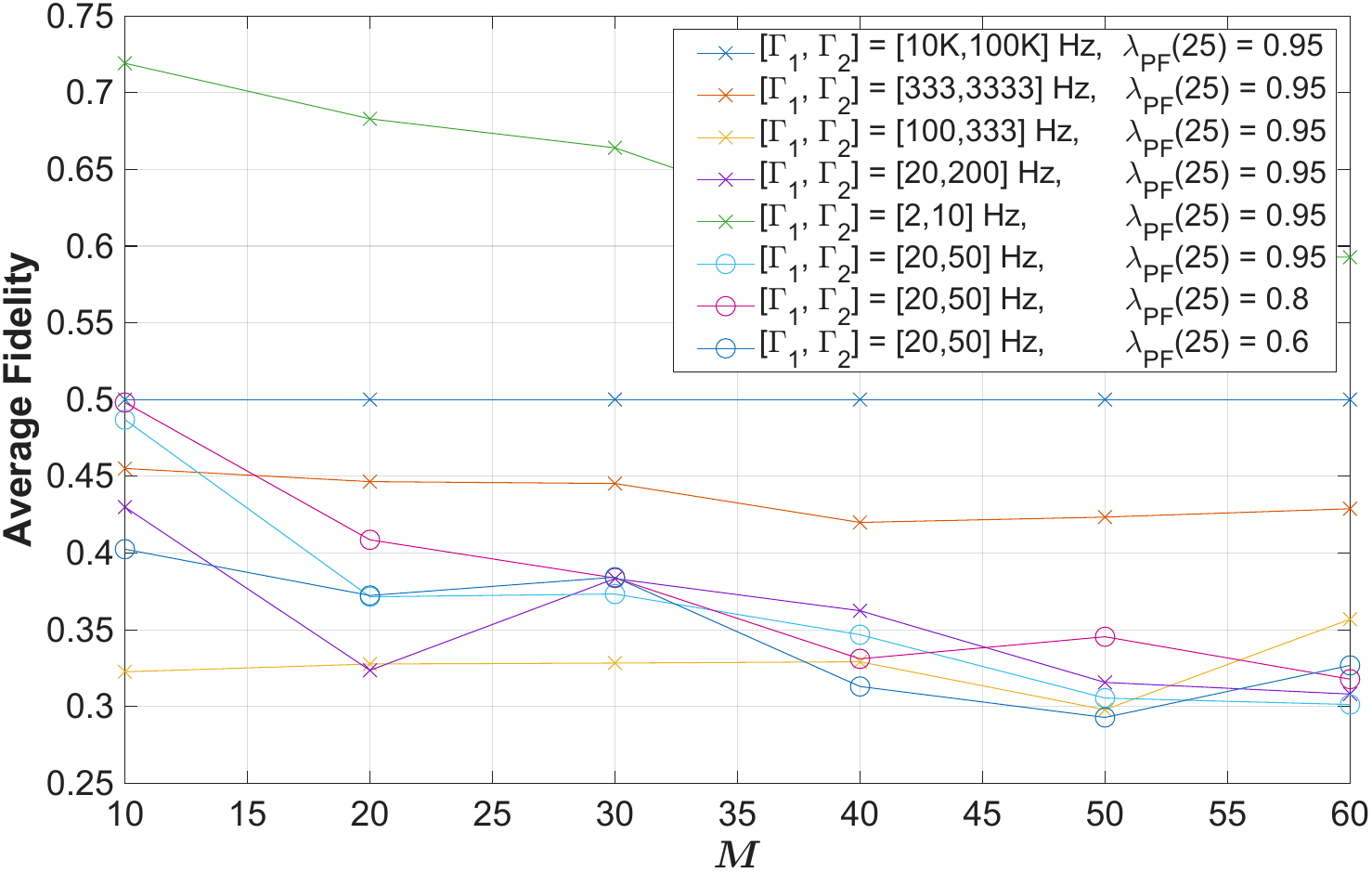}
        \caption{SCY path option 2, Average Fidelity}
        \label{sfig:SCY_path_option_2_Fidelity}
    \end{subfigure}
    \caption{EPR and average fidelity as a function of $M$ (between $10$ to $60$) for the selected $[\Gamma_1,\Gamma_2]$ and $\lambda_{\textrm{PF}}(25)$ quantum error settings: SCY-QNet network. (\subref{sfig:SCY_path_option_1_EPR})-(\subref{sfig:SCY_path_option_1_Fidelity}) Path option 1 (Fig.~\ref{sfig:SCY_line_3R_sketch_option_1}) and (\subref{sfig:SCY_path_option_2_EPR})-(\subref{sfig:SCY_path_option_2_Fidelity}) Path option 2 (Fig.~\ref{sfig:SCY_line_3R_sketch_option_2})}
    \label{fig:SCY_paths_EPR_Fidelity}
\end{figure}

%
%
%
%
%
%
%
%

\section{Conclusion and Future Work}\label{sec:summary}
We have discussed emerging near term practical Quantum networks: new deployment of multi-lab research oriented networks that take advantage of classical fiber infrastructure. They can (mostly) be reduced to a line topology with no loops. The most common quantum repeater used in these networks is the DIR. Prior work showed that its memory utilization cannot be bounded and analysis was only given for three inputs (or more) quantum repeater.

A proprietary quantum network simulation was developed from the ground up to evaluate network performance. First, single repeater performance was compared to prior work with complete agreement. Then, we introduced a novel approach for analyzing the performance of the DIR - showing analytically the performance tradeoff given its memory size, for two models: on-the-fly and full-memory-usage pairing models. Simulation results match the analytical analysis prediction of the desired degraded performance (defined by $\delta$) for the associated minimum memory size $M$. For example, size $M = 3$ memory leads to only $3\%$ degradation in performance for input arrival probabilities of $0.2$ and $0.8$. Hence, we  concluded that although memory utilization cannot be bounded, small memory size selection can provide as close performance to the maximum EPR as desired.

We then shifted to the EPR performance and average Fidelity of quantum networks, that were modeled as line topology with no loops. A simple heuristic algorithm for repeaters swapping policy was introduced followed by performance evaluation of a single synthetic repeater line network (common fiber length and 1:2 ratio fiber lengths) and sub networks of SCY-QNet. Finally, performance evaluation of 3 common length repeater network and SCY-QNet two path options was conducted.
We concluded with several observations.
\textit{For single repeater networks:}
\begin{enumerate*}[label=(\roman*)]
    \item EPR performance for low $M$ is low due to network delay. High $M$ EPR performance is sensitive to outgoing memory size. When the difference between fiber lengths is high (e.g., 1:2 ratio), there is a need for a higher $M$ to approach $R_{\textrm{max}}$, and 
    \item average Fidelity is better than $0.9$ for most quantum error cases and it is more sensitive to memory errors than to fiber errors.
\end{enumerate*}
\textit{for multi-repeater networks:}
\begin{enumerate*}[label=(\roman*)]
    \item the minimum $M$ required to enable EP creation depends on overall network latency and losses (both act opposed to each other), 
    \item as fiber lengths increase, Fidelity degrades for higher $[\Gamma_1,\Gamma_2]$ or lower $\lambda_{\textrm{PF}}(25)$ values due to delay and the increased period a qubit is stored in memory, and
    \item Fidelity produces $0.5$ values as $\Gamma_1,\Gamma_2$ increase along with long fiber lengths.
\end{enumerate*}

Future research will further consider queuing optimization for incoming and outgoing memory qubits, alternative memory storage policies and evaluate the handling of timing offsets and their impact on network optimization. Using our available simulation platform, detailed physical models of memory modules, classical and quantum multiplexing approach could be explored to better understand the advantage (and tradeoff) of common fiber connections.
As we observed, only limited alternative paths exist between specific end nodes, and therefore, simplified, sub-optimal approach for two path network will also be considered. Long term future work will introduce additional topologies on top of a line with practical policies to achieve maximum EPR and high Fidelity.

%
%
%
%
%
%
%
%


%
%
%
%
%
%
%
%
%

\newpage
\bibliographystyle{elsarticle-num}
\bibliography{Bibliography}


%
%
%
%

%

\end{document}